\documentclass[a4paper,UKenglish]{lipics-v2019}
\usepackage{mathpartir} %
\usepackage{amsmath}
\usepackage{amssymb}
\usepackage{amsfonts}
\usepackage{bm}         %
\usepackage{stmaryrd}   %
\usepackage{xspace}     %
\usepackage{verbatim}
\usepackage[numbers,sort&compress]{natbib}
\usepackage{changepage}
\usepackage{listings}
\usepackage{graphicx}
\usepackage[T1]{fontenc}
\usepackage[scaled=0.7]{beramono}
\usepackage{array}
\newcolumntype{x}[1]{>{\centering\arraybackslash\hspace{0pt}}p{#1}}
\usepackage{makecell}
\usepackage{float}
\floatstyle{boxed}
\restylefloat{figure}
\usepackage{url}
\usepackage{multicol}
\usepackage{microtype}
\usepackage[page,header]{appendix}
\usepackage{titletoc}
\usepackage{csquotes}
\usepackage{subcaption}
\usepackage{pgf}
\usepackage{tikz}
\usetikzlibrary{arrows,automata,shapes,shapes.geometric}
\newenvironment{syntax}%
{\footnotesize\[\begin{array}{@{}lr@{\:\:\:}c@{\:\:\:}l}\ignorespaces}
{\end{array}\]\ignorespacesafterend}

\newenvironment{proofcase}[1]
  {\totheleft{\textbf{Case } \textsc{#1}}}
  {}

\newcommand{\calcwd}[1]{\textbf{\textsf{#1}}}
\newcommand{\one}{\ensuremath{\mathbf{1}}}
\newcommand{\mvu}{\ensuremath{\lambda_{\mathsf{MVU}}}\xspace}

\newcommand{\app}{\:}
\newcommand{\oftype}{\: {:} \:}
\newcommand{\letin}[3]{\calcwd{let} \: #1 = #2 \: \calcwd{in} \: #3}
\newcommand{\letintwo}[2]{\calcwd{let} \: #1 = #2 \: \calcwd{in}}

\newcommand{\caseof}[2]{\calcwd{case} \: #1 \: \{ #2 \}}
\newcommand{\caseofone}[1]{\calcwd{case} \: #1 \: }
\newcommand{\midspace}{\: \mid \:}
\newcommand{\mkwd}[1]{\ensuremath{\mathsf{#1}}}
\newcommand{\antiquote}[1]{\mathbf{\{} #1 \mathbf{\}}}
\newcommand{\htmlty}[1]{\mkwd{Html}(#1)}
\newcommand{\attrty}[1]{\mkwd{Attr}(#1)}
\newcommand{\tagname}[1]{\textsf{#1}}
\newcommand{\smalllt}{\scalebox{0.85}{<}}
\newcommand{\smallgt}{\scalebox{0.85}{>}}
\newcommand{\htmltag}[3]{\smalllt \tagname{#1} \: #2 \smallgt #3 \smalllt / \tagname{#1} \smallgt}
\newcommand{\htmltagzero}[2]{\smalllt \tagname{#1} \smallgt #2 \smalllt / \tagname{#1} \smallgt}
\newcommand{\htmltagqueue}[3]{\smalllt \tagname{#1} \qsep #2 \smallgt #3 \smalllt / \tagname{#1} \smallgt}
\newcommand{\tagzero}[1]{\smalllt \tagname{#1}\smallgt}
\newcommand{\tagzeroend}[1]{\smalllt / \tagname{#1}\smallgt}
\newcommand{\opentag}[1]{\smalllt \tagname{#1}}
\newcommand{\closetag}{\smallgt}

\newcommand{\config}[1]{\mathcal{#1}}
\newcommand{\htmlterm}[1]{\calcwd{html} \app #1}
\newcommand{\attrterm}[1]{\calcwd{attribute} \app #1}
\newcommand{\seq}[1]{\overrightarrow{#1}}
\newcommand{\inl}[1]{\calcwd{inl} \: #1}
\newcommand{\inr}[1]{\calcwd{inr} \: #1}
\newcommand{\evtty}[1]{\ensuremath{\mkwd{ty}(#1)}}
\newcommand{\evt}[1]{\ensuremath{\mkwd{#1}}}

\newcommand{\ev}{\evt{ev}\xspace}
\newcommand{\evtpayload}[2]{\evt{#1}(#2)}
\newcommand{\clickevt}{\ensuremath{\attribute{click}}}
\newcommand{\inputevt}{\ensuremath{\attribute{input}}}
\newcommand{\keydownevt}{\ensuremath{\attribute{keyDown}}}
\newcommand{\keyupevt}{\ensuremath{\attribute{keyUp}}\xspace}

\newcommand{\vh}{D}
\newcommand{\pgctx}[1]{\config{D}[#1]}
\newcommand{\mh}{M\xspace}

\newcommand{\state}[2]{(#1, #2)}
\newcommand{\statecomb}[3]{(#1, #2, #3)}
\newcommand{\statesub}[3]{(#1, #2, #3)}
\newcommand{\handlerproc}[2]{\langle #1 \mid #2 \rangle}
\newcommand{\handlerprocexp}[3]{\langle #1 \mid \state{#2}{#3} \rangle}
\newcommand{\handlerprocsub}[4]{\langle #1 \mid \statesub{#2}{#3}{#4} \rangle}
\newcommand{\idle}[1]{\calcwd{idle} \: {#1}}
\newcommand{\sys}[2]{#1 \fatsemi\: #2}
\newcommand{\syssub}[4]{#1 \fatsemi\: #2 \fatsemi\: #3 \fatsemi\: #4}

\newcommand{\evalarrow}{\longrightarrow}
\newcommand{\teval}{\evalarrow_{\textsf{M}}}
\newcommand{\ceval}{\evalarrow}
\newcommand{\cevalminus}{\evalarrow_{\textsf{E}}}
\newcommand{\totheleft}[1]{\begin{flushleft}#1\end{flushleft}}
\newcommand{\evthandler}[1]{\ensuremath{\attribute{#1}}}

\newcommand{\handler}[1]{\mkwd{handler}(#1)}

\newcommand{\handle}[1]{\mkwd{handle}(#1)}

\newcommand{\defeq}{\triangleq}
\newcommand{\vdashs}{\vdash}
\newcommand{\ttrue}{\mkwd{true}}
\newcommand{\ffalse}{\mkwd{false}}
\newcommand{\intty}{\mkwd{Int}}
\newcommand{\boolty}{\mkwd{Bool}}
\newcommand{\stringty}{\mkwd{String}}

\newcommand{\deriv}[1]{\mathbf{#1}}

\newcommand{\secref}[1]{\S\ref{#1}}
\newcommand{\secrefp}[1]{(\secref{#1})}

\newcommand{\produces}[1]{\:{!}\:#1}
\newcommand{\thread}[1]{(\!( #1 )\!)}
\newcommand{\wcirc}{\circ}
\newcommand{\bcirc}{\bullet}

\newcommand{\subscriptionty}[1]{\mkwd{Sub}(#1)}
\newcommand{\subscription}[2]{\calcwd{sub} \: #1 \: #2}
\newcommand{\subempty}{\calcwd{subEmpty}}

\newcommand{\eh}{h}

\newcommand{\metadef}[1]{\mkwd{#1}}

\newcommand{\vs}{{V_{\mathsf{S}}}}

\newcommand{\translate}[1]{\llbracket #1 \rrbracket}

\newcommand{\gvsend}[2]{\calcwd{send} \app (#1, #2)}
\newcommand{\gvrecv}[1]{\calcwd{receive} \app #1}
\newcommand{\gvcancel}[1]{\calcwd{cancel} \app #1}
\newcommand{\gvclose}[1]{\calcwd{close} \app #1}
\newcommand{\gvnew}[1]{\calcwd{new} \: #1}
\newcommand{\tryasinotherwise}[4]{\calcwd{try} \: #1 \: \calcwd{as} \: #2 \: \calcwd{in} \: #3 \: \calcwd{otherwise} \: #4}
\newcommand{\raiseexn}{\calcwd{raise}\xspace}

\newcommand{\gvout}[2]{{!}#1.#2}
\newcommand{\gvin}[2]{{?}#1.#2}
\newcommand{\gvoutone}[1]{{!}#1}
\newcommand{\gvinone}[1]{{?}#1}

\newcommand{\gvend}{\mkwd{End}}

\newcommand{\zap}[1]{\lightning #1}
\newcommand{\gvdual}[1]{\overline{#1}}

\newcommand{\var}[1]{\textit{#1}}

\lstdefinelanguage{Links}{%
  morekeywords={typename, fun, linfun, op, var, if, this, true, false, else, case, switch, handle,
    handler, shallowhandler, open, do, sig, new, send, receive, spawnAt, spawn,
module, request, accept, try, as, otherwise, catch, offer, select, raise,
fork, spawnClient, cancel, catch, page, close, Any, Type, Unl, forall, vdom},%
  sensitive=t, %
  comment=[l]{\#\ },%
  escapeinside={(*}{*)},%
  morestring=[d]{"},%
  keywordstyle=\color{blue},
  showstringspaces=false
 }

\newcommand{\desugarterm}[1]{\llbracket #1 \rrbracket}
\newcommand{\desugarhtml}[1]{\llbracket #1 \rrbracket}
\newcommand{\desugarattr}[1]{\llbracket #1 \rrbracket}

\newcommand{\coretagone}[1]{\calcwd{htmlTag} \: \tagname{#1}}

\newcommand{\coretag}[3]{\calcwd{htmlTag} \: \tagname{#1} \: #2 \: #3}
\newcommand{\htmltext}[1]{\calcwd{htmlText} \: #1}
\newcommand{\htmlempty}{\calcwd{htmlEmpty}}
\newcommand{\append}[2]{#1 \star #2}

\newcommand{\attr}[2]{\calcwd{attr} \: #1 \: #2}
\newcommand{\attrempty}{\calcwd{attrEmpty}}
\newcommand{\ak}{\mathit{ak}}
\newcommand{\at}{\mathit{at}}
\newcommand{\run}[1]{\calcwd{run} \: #1}
\newcommand{\runexp}[3]{\calcwd{run} \: (#1, #2, #3)}
\newcommand{\runsub}[4]{\calcwd{run} \: (#1, #2, #3, #4)}

\newenvironment{fake}[1]{\par\vspace{3pt}\noindent\textbf{#1}\itshape}{\normalfont\ignorespacesafterend\vspace{3pt}\par}

\lstdefinelanguage{JavaScript}{
  morekeywords=[1]{break, continue, delete, else, for, function, if, in,
    new, return, this, typeof, var, void, while, with},
  morekeywords=[2]{false, null, true, boolean, number, undefined,
    Array, Boolean, Date, Math, Number, String, Object, const},
  morekeywords=[3]{eval, parseInt, parseFloat, escape, unescape},
  sensitive,
  morecomment=[s]{/*}{*/},
  morecomment=[l]//,
  morecomment=[s]{/**}{*/}, %
  morestring=[b]',
  morestring=[b]"
}[keywords, comments, strings]

\newcommand{\bl}{\begin{array}{l}}
\newcommand{\el}{\end{array}}
\newcommand{\intstr}[1]{\mkwd{intToString}(#1)}
\newcommand{\cmdty}[1]{\mkwd{Cmd}(#1)}
\newcommand{\cmdempty}{\calcwd{cmdEmpty}}
\newcommand{\cmdspawn}[1]{\calcwd{cmdSpawn} \: #1}

\newcommand{\procs}[1]{\mkwd{procs}(#1)}

\definecolor{shade}{RGB}{215,215,215}
\newcommand{\shade}[1]{\setlength{\fboxsep}{0pt}\colorbox{shade}{\vphantom{$\mid$}\ensuremath{#1}}}

\newcommand{\hlred}[1]{{\color{red}{#1}}}
\newcommand{\transition}[5]{\calcwd{transition} \: #1 \: #2 \: #3 \: #4 \: #5}

\newcommand{\transitionty}[2]{\mkwd{Transition}(#1, #2)}

\newcommand{\notransition}[2]{\calcwd{noTransition} \: #1 \: #2}
\newcommand{\handlerproctrans}[3]{\handlerproc{#1}{#2}_{#3}}

\newcommand{\handlerprocexptc}[5]{\handlerprocsub{#1}{#2}{#3}{#4}_{#5}}
\newcommand{\vdashtrans}[2]{\vdash^{#1}_{#2}}
\newcommand{\threadtrans}[2]{(\!( #1 )\!)_{#2}}

\newcommand{\updating}[1]{\calcwd{updating} \: #1}
\newcommand{\extracting}[2]{\calcwd{extracting}[{#1}] \: #2}
\newcommand{\extractingt}[3]{\calcwd{extractingT}[{#1, #2}] \: #3}
\newcommand{\extractingtexp}[5]{\calcwd{extractingT}[{(#1, #2, #3), #4}] \: #5}
\newcommand{\rendering}[2]{\calcwd{rendering}[{{#1}}] \: #2}
\newcommand{\renderingext}[3]{\calcwd{rendering}[{{#1}, {#2}}] \: #3}
\newcommand{\transitioning}[4]{\calcwd{transitioning}[{{#1, #2, #3}}] \: #4}
\newcommand{\transitioningextexp}[6]{\calcwd{transitioning}[{{#1, (#2, #3, #4), #5}}] \: #6}
\newcommand{\transitioningext}[4]{\calcwd{transitioning}[{{#1, #2, #3}}] \: #4}
\newcommand{\version}[1]{\mkwd{version}(#1)}

\makeatletter
 \renewenvironment{abstract}{%
   \vskip0.7\bigskipamount
   \noindent
   \rlap{\color{lipicsLineGray}\vrule\@width\textwidth\@height1\p@}%
   \hspace*{7mm}\fboxsep1.5mm\colorbox[rgb]{1,1,1}{\raisebox{-0.4ex}{%
     \large\selectfont\sffamily\bfseries\abstractname}}%
   \vskip3\p@
   \fontsize{9}{12}\selectfont
   \noindent\ignorespaces}
   {
    \ifx\@hideLIPIcs\@undefined
     \vskip0.7\baselineskip\noindent
    \subjclassHeading
    \ifx\@ccsdescString\@empty
      \textcolor{red}{Author: Please fill in 1 or more \string\ccsdesc\space macro}%
    \else
      \@ccsdescString
    \fi
    \vskip0.7\baselineskip
    \noindent\keywordsHeading
    \ifx\@keywords\@empty
      \textcolor{red}{Author: Please fill in \string\keywords\space macro}%
    \else
      \@keywords
    \fi
      \ifx\@DOIPrefix\@empty\else
        \vskip0.7\baselineskip\noindent
        \doiHeading\href{https://doi.org/\@DOIPrefix.\@EventAcronym.\@EventYear.\@ArticleNo}{\@DOIPrefix.\@EventAcronym.\@EventYear.\@ArticleNo}%
      \fi
    \ifx\@category\@empty\else
      \vskip0.7\baselineskip\noindent
      \categoryHeading\@category
    \fi
    \ifx\@relatedversion\@empty\else
      \vskip0.7\baselineskip\noindent
      \relatedversionHeading\@relatedversion
    \fi
    \ifx\@supplement\@empty\else
      \vskip0.7\baselineskip\noindent
      \supplementHeading\@supplement
    \fi
    \ifx\@funding\@empty\else
      \vskip0.7\baselineskip\noindent
      \fundingHeading\ifx\authoranonymous\relax\textcolor{red}{Anonymous funding}\else\@funding\fi
    \fi
    \ifx\@acknowledgements\@empty\else
      \vskip0.7\baselineskip\noindent
      \acknowledgementsHeading\ifx\authoranonymous\relax\textcolor{red}{Anonymous acknowledgements} \else\@acknowledgements\fi
    \fi
    \fi
  }
\newcommand{\pagety}[1]{\mkwd{Page}(#1)}
\newcommand{\pgtag}[4]{\calcwd{domTag}({#4}) \: \tagname{#1} \: #2 \: #3}
\newcommand{\pgempty}{\calcwd{domEmpty}}
\newcommand{\pgtext}[1]{\calcwd{domText} \: #1}
\newcommand{\handlers}[2]{\mkwd{handlers}(#1, #2)}
\newcommand{\diff}[2]{\mkwd{diff}(#1, #2)}
\newcommand{\rec}[3]{\calcwd{rec}\app #1(#2) \app . \app #3}
\newcommand{\rectwo}[2]{\calcwd{rec}\app #1(#2) \app .}
\newcommand{\redrowskip}{0.25em}
\newcommand{\syntaxskip}{0.2em}
\newcommand{\qsep}{\,@\,}
\newcommand{\lin}{\mathsf{L}}
\newcommand{\unr}{\mathsf{U}}
\newcommand{\kind}{\kappa}
\newcommand{\kto}[1]{\to^{#1}}
\newcommand{\rtv}[1]{#1}
\newcommand{\recty}[2]{\mu #1 . #2}
\newcommand{\constty}[1]{\Sigma(#1)}
\newcommand{\vers}{\iota}
\newcommand{\uto}{\kto{\unr}}
\newcommand{\lto}{\kto{\lin}}
\newcommand{\halt}{\calcwd{halt}}

\newcommand{\ep}{E_{\textsf{P}}}
\newcommand{\tp}{\config{T}_{\textsf{P}}}
\newcommand{\ta}{\config{T}_{\textsf{A}}}
\newcommand{\domh}{D_{\mkwd{H}}}
\newcommand{\erase}[1]{\mkwd{erase}(#1)}
\newcommand{\attribute}[1]{\mkwd{#1}}
\newcommand{\textnode}[1]{\texttt{#1}}
\newcommand{\fn}[1]{\mkwd{fn}(#1)}
\newcommand{\evtloopty}[2]{\mkwd{EvtLoop}(#1, #2)}
\newcommand{\evtlooptytrans}[3]{\mkwd{EvtLoop}(#1, #2, #3)}
\newcommand{\statetytrans}[3]{\mkwd{State}(#1, #2, #3)}
\newcommand{\qqquad}{\quad \qquad}

\newcommand{\serverthread}[1]{\lfloor #1 \rfloor}
\newcommand{\procctx}{\config{P}}
\newcommand{\simplechan}[1]{\mkwd{Chan}(#1)}
 \hideLIPIcs
\nolinenumbers
\begin{document}

\title{Model-View-Update-Communicate: Session Types meet the Elm Architecture (Extended version)}
\titlerunning{Model-View-Update-Communicate}

\EventEditors{Robert Hirschfeld and Tobias Pape}
\EventNoEds{2}
\EventLongTitle{34th European Conference on Object-Oriented Programming (ECOOP 2020)}
\EventShortTitle{ECOOP 2020}
\EventAcronym{ECOOP}
\EventYear{2020}
\EventDate{July 13--17, 2020}
\EventLocation{Berlin, Germany}
\EventLogo{}
\SeriesVolume{166}
\ArticleNo{14}

\author{Simon Fowler}{University of Edinburgh, Scotland}{simon.fowler@ed.ac.uk}{https://orcid.org/0000-0001-5143-5475}{}%
\authorrunning{S.\ Fowler}

\Copyright{Simon Fowler}

\funding{This work was supported by ERC Consolidator Grant Skye (grant no.\ 682315) and an ISCF Metrology Fellowship grant provided by the UK government’s Department for Business, Energy and Industrial Strategy (BEIS).
}

\relatedversion{An extended version of the paper is available on arXiv (\url{https://arxiv.org/abs/1910.11108}).}

\acknowledgements{
I thank Jake Browning for sparking my interest in Elm and for his help with
an early prototype of the Links MVU library; S\'ara Decova for a previous
version of the multi-room chat server example; Sam Lindley for many useful
discussions and suggestions; and James Cheney,
April Gon\c{c}alves, and the anonymous ECOOP PC and AEC reviewers for detailed
comments.}

\maketitle
\begin{abstract}
Session types are a type discipline for communication channel endpoints
which allow conformance to protocols to be checked statically. Safely implementing
session types requires linearity, usually in the form of a linear type system.
Unfortunately, linear typing is difficult to integrate with graphical user
interfaces (GUIs), and to date most programs using session types are command
line applications.

In this paper, we propose the first principled integration of session typing and
GUI development by building upon the Model-View-Update (MVU) architecture,
pioneered by the Elm programming language. We introduce $\mvu$, the first formal
model of the MVU architecture, and prove it sound. By extending $\mvu$
with \emph{commands} as found in Elm, along with \emph{linearity} and
\emph{model transitions}, we show the first formal integration of session
typing and GUI programming.
We implement our approach in the Links web programming language, and show
examples including a two-factor authentication workflow and multi-room chat
server.
\end{abstract}

\section{Introduction}\label{sec:introduction}

Modern applications are necessarily concurrent and distributed. Along with concurrency
and distribution naturally comes communication, but communication protocols are
typically informally described, resulting in costly runtime failures and code
maintainability issues.

\emph{Session types}~\cite{Honda93:dyadic, HondaVK98:primitives} are a type
discipline for communication channel endpoints which allow conformance to a
protocol to checked statically rather than after an application is deployed.
Many distributed GUI applications, such as chat applications or multiplayer
games, would benefit from session-typed communication with a server.
Unfortunately, safely implementing session types requires a require a linear
type system,
but safely integrating linear resources and GUIs is nontrivial.
As a consequence, to date most programs using session types are batch-style
applications run on the command line.

The lack of a principled integration of GUI applications and session types is a
significant barrier to their adoption. In this paper, we bridge this gap
by extending the Model-View-Update (MVU) architecture, pioneered by the Elm
programming language, to support linear resources. We present \mvu, a core
formalism of the MVU architecture, and an extended version of \mvu which
supports session-typed communication. Informed by the formal development, we
provide a practical implementation in the Links programming
language~\cite{CooperLWY06:links}.

\subparagraph{Session types by example.}

Let us consider a two-factor authentication workflow, introduced
by~\citet{FowlerLMD19:stwt}.  A user first enters their credentials.  If
correct, the server can then either grant access, or send a challenge key.  If
challenged, the user enters the challenge code into a hardware token, which
generates a response to be entered into the web page. The server then either
authenticates the user or denies access.

We can describe the two-factor authentication example as a session type as
follows:

    {\footnotesize
  \begin{minipage}{0.45\textwidth}
\[
\bl
  \metadef{TwoFactorServer} \defeq \\
  \quad\gvinone{(\mkwd{Username}, \mkwd{Password})} . \mathord{\oplus} \{ \\
  \qquad \mkwd{Authenticated}: \metadef{ServerBody}, \\
  \qquad \mkwd{Challenge}: \gvoutone{\mkwd{ChallengeKey}} .
  \gvinone{\mkwd{Response}} . \\
  \qquad \quad \mathord{\oplus} \{
    \begin{aligned}[t]
     & \mkwd{Authenticated}: \metadef{ServerBody}, \\
     & \mkwd{AccessDenied}: \gvend \}, \\
  \end{aligned} \\
  \qquad \mkwd{AccessDenied}: \gvend \}
\el
\]%
\end{minipage}
\hfill
\begin{minipage}{0.45\textwidth}
\[
\bl
  \metadef{TwoFactorClient} \defeq \\
  \quad\gvoutone{(\mkwd{Username}, \mkwd{Password})} . \mathord{\&} \{ \\
  \qquad \mkwd{Authenticated}: \metadef{ClientBody}, \\
  \qquad \mkwd{Challenge}: \gvinone{\mkwd{ChallengeKey}} .
    \gvoutone{\mkwd{Response}} . \\
  \qquad \quad \mathord{\&} \{
    \begin{aligned}[t]
    & \mkwd{Authenticated}: \metadef{ClientBody}, \\
    & \mkwd{AccessDenied}: \gvend \}, \\
    \end{aligned} \\
    \qquad \mkwd{AccessDenied}: \gvend \}
\el
\]
\end{minipage}
}

The \mkwd{TwoFactorServer} type shows the session type for the server, which
firstly receives ($?$) the credentials from the client, and then chooses
($\oplus$) whether to authenticate, deny access, or issue a challenge. If
the server issues a challenge, it sends ($!$) the challenge string, awaits the
response, and then chooses whether to accept or reject the request. The
\mkwd{ServerBody} type abstracts over the actions performed in the remainder of
the application, for example taking out a loan.
The \mkwd{TwoFactorClient} type is the \emph{dual} of the \mkwd{TwoFactorServer}
type: where the server sends, the client receives, and where the client sends,
the server receives. The $\mathord{\&}$ construct denotes offering a choice of
branches.  Suppose we have constructs for sending along, receiving from, and
closing an endpoint:

\vspace{-1em}
{\footnotesize
  \begin{mathpar}
  \calcwd{send} : (A \times \gvout{A}{S}) \to S

  \calcwd{receive} : \gvin{A}{S} \to (A \times S)

  \calcwd{close} : \gvend \to \one
  \end{mathpar}
}%
Let us also suppose we have constructs for selecting and offering a choice:
{\footnotesize
\[
  \begin{array}{ll}
    \calcwd{select} \: \ell_j \: M : S_j & \quad \text{where } M \text{ has session type } \mathord{\oplus} \{ \ell_i : S_i \}_{i \in I},
    \text{ and } j \in I \\
    \calcwd{offer} \: M \: \{ \ell_i(x_i) \mapsto N_i \}_{i \in I} : A &
    \quad \text{where } M \text{ has session type }  \& \{ \ell_i : S_i \}_{i \in I}, \text{ each }
    x_i \text{ binds an} \\
   & \quad \text{endpoint with session type }S_i, \text{ and each } N_i \text{ has type } A
  \end{array}
\]%
}%
We can write a server implementation as follows:
{\footnotesize
\[
  \bl
  \metadef{twoFactorServer} : \mkwd{TwoFactorServer} \to \one \\
  \metadef{twoFactorServer}(s) \defeq
  \begin{array}[t]{l}
       \letintwo{((\textit{username}, \textit{password}), s)}{\gvrecv{s}} \\
       \calcwd{if} \: \metadef{checkDetails}(\textit{username}, \textit{password})\: \calcwd{then} \\
       \quad \letintwo{s}{\calcwd{select} \: \mkwd{Authenticated} \: s} \: \metadef{serverBody}(s) \\
       \calcwd{else} \: \letintwo{s}{\calcwd{select} \: \mkwd{AccessDenied} \: s} \: \gvclose{s}
  \end{array}
  \el
\]}%
To implement session-typed communication safely, we require a linear type
system~\cite{Wadler90:linear-types} to ensure each communication endpoint
is used exactly once: as an example, without linearity it would be possible to
attempt to receive the credentials twice.

\subparagraph{Linearity and GUIs.}
We can also write a client application:

{\footnotesize
\[
  \bl
  \metadef{twoFactorClient} : (\mkwd{Username} \times \mkwd{Password} \times \mkwd{TwoFactorClient}) \to \one \\
  \metadef{twoFactorClient}(\textit{username}, \textit{password}, s) \defeq \\
  \quad \letintwo{s}{\gvsend{(\textit{username}, \textit{password})}{s}} \\
  \quad \calcwd{offer} \: s \: \{
    \begin{array}[t]{@{}l@{~}c@{~}l@{}}
     \mkwd{Authenticated}(s) &\mapsto& \metadef{clientBody}(s) \\
     \mkwd{Challenge}(s)     &\mapsto&
       \begin{array}[t]{l}
       \letintwo{(\textit{key}, s)}{\gvrecv{s}} \\
       \letintwo{s}{\gvsend{\metadef{generateResponse}(\textit{key})}{s}} \\
       \calcwd{offer} \: s \: \{
         \begin{array}[t]{@{}l@{~}c@{~}l@{}}
         \mkwd{Authenticated}(s) &\mapsto& \metadef{clientBody}(s) \\
         \mkwd{AccessDenied}(s)  &\mapsto& \gvclose{s}; \metadef{loginFailed} \} \\
         \end{array} \\
       \end{array} \\
    \mkwd{AccessDenied}(s) &\mapsto &\gvclose{s}; \metadef{loginFailed}\} \\
    \end{array} \\
  \el
\]
}

\noindent
However, such a client is of little use, as it sends only a pre-defined set of
credentials, and the step where a user enters the response to the challenge is
replaced by a function $\mkwd{generateResponse}$. Ideally, we would like the
credentials to be entered into a GUI, and for a button press to trigger the session
communication with the server.

Let us attempt to write a GUI for the first stage of the two-factor
authentication example; as HTML is well-understood, we concentrate on web pages
in the remainder of the paper.

{\footnotesize
\begin{minipage}{0.4\textwidth}
\[
  \bl
  \mkwd{render}(c) \defeq \\
  \quad \tagzero{html} \\
  \qquad \tagzero{body} \\
  \qqquad \htmltag{input}{\attribute{id} = \textnode{"username"}}{} \\
  \qqquad \htmltag{input}{\attribute{id} = \textnode{"password"}}{} \\
  \qqquad \htmltag{button}{\attribute{onClick} =
  \mkwd{login}(c)}{\textnode{Submit}} \\
  \qquad \tagzeroend{body} \\
  \quad \tagzeroend{html}
  \el
\]
\end{minipage}
\hfill
\begin{minipage}{0.4\textwidth}
  \[
  \bl
  \mkwd{login}(c) \defeq \lambda (). \\
  \quad \letintwo{\var{user}}{\mkwd{getContents}(\textnode{"username"})} \\
  \quad \letintwo{\var{pass}}{\mkwd{getContents}(\textnode{"password"})} \\
  \quad \letintwo{c}{\gvsend{(\var{user}, \var{pass})}{c}} \\
  \quad \mkwd{handleResponse}(c)
  \el
\]
\end{minipage}
}

Given a channel $c$ of type $\mkwd{TwoFactorClient}$, the \mkwd{render} function
generates a web page with input boxes for the username and password, and a
button to submit the credentials. The \mkwd{login} function, triggered when the
button is clicked, retrieves the username and password from the two input boxes,
and sends the credentials along $c$. The $\mkwd{handleResponse}$ function, which
we omit, receives the response from the server and updates the web page.

On first inspection, this implementation seems sound since the endpoint $c$ is used
linearly. However, the above attempt is unsound due to the
asynchronous nature of GUI programming: there is nothing stopping the user
pressing the button twice and sending the credentials twice along $c$, in
contravention of the session type.
As a further complication, suppose we augmented the protocol with a ``forgotten
password'' branch, triggered by another button. This would require two instances
of $c$ in the GUI, again violating linearity:

\vspace{-0.75em}
{\footnotesize
\[
  \bl
  \htmltag{button}{\attribute{onClick} = \mkwd{login}(\hlred{c})}{\textnode{Submit}} \\
  \htmltag{button}{\attribute{onClick} = \mkwd{reset}(\hlred{c})}{\textnode{Reset password}}
  \el
\]%
}
\noindent
It is clear that directly embedding linear resources into a GUI is a
non-starter. A more successful approach involves spawning a separate process
which contains the linear resource, and which receives \emph{non-linear}
messages from the GUI. Upon receiving a GUI message, the process can then perform
the session communication, while ignoring duplicate GUI messages:

{\footnotesize
\begin{minipage}{0.4\textwidth}
\[
  \bl
  \mkwd{render}(c) \defeq \\
  \quad \letintwo{\var{pid}}{\calcwd{spawn} \: \mkwd{handler}(c)} \\
  \quad \tagzero{html} \\
  \qquad \tagzero{body} \\
  \qqquad \htmltag{input}{\attribute{id} = \textnode{"username"}}{} \\
  \qqquad \htmltag{input}{\attribute{id} = \textnode{"password"}}{} \\
  \qqquad \htmltag{button}{\attribute{onClick} =
  \mkwd{login}(\var{pid})}{\textnode{Submit}} \\
  \qquad \tagzeroend{body} \\
  \quad \tagzeroend{html}
  \el
\]
\end{minipage}
\hfill
\begin{minipage}{0.4\textwidth}
  \[
  \bl
  \mkwd{login}(\var{pid}) \defeq \lambda (). \\
  \quad \letintwo{\var{user}}{\mkwd{getContents}(\textnode{"username"})} \\
  \quad \letintwo{\var{pass}}{\mkwd{getContents}(\textnode{"password"})} \\
  \quad \var{pid} \: ! \: \mkwd{SubmitLogin}(\var{user}, \var{pass})\\ \\
  \mkwd{handler}(c) \defeq \\
  \quad \caseofone{(\calcwd{get} \: ())} \: \{ \\
    \qquad \mkwd{SubmitLogin}(\var{user}, \var{pass}) \mapsto \\
    \qquad
    \bl
    \quad \letintwo{c}{\gvsend{(\var{user}, \var{pass})}{c}} \\
    \quad \mkwd{handleResponse}(c)
    \el \\
  \quad \}
  \el
\]
\end{minipage}
}

\noindent
The \mkwd{render} function begins by spawning $\mkwd{handler}(c)$ as a separate
process with an incoming message queue (or \emph{mailbox}), returning the
process ID $\var{pid}$. As before, the \mkwd{login} function is triggered by
pressing the button, and retrieves the credentials from the web page. Instead of
communicating on the channel directly, it sends a \mkwd{SubmitLogin} message
containing the credentials to the process ID of handler process, written
$\var{pid} \: ! \: \mkwd{SubmitLogin}(\var{user}, \var{pass})$. The handler
process retrieves the message from its mailbox ($\calcwd{get} \: ()$), and can
then communicate with the server over the linear endpoint. Such an approach also
scales to the ``forgotten password'' extension, by adding another GUI message.

The above approach is used by~\citet{FowlerLMD19:stwt}, who provide the first
integration of session types and web application development, including the
ability to gracefully handle failures such as the user closing their browser
mid-session.
Unfortunately, the approach is brittle and ad-hoc.
All interaction with the web page occurs using imperative operations such
as $\mkwd{getContents}$ and $\mkwd{setContents}$; contrary to best practices
such as the Model-View-Controller (MVC)~\cite{KrasnerP88:mvc} pattern, the state of the web page is
not derived directly from the data contained by the application.
Furthermore, there is no connection between the state of the handler process and what
is displayed on the web page: this can easily lead to mismatches between the
possible GUI messages which can be sent and which can be handled.

\subparagraph{Model-View-Update.}
This paper is about doing better. Our approach is to formalise
Model-View-Update, an architectural pattern for GUI development popularised by
the Elm programming language~\cite{elm-lang}, and extend it to support linear
resources. MVU is an appealing starting point as it is
particularly suited to functional programming. Furthermore, MVU has directly
inspired popular technologies such as Redux~\cite{redux} and the Flux
architecture~\cite{flux}, which are used with the popular React~\cite{react}
frontend web framework for JavaScript.

The Elm programming language~\cite{elm-lang} is a functional programming
language designed for writing web applications. Elm was
originally designed to use \emph{functional reactive programming}
(FRP)~\cite{ElliottH97:frp}, where time-varying \emph{signals} can be used to construct reactive web
applications.  A paper describing Elm, and its core formal semantics, was
published at PLDI 2013~\cite{CzaplickiC13:elm}.

For many languages, that would be the end of the story. But unusually for a
research language, Elm gained a user community, and a standard architectural
pattern known as \emph{The Elm Architecture} grew organically to such a point
that Elm abandoned FRP altogether~\cite{farewell-frp}.
At its core, The Elm Architecture is a descendant of MVC where a \emph{model}
contains the state of the application; a \emph{view function} renders the
model; and the rendered model produces \emph{messages} which are handled
by an \emph{update} function to produce a new model. More generally, this
pattern has been referred to as \emph{Model-View-Update}, or MVU for
short~\cite{websharper-mvu, mvu-flutter}.

Consider the following web application, where a user enters text into a text
box, and the application displays the text, reversed:
\vspace{-0.5em}
\begin{center}
  \includegraphics[scale=0.45]{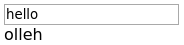}
\end{center}%
\vspace{-1em}

\noindent
We can write this example using MVU as follows:
\vspace{-1em}

{\footnotesize
  \begin{minipage}[t]{0.45\textwidth}
  \[
    \bl
    \mkwd{Model} \defeq (\var{contents} : \stringty) \\
    \mkwd{Message} \defeq \mkwd{UpdateBox}(\stringty) \\ \\
    \mkwd{model} : \mkwd{Model} \\
    \mkwd{model} \defeq (\var{contents} = \textnode{""}) \\ \\
  \mkwd{update} : (\mkwd{Message} \times \mkwd{Model}) \to \mkwd{Model} \\
  \mkwd{update} \defeq \lambda (\mkwd{UpdateBox}(str), \var{m}) . (\textit{contents} = \var{str})
  \el
\]
\end{minipage}
\hfill
\begin{minipage}[t]{0.5\textwidth}
\[
  \bl
  \mkwd{view} : \mkwd{Model} \to \htmlty{\mkwd{Message}} \\
  \mkwd{view} \defeq \lambda \textit{model} . \calcwd{html} \\
      \quad \opentag{input} \: \attribute{type}=\textnode{``text''} \:
      \attribute{value} = \{ \var{model}.\var{contents} \} \\
      \qquad
        \attribute{onInput} = \{ \lambda \textit{str}.
          \mkwd{UpdateBox}(\textit{str})\} \closetag \tagzeroend{input}  \\
          \quad
          \tagzero{div} \\
          \qquad \antiquote{\htmltext{(\metadef{reverseString} \: (\textit{model}. \textit{contents}))}} \\
          \quad \tagzeroend{div} \\ \\
  (\mkwd{model}, \mkwd{view}, \mkwd{update})
  \el
\]
\end{minipage}
}
\vspace{0.5em}

\noindent
We define two type aliases: the \mkwd{Model} captures the state of the
application and is defined as a record with a single $\stringty$ field,
$\var{contents}$.
\emph{Messages} are produced as a result of user interaction. The \mkwd{Message}
type is defined as a singleton variant type with constructor
\mkwd{UpdateBox}, containing the updated value of the text box.

The \mkwd{view} function renders a model. It has the type
$\mkwd{Model} \to \htmlty{\mkwd{Message}}$, which is a function taking a
\mkwd{Model} as its argument, and returning HTML which may
produce messages of type \mkwd{Message}.
The $\attribute{value} = \antiquote{\var{model.contents}}$ attribute of the
\mkwd{input} box states that the contents of the text box should reflect
the $\var{contents}$ field of the model.
The \mkwd{onInput} attribute is an \emph{event handler}: its body is a
function taking the current value of the input box ($\var{str}$) and producing
an \mkwd{UpdateBox} message containing the updated contents of the box. The
contents of the $\mkwd{div}$ tag are derived from the reversed contents.

The \mkwd{update} function takes a message and previous model as its
arguments, and produces a new model. In this case, the \mkwd{update} function
constructs a new model where the \var{contents} field is set to the
payload of the \mkwd{UpdateBox} message.
Finally, the program is a 3-tuple containing the initial model, and the view and update functions.

To achieve our goal of a formal integration of session typing and GUI
programming, we first formalise MVU, and then generalise the
architecture to support linear models and messages. Supporting linearity poses
some challenges, as we will see in~\secref{sec:extensions}.

\subsection{Contributions.}
The overarching contribution of this paper is the first principled integration
of session-typed communication with a GUI framework. Concretely, we make three contributions:
\begin{enumerate}
  \item We introduce the first formal model of the MVU
    architecture, \mvu~\secrefp{sec:formalism}. We prove~\secrefp{sec:metatheory}
    that \mvu satisfies preservation and event progress properties.
  \item We extend \mvu with \emph{commands}, \emph{linearity}, and \emph{model
    transitions}~\secrefp{sec:extensions}, which allow \mvu to support GUIs incorporating session-typed
    communication, and we prove the soundness of the extended calculus.
  \item We implement the architecture in the Links web programming language. We
    show an extended example of a chat application where client code
    uses the linear MVU framework, and where client-server communication happens
    over session-typed channels~\secrefp{sec:example}.

    The implementation and examples are available in the paper's companion artifact.
\end{enumerate}
\begin{figure}[t]
  {
\footnotesize
~\textbf{Syntax} \hfill
\begin{syntax}
  \text{Types} & A, B, C & ::= & \one \midspace A \to B \midspace A \times B
  \midspace A + B \midspace \stringty \midspace \intty \\
               &      & \midspace & \htmlty{A} \midspace \attrty{A} \vspace{\syntaxskip} \\
\text{String literals} & s \vspace{\syntaxskip}\\
\text{Integers} & n \vspace{\syntaxskip}\\
\text{Terms} & L, M, N & ::= &
  x \midspace \lambda x . M \midspace \rec{f}{x}{M} \midspace M \app N \midspace () \midspace s \midspace n \\
  & & \midspace & (M, N) \midspace \letin{(x, y)}{M}{N} \\
  & & \midspace & \inl{x} \midspace \inr{x} \midspace \caseof{L}{\inl{x} \mapsto M; \inr{y} \mapsto N} \\
  & & \midspace & \coretag{t}{M}{N} \midspace \htmltext{M} \midspace \htmlempty \\
  & & \midspace & \attr{\ak}{M} \midspace \attrempty \midspace \append{M}{N} \\
  \end{syntax}
  \[
  \begin{array}{llp{5em}lrcl}
    \text{Tag names} & \tagname{t} & & \text{Attribute keys} & \ak & ::= & \mathit{at} \midspace h \\
    \text{Attribute names} \:\: & \mathit{at} & & \text{Event handler names} & h
    \end{array}
  \]
  \vspace{1em}

  ~\textbf{Typing rules for terms} \hfill \framebox{$\Gamma \vdash M : A$}
\begin{mathpar}
    \inferrule
    [T-HtmlTag]
    { \Gamma \vdash M : \attrty{A} \\ \Gamma \vdash N : \htmlty{A} }
    { \Gamma \vdash \coretag{t}{M}{N} : \htmlty{A}}

    \inferrule
    [T-HtmlText]
    { \Gamma \vdash M : \stringty }
    { \Gamma \vdash \htmltext{M} : \htmlty{A} }

    \inferrule
    [T-HtmlEmpty]
    { }
    { \Gamma \vdash \htmlempty : \htmlty{A} }

    \inferrule
    [T-Attr]
    { \Gamma \vdash M : \stringty }
    { \Gamma \vdash \attr{\at}{M} : \attrty{A} }

    \inferrule
    [T-EvtAttr]
    { \Gamma \vdash M : \evtty{h} \to A }
    { \Gamma \vdash \attr{h}{M} : \attrty{A} }

    \inferrule
    [T-AttrEmpty]
    { }
    { \Gamma \vdash \attrempty : \attrty{A}}

    \inferrule
    [T-HtmlAppend]
    { \Gamma \vdash M : \htmlty{A} \\ \Gamma \vdash N : \htmlty{A} }
    { \Gamma \vdash \append{M}{N} : \htmlty{A} }

    \inferrule
    [T-AttrAppend]
    { \Gamma \vdash M : \attrty{A} \\ \Gamma \vdash N : \attrty{A} }
    { \Gamma \vdash \append{M}{N} : \attrty{A} }
\end{mathpar}
}
\caption{Syntax and typing rules for \mvu terms}
\label{fig:mvu-syntax}
\end{figure}
 \section{Model-View-Update, Formally}\label{sec:formalism}
In this section, we formalise MVU as a core calculus, \mvu, an extension of the
simply-typed $\lambda$-calculus with products, sums, HTML, and event
handling.  Even without extensions, \mvu is expressive enough to support many
common applications such as form handling.

\subsection{Syntax}

\subparagraph{Types.}
Figure~\ref{fig:mvu-syntax} shows the syntax and typing rules for \mvu.
Types are ranged over by $A,B,C$, and consist of the unit type $\one$, functions
$A \to B$, products $A \times B$, sums $A + B$, and string and integer types.
Types $\htmlty{A}$ and $\attrty{A}$ are the type of HTML
elements and attributes which can produce messages of type $A$.

\subparagraph{Terms.}
Terms, ranged over by $L, M, N$, include variables, $\lambda$ abstractions,
anonymous recursive functions, function application, the unit value, string
literals, integers, and sum and pair introduction and elimination.
The remaining terms encode HTML \emph{elements} and \emph{attributes}.
The $\coretag{t}{M}{N}$ construct denotes an HTML element with tag name
$\tagname{t}$ (for example, \lstinline+div+), attributes $M$, and children
$N$; the $\htmltext{M}$ construct
describes a text node with text $M$; and $\htmlempty$ defines an empty HTML node.

The $\attr{\ak}{M}$ construct describes an attribute with key $\ak$ and body
$M$, where the key $\ak$ is either an attribute name $\at$ or an event handler
name $h$. The $\attrempty$ construct defines an empty attribute.

The $\append{M}{N}$ operator appends two HTML elements or
attributes. Since both HTML elements and attributes support a unit element
($\htmlempty$ and $\attrempty$ respectively), elements and attributes together
with $\star$ form two monoids.

\floatstyle{plain}
\restylefloat{figure}
\begin{figure}[t]
  \centering
{\footnotesize
  \begin{tabular}{| c | x{8em} | x{7em} | c |}\hline
    Event name \evt{ev} &  Event Handler $h$ (\handler{\evt{ev}}) & Payload type
    (\evtty{\evt{ev}}, \evtty{h}) & Payload Description \\ \hline
    \clickevt & \evthandler{onClick} & \one & Unit value\\ \hline
    \inputevt & \evthandler{onInput} & \stringty & Updated contents of a text field \\ \hline
    \keyupevt & \evthandler{onKeyUp} & \intty & Key code  \\ \hline
    \keydownevt & \evthandler{onKeyDown} & \intty & Key code \\ \hline
  \end{tabular}
  \caption{Example event signatures}
  \label{fig:events}
}
\end{figure}
\floatstyle{boxed}
\restylefloat{figure}

\subparagraph{Events.}
We model interaction with the Document Object Model (DOM) through \emph{events},
which model those dispatched by a browser.
An \emph{event signature} is a 3-tuple $(\evt{ev}, h, A)$ consisting of an event
name $\evt{ev}$, handler name $h$, and payload type $A$. We require a
bijective mapping between event and handler names.
Figure~\ref{fig:events} describes example event signatures used in the remainder
of the paper. We consider four primitive events: \clickevt, which is fired when
an element is clicked; \inputevt, which is fired when the contents of a text
field are changed; and \keyupevt and \keydownevt, which are fired when a key is
pressed while focused on an element.

Event handlers are attached to elements as attributes, and generate a message in
response to an event. We write $\handler{\evt{ev}}$ to refer to the handler for
$\evt{ev}$: for example, $\handler{\clickevt} =
\evthandler{onClick}$. We write $\evtty{\evt{ev}}$ to refer to the payload type
of $\evt{ev}$ and write $\evtty{h}$ for the payload
type of an event handled by $h$. As an example, both
$\evtty{\clickevt} = \one$ and $\evtty{\evthandler{onClick}} = \one$.

\subparagraph{Term typing.}
Term typing rules for
$\lambda$-calculus constructs are standard, so are omitted.
Rule \textsc{T-HtmlTag} states that $\coretag{t}{M}{N}$ can be given type
$\htmlty{A}$ if its attributes $M$ have type $\attrty{A}$ and children have type
$\htmlty{A}$. Text nodes $\htmltext{M}$ do not produce any messages, and so
have type $\htmlty{A}$ if $M$ has type $\stringty$
(\textsc{T-HtmlText}); similarly, $\htmlempty$ has type $\htmlty{A}$
(\textsc{T-HtmlEmpty}).

Rule \textsc{T-Attr} assigns attributes $\attr{\at}{M}$
type $\attrty{A}$ for any $A$ if $M$ has type $\stringty$. Rule
\textsc{T-EvtAttr} types event handler attributes $\attr{\eh}{M}$: if the event
handler $M$ has type $\evtty{\eh} \to A$ (i.e., it \emph{produces messages of
type} $A$), then the attribute can be given type $\attrty{A}$. Finally,
\textsc{T-AttrEmpty} states that the empty attribute $\attrempty$ has type
$\attrty{A}$ for any type $A$.
We overload the $\star$ operator to append both HTML elements and attributes
(\textsc{T-HtmlAppend} and \textsc{T-AttrAppend}).

\subparagraph{Syntactic sugar.}
We assume the usual encodings of records as pairs and
variant types as binary sums, and use pattern matching notation.
It is useful to be able to write HTML using XML-like notation, where
an \emph{antiquoted expression} $\antiquote{M}$ allows a term $M$ to be embedded
within an HTML tree.
The \mkwd{view} function from the introduction desugars to:

\begin{minipage}{\textwidth}
  {\footnotesize
  \[
    \begin{array}{l}
      \lambda \textit{model} . \\
      \quad
      \begin{aligned}
       & (\coretagone{input} \\
       & \quad ((\attr{\attribute{type}}{\textnode{``text''}}) \star
       (\attr{\attribute{value}}{\var{model}.\var{contents}}) \star \\
       & \qquad
         (\attr{\attribute{onInput}}{(\lambda \textit{str}.
         \mkwd{UpdateBox}(\textit{str})}))) \: \htmlempty) \: \star \\
       & \quad
         \coretag{div}{\attrempty}{(\htmltext{\metadef{reverseString} \:
         (\textit{model}.\textit{contents}}))}
    \end{aligned}
    \end{array}
  \]%
}
\end{minipage}%

\noindent
The formal definitions and desugaring translations are unsurprising; the details
can be found in the extended version~\cite{Fowler20:mvuc-extended}.

\subsection{Operational Semantics}
We can now provide \mvu with a small-step operational semantics.

\subsubsection{Runtime Syntax}
\begin{figure}[t]
  {\small
  \begin{syntax}
    \text{Values} & U, V, W & ::= & \lambda x . M \midspace \rec{f}{x}{M} \midspace ()
    \midspace (V, W) \midspace
    \inl{V} \midspace \inr{V} \midspace s \midspace n \\ %
                  &         & \midspace & \coretag{t}{V}{W} \midspace \htmlempty \midspace \htmltext{V} \\
                  &         & \midspace & \attr{\ak}{V} \midspace \attrempty \midspace \append{V}{W}
                  \vspace{\syntaxskip} \\
    \text{Events} & e & ::= & \evtpayload{\evt{ev}}{V} \vspace{\syntaxskip}\\
    \text{DOM Pages}  & \vh & ::= & \pgtag{t}{V}{\vh}{\seq{e}} \midspace \pgtext{V}
      \midspace \pgempty \midspace \append{\vh}{\vh'} \vspace{\syntaxskip}\\
    \text{Active thread} & T & ::= & \idle{V_m} \midspace M \vspace{\syntaxskip} \\
    \text{Function state} & F & ::= & \state{V_v}{V_u} \vspace{\syntaxskip} \\
    \text{Processes} & P, Q & ::= & \run{M} \midspace \handlerproc{T}{F}
    \midspace \thread{M} \midspace P \parallel Q \vspace{\syntaxskip} \\
    \text{Configurations}  & \config{C} & ::= & \sys{P}{\vh} \\ \\
    \text{Process contexts} & \procctx & ::= & [~] \midspace \procctx \parallel P \vspace{\syntaxskip} \\
    \text{DOM contexts} & \config{D} & ::= &
      [~] \midspace \pgtag{t}{V}{\config{D}}{\seq{e}} \midspace
      \append{\config{D}}{D} \midspace \append{D}{\config{D}} \vspace{\syntaxskip} \\
    \text{Thread contexts} & \config{T} & ::= & \run{E} \midspace \handlerproc{E}{F} \midspace \thread{E} \\
  \end{syntax}
}
\caption{Runtime syntax for \mvu}
\label{fig:rt-syntax}
\end{figure}
 Figure~\ref{fig:rt-syntax} describes the runtime syntax of \mvu.
Values, ranged over by $U, V, W$, are standard.
An event $\evtpayload{\evt{ev}}{V}$ consists of event name $\evt{ev}$ and
payload $V$.
We write $\epsilon$ for an empty meta-level sequence, and use $\cdot$ for
sequence concatenation.
DOM pages, ranged over by $D$, are the runtime representation of
HTML, where tags $\pgtag{t}{V}{D}{\seq{e}}$ contain an event queue $\seq{e}$ of
events dispatched to the element.

\subparagraph{Concurrency.}
Concurrency is vital when modelling GUI applications as event handling is
asynchronous: computation triggered by a user interaction should not block the
UI. Concurrency is also essential when considering session-typed communication.
We therefore formulate the calculus as a concurrent $\lambda$-calculus
in the style of~\citet{NiehrenSS06:concurrent}, by augmenting the simply-typed
$\lambda$-calculus with processes and concurrent reduction.

\subparagraph{Processes.}
An \emph{initialisation process} $\run{M}$ evaluates the initial system state
written by a user, where $M$ is a 3-tuple containing the initial model, view
function, and update function.
An \emph{event loop process} $\handlerproc{T}{F}$ consists of an active
thread $T$ and function state $F$ comprising the view and update functions. The
thread can either be $\idle{V_m}$, meaning the process has current model $V_m$
and is waiting for another message to process, or evaluating a term $M$.  An
\emph{event handler process} $\thread{M}$ is spawned to generate a message in
response to an event.

\subparagraph{Configurations.}
Concurrent and event-driven reduction happens in the context of a \emph{system
configuration} $\sys{P}{\vh}$, where $P$ is the concurrent fragment of the
system and $\vh$ is the current DOM page.
An MVU program as written by a user is a term $M$ specifying the initial model,
view function, and update function, of type $(A \times (A \to
\htmlty{B}) \times ((B \times A) \to A))$.  A program is evaluated in the
context of an \emph{initial configuration}:

\begin{definition}[Initial configuration]
  An \emph{initial configuration} for a term $M$ is of the form
  $\sys{\run{M}}{\pgempty}$.
\end{definition}

\subparagraph{Evaluation contexts.}
Term evaluation contexts $E$ (omitted) are set up for call-by-value,
left-to-right evaluation.
Process contexts $\procctx$ allow reduction under
parallel composition.
Thread contexts $\config{T}$ allow reduction inside threads. DOM contexts
$\config{D}$ allow us to focus on each element of a DOM forest; note that
they deliberately allow non-unique decomposition in order to support
nondeterministic reduction.

\begin{figure}[t]
  {\footnotesize
      ~\textbf{Meta-level definitions}
\vspace{-1em}

\begin{minipage}{0.25\textwidth}
\[
  \begin{array}{l}
    \handle{\textit{m}, (\textit{v}, \textit{u}), \textit{msg}} \defeq \\
      \quad \letintwo{m'}{u \app (\textit{msg}, m)} \\
      \quad {(m', v \app m')}
  \end{array}
\]
\end{minipage}
\hfill
\begin{minipage}{0.675\textwidth}
\[
\begin{array}{rcl}
  \handlers{\ev}{\attrempty} & = & \epsilon \\
  \handlers{\ev}{\append{V}{W}} & = & \handlers{\ev}{V} \cdot \handlers{\ev}{W} \\
  \handlers{\ev}{\attr{\at}{V}} & = & \epsilon \\
  \handlers{\ev}{\attr{h}{V}} & = &
  \begin{cases}
    V & \text{if } \handler{\ev} = h \\
    \epsilon & \text{otherwise}
  \end{cases}
  \end{array}
\]
\end{minipage}

~\textbf{Process reduction} \hfill \framebox{$P \ceval P'$}
\vspace{-0.5em}
\def\arraystretch{1.2}
\[
  \begin{array}{@{}l@{\:\:}r@{\:\:\:}c@{\:\:\:}l@{}}
    \textsc{EP-Handle} &
    \handlerproc{\idle{V_m}}{F} \parallel \thread{V} & \ceval &
    \handlerproc{\handle{V_m, F, V}}{F} \\
    \textsc{EP-Par} &
      P_1 \parallel P_2 & \ceval & P'_1 \parallel P_2 \qquad \text{if }
      P_1\ceval P'_1 \\
    \textsc{EP-LiftT} &
    \config{T}[M] & \ceval & \config{T}[N] \qquad \text{if } M \teval N
  \end{array}
\]

\vspace{-0.75em}
  ~\textbf{Configuration reduction} \hfill \framebox{$\config{C} \ceval \config{C}'$}
\vspace{-0.5em}
\def\arraystretch{1.2}
\[
  \begin{array}{l@{\:\:\:}c@{\:\:\:}l@{}}
      \textsc{E-Run} \hfill
      \sys{\procctx[\runexp{V_m}{V_v}{V_u}]}{D}  & \ceval & \sys{\procctx[\handlerprocexp{(V_m, V_v \app
      V_m)}{V_v}{V_u}]}{D} \vspace{\redrowskip} \\
    \textsc{E-Update} \hfill
    \sys{\procctx[\handlerproc{(V_m, U)}{F}]}{\vh}
 & \ceval &
 \sys{\procctx[\handlerproc{\idle{V_m}}{F}]}{\vh'} \quad \text{where } \diff{U}{\vh} = \vh' \vspace{\redrowskip} \\
    \textsc{E-Interact} \qquad \hfill
    \sys{P}{\config{D}[\pgtag{t}{U}{D}{\seq{e}}]} & \ceval &
    \sys{P}{\config{D}[\pgtag{t}{U}{D}{\seq{e} \cdot \evtpayload{\evt{ev}}{V}}]}
    \\
                                                  & &
                                                  \quad \text{for some }
                                                  \evt{ev}, V \text{ such that } \vdash \evtpayload{\evt{ev}}{V}
                                                  \vspace{\redrowskip}\\
    \textsc{E-Evt} \\
    \hfill
\sys{P}{\config{D}[\pgtag{t}{U}{D}{\evtpayload{\evt{ev}}{W} \cdot \seq{e}}]}
                 & \ceval &
 \sys{P \parallel \thread{V_1 \app W} \parallel \cdots \parallel \thread{V_n \app W}}{\config{D}[\pgtag{t}{U}{D}{\seq{e}}]} \\
                 & &  \text{where } \handlers{\evt{ev}}{U} = \seq{V}
                 \vspace{\redrowskip} \\
    \textsc{E-Struct} \hfill \config{C} & \ceval & \config{C}' \qquad
    \text{if } \config{C} \equiv \config{C}_1, \config{C}_1 \ceval \config{C}_2, \text{ and } \config{C}_2 \equiv \config{C}' \\
    \textsc{E-LiftP} \hfill
    \sys{P}{D} & \ceval & \sys{P'}{D} \quad \text{if } P \ceval P'
\end{array}
\]
}
\caption{Reduction rules for \mvu terms and configurations}
\label{fig:reduction}
\end{figure}%
\subsubsection{Reduction Rules}
Figure~\ref{fig:reduction} shows the reduction rules for \mvu processes and
configurations; reduction on terms is standard $\beta$-reduction.
Reduction on configurations is defined modulo the associativity and
commutativity of parallel composition.

\subparagraph{Diffing.}
As DOM pages include event queues, they contain strictly more information
than HTML. To avoid losing pending events, we require a diffing operation.
Define $\erase{D}$ as the operation $\erase{\pgtag{t}{U}{D}{\seq{e}}} =
\coretag{t}{U}{(\erase{D})}$, with the other cases defined recursively.  DOM
pages can be modified by adding a node with an empty queue, removing a node, or
updating a node's attributes.  We define operation $\diff{U}{D} = D'$ if
$\erase{D'} = U$, and $D'$ is obtained from $D$ by the minimum number of
insertions and deletions.

\subparagraph{Semantics by example.}
Let us return to our original example from~\secref{sec:introduction}: a box
and a text node displaying the reversed box contents.
We reuse the \mkwd{view} and \mkwd{update} functions and let $V_m =
(\var{contents} = \textnode{``''})$, $V_v = \mkwd{view}$, and $V_u =
\mkwd{update}$.
We extend the HTML syntactic sugar to pages, letting $\translate{-}$ be a
desugaring function and
$\translate{\htmltag{t}{\seq{a} \qsep \seq{e}}{\seq{\domh}}} =
\pgtag{t}{\translate{\seq{a}}}{\translate{\seq{\domh}}}{\seq{e}}$.

We write $\mathcal{R}^+$ for the transitive closure of a relation $\mathcal{R}$.
We begin by supplying the model, view, and update parameters to an initial
configuration.
By \textsc{E-Run}, we get an event loop process, and then term $V_v \app V_m$
reduces to the initial rendered HTML.  By diffing against the empty page, we
display the initial DOM page (\textsc{E-Update}).

{\footnotesize
  \[
    \begin{array}{l}
      \quad \sys{\runexp{V_m}{V_v}{V_u}}{\pgempty} \vspace{\redrowskip} \\
      \ceval (\textsc{E-Run})
      \quad \sys{\handlerprocexp{(V_m, V_v \app V_m)}{V_v}{V_u}}{\pgempty} \vspace{\redrowskip}  \\
      \teval^+ \vspace{-0.5em} \\
      \quad \sys{\handlerprocexp{(V_m,
        \bl
          \opentag{input} \: \attribute{type}=\textnode{``text''} \:\:
          \attribute{value}=\textnode{``''} \\
            \quad  \attribute{onInput} = \{ \lambda \var{str}.
              \mkwd{UpdateBox}(\var{str})\} \closetag \tagzeroend{input}  \\
               \htmltagzero{div}{}
        \el
      )}{V_v}{V_u}}{\pgempty} \vspace{\redrowskip} \\
      \ceval (\textsc{E-Update}) \vspace{-0.5em} \\
        \quad \sys{\handlerprocexp{\idle{V_m}}{V_v}{V_u}}{
          {
          \bl
            \opentag{input} \: \attribute{type}=\textnode{``text''} \:\:
            \attribute{value}=\textnode{``''} \\
              \quad  \attribute{onInput} = \{ \lambda \var{str}.
                \mkwd{UpdateBox}(\var{str})\} \qsep \epsilon \closetag \tagzeroend{input}  \\
                \htmltagqueue{div}{\epsilon}{}
          \el
        }
        }
      \end{array}
    \]%
  }

  \noindent
  The system now does not reduce until a user interacts with the text box and
presses the $k$ key, modelled by \textsc{E-Interact}.
At this point, the event queue for the \tagname{input} box receives
four events: \evt{click}, \evt{keyDown}, \evt{keyUp}, and \evt{input}, which are
are processed by rule \textsc{E-Evt}. The \tagname{input} element does
not have handlers for the \evt{click}, \evt{keyDown}, and \evt{keyUp} events, so
no processes are spawned, but \emph{does} contain an \attribute{onInput} handler,
which handles the \evt{input} event by spawning
$\thread{\mkwd{UpdateBox}(\textnode{``k''})}$.%

  {\footnotesize
    \[
      \bl
      \ceval^+ (\textsc{E-Interact}) \vspace{-1em} \\
      \quad
        \sys{\handlerprocexp{\idle{V_m}}{V_v}{V_u}}{
      \bl
      \opentag{input} \: \attribute{type}=\textnode{``text''} \:
      \attribute{value}=\textnode{``''} \\
              \quad \attribute{onInput} = \{ \lambda \var{str}.
              \mkwd{UpdateBox}(\var{str})\} \qsep \evtpayload{click}{()} \cdot \\
              \quad \evtpayload{keyDown}{75} \cdot \evtpayload{keyUp}{75} \cdot
              \evtpayload{input}{\textnode{``k''}}\closetag \tagzeroend{input} \\
              \htmltagqueue{div}{\epsilon}{}
      \el
      } \vspace{\redrowskip} \\
      \ceval^+ (\textsc{E-Evt}) \vspace{-1em} \\
      \quad
        \sys{\handlerprocexp{\idle{V_m}}{V_v}{V_u}}{
        \bl
        \opentag{input} \: \attribute{type}=\textnode{``text''} \:
        \attribute{value}=\textnode{``''} \\
                \quad \attribute{onInput} = \{ \lambda \var{str}.
                \mkwd{UpdateBox}(\var{str})\} \qsep
                  \evtpayload{input}{\textnode{``k''}}\closetag
                  \\ \tagzeroend{input} \\
                \htmltagqueue{div}{\epsilon}{}
        \el
      } \vspace{\redrowskip} \\
      \ceval (\textsc{E-Evt}) \vspace{-1em} \\
      \quad
        \sys{\handlerprocexp{\idle{V_m}}{V_v}{V_u} \parallel
          \thread{\mkwd{UpdateBox}(\textnode{``k''})}}{
        \bl
        \opentag{input} \: \attribute{type}=\textnode{``text''} \:
        \attribute{value}=\textnode{``''} \\
                \quad \attribute{onInput} = \{ \lambda \var{str}.
                \mkwd{UpdateBox}(\var{str})\}
                \\ \qsep \epsilon \closetag \tagzeroend{input} \\
                \htmltagqueue{div}{\epsilon}{}
        \el
      }
      \el
    \]%
    }

    Since $\mkwd{UpdateBox}(\textnode{"k"})$ is already a value and the event
    loop process is idle, we can process the message (\textsc{E-Handle}).  The
    \mkwd{handle} meta-function calculates a new model $m'$ by applying the
    \mkwd{update} function to a pair of the previous model and the message,
calculates a new HTML value $v'$ by applying the \mkwd{view} function to $m'$,
and returns the pair $(m', v')$.
Finally, the page is diffed against the previous DOM page to generate a new DOM
page $D'$, and the event loop process reverts to being idle:

{\footnotesize
      \[
      \bl
      \ceval (\textsc{EP-Handle}) \vspace{-0.75em} \\
      \quad
      \sys{\handlerprocexp{\handle{V_m, (V_v, V_u), \mkwd{UpdateBox}(\textnode{``k''})}}{V_v}{V_u}}{
        \bl
        \opentag{input} \: \attribute{type}=\textnode{``text''} \:
        \attribute{value}=\textnode{``''} \\
                \quad \attribute{onInput} = \{ \lambda \var{str}.
                \mkwd{UpdateBox}(\var{str})\}
                \\ \qsep \epsilon \closetag \tagzeroend{input} \\
                \htmltagqueue{div}{\epsilon}{}
        \el
      } \vspace{\redrowskip} \\
      \teval^+
\vspace{\redrowskip}
      \\
      \quad
      \sys{\handlerprocexp{(
    \bl
      (\var{contents}=\textit{``k''}), \\
       \quad \opentag{input} \: \attribute{type}=\textnode{``text''} \:\:
         \attribute{value}=\textnode{``k''} \\
        \qquad  \attribute{onInput} = \{ \lambda \var{str}.
           \mkwd{UpdateBox}(\var{str})\} \closetag \\
        \quad \tagzeroend{input}  \\
           \quad  \htmltagzero{div}{k} \\
    \el
        )}{V_v}{V_u}}{
        \bl
        \opentag{input} \: \attribute{type}=\textnode{``text''} \:
        \attribute{value}=\textnode{``''} \\
                \quad \attribute{onInput} = \{ \lambda \var{str}.
                \mkwd{UpdateBox}(\var{str})\}
                \\ \qsep \epsilon \closetag \tagzeroend{input} \\
                \htmltagqueue{div}{\epsilon}{}
        \el
      } \vspace{\redrowskip} \\
      \ceval (\textsc{E-Update}) \\
      \quad
      \sys
      {\handlerprocexp{\idle{(\textit{contents}=\textit{\textnode{``k''}})}}{V_v}{V_u}}
{ \bl
   \opentag{input} \: \attribute{type}=\textnode{``text''} \:\:
   \attribute{value}=\textnode{``k''} \\
     \quad  \attribute{onInput} = \{ \lambda \var{str}.
       \mkwd{UpdateBox}(\var{str})\} \qsep \epsilon \closetag \tagzeroend{input}  \\
       \htmltagqueue{div}{\epsilon}{\textnode{k}}
\el
}
    \el
  \]%
}%

\subsection{Metatheory}\label{sec:metatheory}

\subparagraph{Runtime typing.}
\begin{figure}[t]
  {\footnotesize
    \begin{minipage}[t]{0.325\textwidth}
      ~\textbf{Typing rules for events} \hfill \framebox{$\vdash e$}
\begin{mathpar}
  \inferrule
  [TE-Evt]
  { \cdot \vdash V : \evtty{\evt{ev}} }
  { \vdash \evtpayload{\evt{ev}}{V} }
\end{mathpar}
\end{minipage}
\hfill
\begin{minipage}[t]{0.675\textwidth}
  ~\textbf{Typing rules for active threads}
\hfill
\framebox{$\vdashs T :
\evtloopty{A}{B}$}
  \begin{mathpar}
    \inferrule
    [TS-Idle]
    { \cdot \vdash V_m : A }
    { \vdash \idle{V_m} : \evtloopty{A}{B} }

    \inferrule
    [TS-Processing]
    { \cdot \vdash M : (A \times \htmlty{B}) }
    { \vdash M : \evtloopty{A}{B} }
  \end{mathpar}
\end{minipage}
\vspace{1em}

~\textbf{Typing rules for processes and configurations} \hfill
\framebox{$\vdash^\phi P {:} A$}~\framebox{$\vdash\vphantom{\vdash^\phi} \config{C}$}
    \begin{mathpar}
    \inferrule
    [TP-Run]
    { \cdot \vdash M : (A \times (A \to \htmlty{B}) \times ((B \times A) \to A))  }
    { \vdash^\bcirc \run{M} : B }

    \inferrule
    [TP-EventLoop]
    { \vdash T {:} \evtloopty{A}{B} \\\\
      \cdot \vdash V_v {:} A \to\! \htmlty{B} \\
      \cdot \vdash V_u {:} (B \times A) \to A
    }
    { \vdash^\bcirc \handlerprocexp{T}{V_v}{V_u} : B }

    \inferrule
    [TP-Thread]
    { \cdot \vdash M {:} A }
    { \vdash^\wcirc \thread{M} {:} A }

    \inferrule
    [TP-Par]
    { \vdash^{\phi_1} P_1 {:} A \\ \vdash^{\phi_2} P_2 {:} A }
    { \vdash^{\phi_1 + \phi_2} P_1 \parallel P_2 {:} A }

    \inferrule
    [TC-System]
    { \vdash^{\bcirc} P : A \\ \vdash D : \pagety{A} }
    { \vdash \sys{P}{\vh} }
  \end{mathpar}

  ~\textbf{Combination of flags} \hfill \framebox{$\phi_1 + \phi_2$}
  \begin{mathpar}
    \wcirc + \wcirc = \wcirc

    \wcirc + \bcirc = \bcirc

    \wcirc + \bcirc = \bcirc

    \bcirc + \bcirc \text{ undefined}
  \end{mathpar}
  \vspace{-0.75em}
}
  \caption{Runtime typing for \mvu}
  \label{fig:formalism:runtime-typing}
\end{figure}
 To reason about the metatheory, we require runtime typing rules, shown in
Figure~\ref{fig:formalism:runtime-typing}.
Judgement $\vdash e$ states that the payload of an event $e$
has the payload type specified by its signature.
Judgement $\vdashs T : \evtloopty{A}{B}$ can be read ``Active thread $T$ has model type $A$ and message type $B$''. An idle
thread $\idle{V_m}$ has type $\evtloopty{A}{B}$ if $V_m$ has type $A$ (\textsc{TS-Idle}).
An active thread $M$ currently processing a message has type $\evtloopty{A}{B}$
if $M$ has type $(A \times \htmlty{B})$, i.e., computes a pair of a new model
with type $A$ and HTML which produces messages of type $B$
(\textsc{TS-Processing}).

Judgement $\vdash^\phi P : A$ states that process $P$ is well typed and produces
or consumes messages of type $A$.  The parallel composition of two processes
$P_1 \parallel P_2$ has message type $A$ if both $P_1$ and $P_2$ have message
type $A$ (\textsc{TP-Par}).  An event handler process $\thread{M}$ has message
type $A$ if term $M$ has type $A$ (\textsc{TP-Thread}).

An initialisation process $\run{M}$ is well-typed if $M$ is a product type where
each component has the correct model, view, and update types.  An event loop
process $\handlerprocexp{T}{V_v}{V_v}$ has message type
$B$ if its active thread $T$ has model type $A$ and message type $B$; its view
function $V_v$ has type $A \to \htmlty{B}$; and its update function has type $(B
\times A) \to A$ (\textsc{TP-EventLoop}). Thread flags $\phi$ ensure
that there is precisely one initialisation process or event loop process in a
process typeable under flag $\bcirc$.

Judgement $\vdash \config{C}$ states that configuration
$\config{C}$ is well-typed: a system configuration $\sys{P}{\vh}$ is
well-typed if process $P$ has precisely one event loop process with message type
$A$ and page $D$ has type $\pagety{A}$.
The omitted typing rules for pages (of shape $\vdash D : \pagety{A}$) are
similar to those for terms of type $\htmlty{A}$.

Note that we consider only closed configurations and processes since
it makes little sense for DOM pages $D$ to contain free variables, and because
processes do not bind variables.

We are now well-placed to state some formal results.
We omit proofs in the main
body of the paper; full proofs can be found in the extended
version~\cite{Fowler20:mvuc-extended}.

\subparagraph{Preservation.}
Reduction preserves typing.

\begin{theorem}[Preservation]\label{thm:config-pres}
  If $\vdash \config{C}$ and $\config{C} \ceval \config{C}'$, then
  $\vdash \config{C}'$.
\end{theorem}

\subparagraph{Progress.}
The system vacuously satisfies a progress property as it can always
reduce by \textsc{E-Interact} due to user input. It is more interesting to
consider the \emph{event progress} property enjoyed by the system \emph{without}
\textsc{E-Interact}: either there are no events to process and the
system is idle, or the system can reduce.
Functional reduction satisfies progress.
\begin{lemma}[Progress (Terms)]\label{lem:term-progress}
If $\cdot \vdash M : A$, then either $M$ is a value, or there exists some $N$ such that $M \teval N$.
\end{lemma}

Let $\cevalminus$ be the relation $\ceval$ without rule \textsc{E-Interact}.
The concurrent fragment of the language will reduce until all event handler
threads have finished evaluating, and there are no more messages to process.
By appeal to Lemma~\ref{lem:term-progress}, we
can show event progress.

\begin{theorem}[Event Progress]\label{thm:event-progress}
  If $\vdash \config{C}$, either:
  \begin{enumerate}
    \item there exists some $\config{C}'$ such that $\config{C} \cevalminus
      \config{C'}$; or
    \item $\config{C} = \sys{\handlerprocexp{\idle{V_m}}{V_v}{V_u}}{\vh}$ where $\vh$
      cannot be written $\config{D}[\pgtag{t}{V}{W}{\seq{e}}]$ for some
      non-empty $\seq{e}$.
  \end{enumerate}
\end{theorem}

\section{\mvu with Session Types}\label{sec:extensions}
In this section, we extend \mvu to support session types.
We require three extensions: \emph{commands}, to perform side-effects;
\emph{linearity}, to implement session types safely; and \emph{transitions}, to
allow multiple model and message types.
We begin by showing each extension by example, and show the extended formalism
in~\secref{sec:extensions:combined}.

\subsection{Commands}\label{sec:extensions:commands}
Real-world applications require side-effects. To this end, Elm supports
\emph{commands} which describe side-effects to be performed in the event loop.
Although commands in Elm are more general, for our purposes, it is particularly
useful to be able to spawn a process which will run concurrently and eventually
return a message.
As an example, we may want to await the result of an expensive computation,
and display the result when the computation completes. Letting
$\textsf{na\"iveFib}(x)$ be the na\"ive Fibonacci function and
assuming an \mkwd{intToString} function, we can write:

{\footnotesize
\[
\bl
\mkwd{Model} \defeq \mkwd{Maybe}(\intty) \qquad
\mkwd{Message} \defeq \mkwd{StartComputation} \midspace \mkwd{Result}(\intty) \vspace{0.35em} \\
\mkwd{view} : \mkwd{Model} \to \htmlty{\mkwd{Message}} \\
\mkwd{view} = \lambda \var{model} . \calcwd{html} \\
\quad \tagzero{html} \\
\qquad \tagzero{body} \\
\qquad \quad \{
  \caseofone{model} \{ \\
  \qquad \qquad
  {\begin{aligned}[t]
    & \quad \mkwd{Just}(\textit{result}) \mapsto \htmltext{\intstr{\var{result}}}; \\
    & \quad \mkwd{Nothing} \mapsto \htmltext{\textnode{``Waiting \ldots''}}  \: \}
  \quad \} \\
  \end{aligned}} \\
  \qquad \quad \htmltag{button}{\attribute{onClick} = \antiquote{\lambda () .
  \mkwd{StartComputation}} }{\textnode{Start!}} \\
\qquad \tagzeroend{body} \\
\quad \tagzeroend{html} \vspace{0.35em}
\el
\]%
\[
  \bl
\mkwd{update} : (\mkwd{Message} \times \mkwd{Model}) \to (\mkwd{Model},
\cmdty{\mkwd{Message}})\\
\mkwd{update} = \lambda (\textit{message}, \textit{model}) . \\
\quad \caseofone{\textit{message}} \{ \\
  \qquad \mkwd{StartComputation} \mapsto (\mkwd{Nothing},
  \cmdspawn{\mkwd{Result}(\textsf{na\"iveFib}(1000))}) \\
  \qquad \mkwd{Result}(x) \mapsto (\mkwd{Just}(x), \cmdempty) \\
  \quad \}
\el
\]%
}%
The model is of type $\mkwd{Maybe}(\intty)$, with value $\mkwd{Just}(V)$ for
some integer value $V$ if the result has been computed, or $\mkwd{Nothing}$ if
the application is awaiting the result. The $\mkwd{Message}$ type is a variant
type consisting of $\mkwd{StartComputation}$ which is sent to start the
computation, and $\mkwd{Result}(\intty)$, which is sent to return a result.
The \mkwd{view} function renders either the result, or \textnode{``Waiting\ldots''} if no result
is available.

The type of the \mkwd{update} function is changed to return a \emph{pair} of an updated
model and a command. In our case, the $\mkwd{StartComputation}$ message results
in a pair of $\mkwd{Nothing}$ and
$\cmdspawn{\mkwd{Result}(\textsf{na\"iveFib}(1000))}$, which
spawns $\mkwd{Result}(\textsf{na\"iveFib}(1000))$ to evaluate in a separate thread. When the
function (eventually) completes, the thread will have evaluated to a
$\mkwd{Result}$ message, which can be processed by the $\mkwd{update}$ function
to update the model and display the result.

\subsection{Linearity}\label{sec:extensions:linearity}
As we showed in~\secref{sec:introduction}, safely implementing session types
requires linearity: we therefore require linear model and
message types.
Linearity would also prove useful for other linear resources such as functional
arrays with in-place update~\cite{Wadler90:linear-types}.
Unfortunately, \mvu as defined so far does not support linearity.
Consider $\mkwd{handle}$:
{\small
\[
\begin{array}{l}
    \handle{\textit{m}, (\textit{v}, \textit{u}), \textit{msg}} \defeq
    \begin{aligned}[t]
      & \letin{m'}{u \app (\textit{msg}, m)}{({\color{red}{m'}}, v \app {\color{red}{m'}})}
    \end{aligned}
  \end{array}
\]%
}%
The updated model, $m'$, is used non-linearly as it is returned for use in
subsequent requests, and also used to render the model as HTML.

\subparagraph{Extraction.}
Linear resources are needed only when \emph{updating}
the model---not when rendering the webpage---as we will not need to communicate on
session channels when rendering. If the developer implements a function:
{\small
\[
  \mkwd{extract} : A \to (A \times B) %
\]%
}
\noindent
where $A$ is the type of a model, and $B$ is the \emph{unrestricted} fragment of the
model, we can restore linear usage of the model (letting $e$ be the extraction function):
{\small
\[
\begin{array}{l}
  \handle{\textit{m}, (\textit{v}, \textit{u}, \textit{e}), \textit{msg}} \defeq
  \begin{array}[t]{l}
      \letintwo{m'}{u \app (\textit{msg}, m)} \\
      \letin{(m', \textit{unrM})}{e \app m'}{(m', v \app \textit{unrM})}
    \el
  \end{array}
\]%
}
\noindent
  An alternative approach would be to assign the view function type
  $A \to (A \times \htmlty{B})$, returning the linear model and allowing it to
  be re-bound. We would need to modify $\mkwd{handle}$:
  {\small
  \[
    \begin{array}{l}
    \handle{\textit{m}, (\textit{v}, \textit{u}), \textit{msg}} \defeq
    \begin{aligned}[t]
      & \letin{m'}{u \app (m, \textit{msg})}{v \app m'} \\
    \end{aligned}
  \end{array}
  \]%
}%
\noindent
  A key disadvantage of this approach is that rendering is no longer a
  read-only operation, breaking an important abstraction barrier.

\subparagraph{Example.}
We can now write our first session-typed \mvu application.
Our web client consists of a button which, when clicked, triggers the sending of
a \mkwd{Ping} message to the server. Once clicked, the button is disabled. The
server then receives the \mkwd{Ping} message and responds with a \mkwd{Pong}
message; upon receiving the response, the client then re-enables the button.

\begin{center}
\begin{minipage}[t]{0.25\textwidth}
  \textsf{Pinging}:

  \includegraphics[scale=1.25]{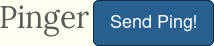}
\end{minipage}
\qquad
\begin{minipage}[t]{0.25\textwidth}
  \textsf{Waiting}:

  \includegraphics[scale=1.25]{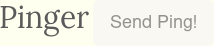}
\end{minipage}
\end{center}

\subparagraph{Simply-typed channels.}
\begin{figure}
  {\footnotesize
    \[
  \mkwd{Model} \defeq (\boolty \times \simplechan{\mkwd{Ping}} \times \simplechan{\mkwd{Pong}}) \qquad \qquad
  \mkwd{Message} \defeq \mkwd{Click} \midspace \mkwd{Ponged}
    \]%
    \vspace{-3em}

\begin{minipage}[t]{0.4\textwidth}
  \[
  \bl
  \mkwd{view} : \mkwd{Model} \to \htmlty{\mkwd{Message}} \\
  \mkwd{view} \defeq \lambda (\var{pinging}, \_, \_) . \: \\
  \quad \calcwd{let} \: {\var{attr}} = \\
  \qquad \calcwd{if} \: \var{pinging} \: \calcwd{then} \\
  \qquad \quad \attrempty \\
  \qquad \calcwd{else} \\
  \qquad \quad  \attr{\textnode{``disabled''}}{\textnode{``true''}} \: \calcwd{in} \\
  \quad \calcwd{html} \\
    \qquad \tagzero{html} \\
    \quad \qquad \tagzero{body} \\
    \qquad \qquad \opentag{button} \: \{ \var{attr} \} \: {\attribute{onClick} = \antiquote{\lambda () . \mkwd{Click}}} \closetag \\
    \quad \qquad \qquad \textnode{Send Ping!} \\
    \qquad \qquad \tagzeroend{button} \\
    \qquad \quad \tagzeroend{body} \\
    \qquad \tagzeroend{html} \vspace{0.4em} \\
    \el
  \]
\end{minipage}
\hfill
\begin{minipage}[t]{0.5\textwidth}
  \[
    \bl
    \mkwd{update} : (\mkwd{Message} \times \mkwd{Model}) \to \mkwd{Model} \\
    \mkwd{update} \defeq \lambda (\var{msg}, (\_, \var{pingCh}, \var{pongCh})) . \\
    \quad \caseofone{\var{msg}} \: \{  \\
    \qquad \mkwd{Click} \mapsto \\
    \qquad \quad \calcwd{let} \: {\var{cmd}} = \\
    \qquad \qquad {\cmdspawn{(
        {\begin{array}[t]{l}
          \gvsend{\mkwd{Ping}}{\var{pingCh}}; \\
          \letintwo{\mkwd{Pong}}{\gvrecv{\var{pongCh}}} \\
          \mkwd{Ponged}) \: \calcwd{in}
          \end{array}}}} \\
    \quad \qquad ((\ffalse, \var{pingCh}, \var{pongCh}), \var{cmd}) \\
    \qquad \mkwd{Ponged} \mapsto  ((\ttrue, \var{pingCh}, \var{pongCh}), \cmdempty) \\
   \quad \} \\ \\
   \mkwd{server} : (\simplechan{\mkwd{Ping}} \times \simplechan{\mkwd{Pong}}) \to (\one \to A) \\
   \mkwd{server} \defeq \lambda (\var{pingCh}, \var{pongCh}) . \\
  \quad (\rectwo{f}{} \\
  \qquad \letintwo{\mkwd{Ping}}{\gvrecv{\var{pingCh}}} \\
  \qquad \gvsend{\mkwd{Pong}}{\var{pongCh}}; \app f \app ()) \\
  \el
  \]
\end{minipage}
}
\caption{\mkwd{PingPong} application using simply-typed channels}
\label{fig:extensions:pingpong-simple}
\end{figure}
 Before considering a session-typed version of the application,
it is instructive to consider a version \emph{without} session typing, shown in
Figure~\ref{fig:extensions:pingpong-simple}. Let $\simplechan{A}$ be
the type of a simply-typed channel over which one can send and receive values of
type $A$. The model is a 3-tuple containing a Boolean value which is true when
waiting for the user to click the ``Send Ping!'' button, and false when waiting
for a response; a channel for \mkwd{Ping} messages; and a channel for
\mkwd{Pong} messages. There are two types of UI message: \mkwd{Click} denotes
that the button has been clicked, and \mkwd{Ponged} denotes that a \mkwd{Pong}
message has been received along the \mkwd{Pong} channel.

The \mkwd{view} function displays the page, adding the \verb+disabled+ attribute
to the button
if we are waiting for a \mkwd{Pong} message. The \mkwd{update} function
case-splits on the UI message: in the case of a \mkwd{Click} message raised by
the button, the model is updated to set the \var{pinging} flag to \ffalse, and
the function creates a command to send a \mkwd{Ping} message along \var{pingCh},
receive a \mkwd{Pong} message from \var{pongCh}, and return a \mkwd{Ponged} UI
message. In the case of a \mkwd{Ponged} message, the model is updated to set
the \var{pinging} flag to \ttrue, enabling the button again.
The \mkwd{server} function models a server thread, which repeatedly receives
\mkwd{Ping} messages from \var{pingCh} and sends \mkwd{Pong} messages to
\var{pongCh}.

Even in this simple example, it is very easy to communicate incorrectly: if the
client neglected to send a \mkwd{Ping} message before trying to receiving a
\mkwd{Pong} message along \var{pongCh}, then the command would hang forever and
the GUI would never re-enable the button. A similar situation would arise if the
server received the \mkwd{Ping} message but failed to respond.%

\subparagraph{Session types.}
Session types $S$ range over output $\gvout{A}{S}$, input $\gvin{A}{S}$, the
completed session $\gvend$, recursive session types $\recty{\rtv{t}}{S}$, and
(possibly dualised) recursive type variables $\rtv{t}$. We take an equi-recursive
treatment of recursive session types, identifying a recursive session type with
its unfolding.
We omit types and constructs for branching and selection as they can be
encoded~\cite{Kobayashi02:type-systems, DardhaGS17:revisited}.  The
$\calcwd{send}$ constant sends a value of type $A$ over an endpoint of type
$\gvout{A}{S}$ and returns the continuation of the session, $S$.
The $\calcwd{close}$ constant closes a completed session endpoint.  The
$\calcwd{receive}$ constant takes an endpoint of type $\gvin{A}{S}$ and receives
a pair of a value of type $A$ and endpoint of type $S$.
The $\calcwd{cancel}$ constant allows an endpoint to be
discarded safely~\cite{MostrousV18:affine, FowlerLMD19:stwt}.

\vspace{-1em}
{\small
\begin{mathpar}
\begin{array}{lrcl}
    \text{Session types} & S & ::= &
    \gvout{A}{S} \midspace \gvin{A}{S} \midspace \recty{t}{S} \midspace t \midspace \gvdual{t} \midspace \gvend
  \end{array} \\

\calcwd{send} : (A \times \gvout{A}{S}) \to S

  \calcwd{receive} : \gvin{A}{S} \to (A \times S)

  \calcwd{close} : \gvend \to \one

  \calcwd{cancel} : S \to \one
\end{mathpar}
}%

\begin{figure}[t]
{\footnotesize
\[
  \begin{array}{l@{\qquad\qquad}l}
  \mkwd{PingPong} \defeq \mu t . {!}\mkwd{Ping} . {?}\mkwd{Pong} . t &
  \mkwd{Model} \defeq \mkwd{Pinging}(\mkwd{PingPong}) \midspace \mkwd{Waiting} \\
  \mkwd{UModel} \defeq \mkwd{UPinging} \midspace \mkwd{UWaiting} &
  \mkwd{Message} \defeq \mkwd{Click} \midspace \mkwd{Ponged}(\mkwd{PingPong})
  \vspace{-2em}
  \end{array}
\]

\begin{minipage}[t]{0.4\textwidth}
\[
  \bl
  \mkwd{view} : \mkwd{UModel} \to \htmlty{\mkwd{Message}} \\
  \mkwd{view} \defeq \lambda \var{uModel} . \\
  \quad \calcwd{let} \: {\var{attr}} = \\
  \qquad \caseofone{\var{uModel}} \{ \\
    \qquad \quad \mkwd{UPinging} \mapsto \attrempty \\
    \qquad \quad \mkwd{UWaiting} \mapsto
    \attr{\textnode{``disabled''}}{\textnode{``true''}} \\
  \qquad \} \: \calcwd{in} \\
  \quad \calcwd{html} \\
  \qquad \tagzero{html} \\
  \qquad \quad \tagzero{body} \\
  \qquad \qquad \opentag{button} \:
    \antiquote{\var{attr}} \: \attribute{onClick} = \antiquote{\lambda () . \mkwd{Click}}
    \closetag \\
  \qquad \qquad \quad \textnode{Send Ping!} \\
  \qquad \qquad \tagzeroend{button} \\
  \qquad \quad \tagzeroend{body} \\
  \qquad \tagzeroend{html} \\ \\
  \mkwd{handleClick}(\var{model}) \defeq \\
  \quad \caseofone{\var{model}} \{ \\
    \qquad \mkwd{Pinging}(c) \mapsto \\
    \quad \qquad \calcwd{let} \: {\var{cmd}} = \\
    \quad \qquad \quad {\cmdspawn{(
        {\begin{array}[t]{l}
         \letintwo{c}{\gvsend{\mkwd{Ping}}{c}} \\
         \letintwo{(\var{pong}, c)}{\gvrecv{c}} \\
         {\mkwd{Ponged}(c)}) \: \calcwd{in}
    \end{array}}}}  \\
    \qquad \quad (\mkwd{Waiting}, \var{cmd}) \\
    \qquad \hlred{\mkwd{Waiting} \mapsto (\mkwd{Waiting}, \cmdempty)} \\
   \quad \}
  \el
\]
\end{minipage}
\hfill
\begin{minipage}[t]{0.45\textwidth}
 \[
   \bl
   \mkwd{update} : (\mkwd{Message} \times \mkwd{Model}) \to \\
   \qquad (\mkwd{Model} \times \cmdty{\mkwd{Message}}) \\
  \mkwd{update} \defeq  \lambda (\textit{msg}, \textit{model}) . \\
  \quad \caseofone{\textit{msg}} \{ \\
    \qquad \mkwd{Click} \mapsto \\
    \qquad \quad \mkwd{handleClick}(\var{model}) \\
    \qquad \mkwd{Ponged}(c) \mapsto \\
    \qquad \quad \mkwd{handlePonged}(\var{model}, c) \\
      \quad \} \\ \\
  \mkwd{extract} : \mkwd{Model} \to (\mkwd{Model} \times \mkwd{UModel}) \\
  \mkwd{extract} \defeq \lambda \var{model} . \\
  \quad \caseofone{\var{model}} \{ \\
    \qquad \mkwd{Pinging}(c) \mapsto (\mkwd{Pinging}(c), \mkwd{UPinging}) \\
    \qquad \mkwd{Waiting} \mapsto (\mkwd{Waiting}, \mkwd{UWaiting}) \\
  \quad \} \\ \\
\mkwd{handlePonged}(\var{model}, c) \defeq \\
    \quad \caseofone{\var{model}} \{ \\
      \qquad \hlred{\mkwd{Pinging}(c') \mapsto} \\
      \qquad \quad \hlred{\gvcancel{c'};} \\
      \qquad \quad \hlred{(\mkwd{Pinging}(c), \cmdempty)} \\
       \qquad \mkwd{Waiting} \mapsto \\
       \qquad \quad (\mkwd{Pinging}(c), \cmdempty) \\
     \quad \}
   \el
\]
\end{minipage}
}
\caption{$\mkwd{PingPong}$ application}
\label{fig:extensions:pingpong-monolith}
\end{figure}
 Figure~\ref{fig:extensions:pingpong-monolith} shows the \mkwd{PingPong}
client written in \mvu. We can encode the PingPong protocol as a session
type, $\recty{\rtv{t}}{{!}\mkwd{Ping}.{?}\mkwd{Pong}.\rtv{t}}$.
The \mkwd{Model} type encodes the two states of the application:
$\mkwd{Pinging}(c)$ is the state where the \textnode{``Send Ping!''} button is
enabled and the user can send a \mkwd{Ping} message along session channel $c$, whereas \mkwd{Waiting} is
the state where the button is disabled and awaiting a \mkwd{Pong} message from
the other party. The \mkwd{UModel} type is the unrestricted model type which
does not include the session channel.
Again, the \mkwd{Message} type encodes the UI messages in the application: the
\mkwd{Click} UI message is produced when the button is pressed, whereas the
$\mkwd{Ponged}(\mkwd{PingPong})$ UI message is produced when a \mkwd{Pong} session message has
been received. Note that the $\mkwd{Ponged}$ UI message now contains a session
channel of type $\mkwd{PingPong}$ as a parameter.

The \mkwd{view} function takes an unrestricted model and displays a button,
which is disabled in the \mkwd{Waiting} state but enabled
in the \mkwd{Pinging} state. The \mkwd{extract} function pairs the
linear model with the associated unrestricted model.

The \mkwd{update} function case-splits on the message. The
\mkwd{handleClick} function handles the \mkwd{Click} message,
and case-splits on the model.
If the model is in the $\mkwd{Pinging}(c)$ state,
then the function creates a command to spawn a process
which will send a \mkwd{Ping} message along $c$,
receive a \mkwd{Pong} message along $c$, and
generate a \mkwd{Ponged} UI message when the \mkwd{Pong} message is received.
The function finally updates the model to the \mkwd{Waiting} state.
If the model is in the $\mkwd{Waiting}$ state---which should not occur, since
the button is disabled---then the model remains the same and no command is
created.

The \mkwd{handlePonged} function handles a $\mkwd{Ponged}(c)$ message. Again, we
must case split on the model. If the model is in the $\mkwd{Waiting}$ state,
then we can change to the $\mkwd{Pinging}$ state, given endpoint $c$.
However, if the model is in the $\mkwd{Pinging}(c')$ state and a \mkwd{Ponged}
message is received---which should not occur, since according to the
session type, there is no way for the peer to send a \mkwd{Pong} message
while we are waiting to send a \mkwd{Ping}---we now have \emph{two} linear
resources. We choose to discard $c'$ using \calcwd{cancel}, and change the model
to $\mkwd{Pinging}(c')$, but this is an arbitrary choice to satisfy a code path
that must exist, but should never be used.

\floatstyle{plain}
\restylefloat{figure}
\begin{figure}[t]
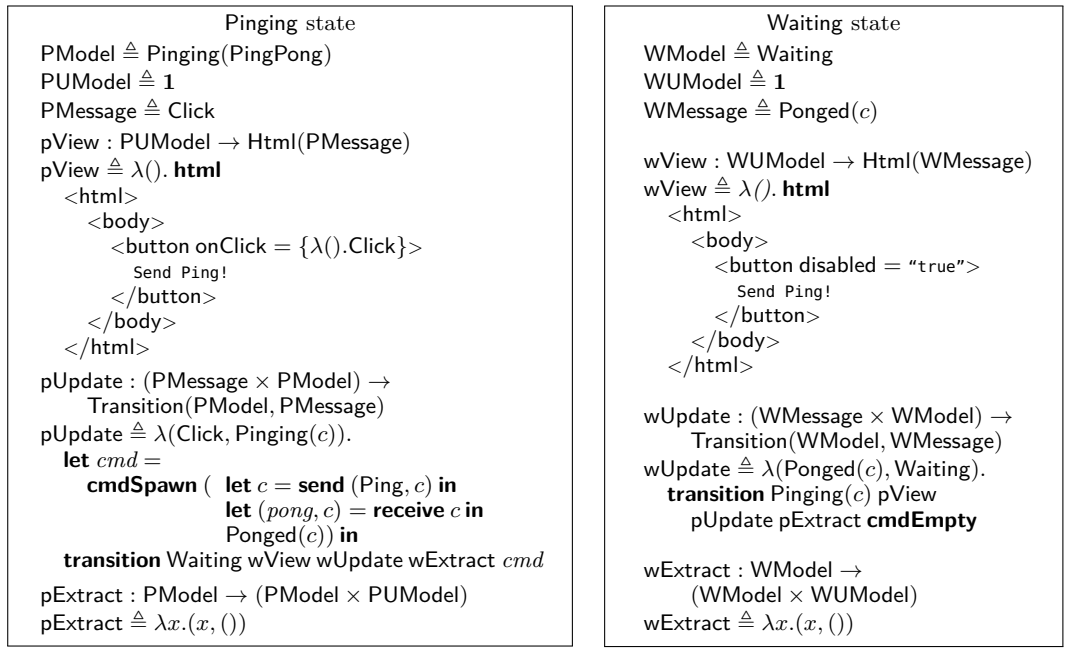

  {\footnotesize
\framebox{
\begin{minipage}[t]{0.5\textwidth}
  \begin{center}
  \mkwd{Pinging} \text{state}
  \end{center}
  \vspace{-1em}
  \[
  \bl
  \mkwd{PModel} \defeq \mkwd{Pinging}(\mkwd{PingPong}) \\
  \mkwd{PUModel} \defeq \one \\
  \mkwd{PMessage} \defeq \mkwd{Click}
  \vspace{0.275em} \\
  \mkwd{pView} : \mkwd{PUModel} \to \htmlty{\mkwd{\mkwd{PMessage}}} \\
  \mkwd{pView} \defeq \lambda () . \: \calcwd{html} \\
    \quad \tagzero{html} \\
    \qquad \tagzero{body} \\
    \qquad \quad \opentag{button} \: {\attribute{onClick} = \antiquote{\lambda () . \mkwd{Click}}} \closetag \\
    \qquad \qquad \textnode{Send Ping!} \\
    \qquad \quad \tagzeroend{button} \\
    \qquad \tagzeroend{body} \\
    \quad \tagzeroend{html} \vspace{0.4em} \\
    \mkwd{pUpdate} : (\mkwd{PMessage} \times \mkwd{PModel}) \to \\
    \qquad \transitionty{\mkwd{PModel}}{\mkwd{PMessage}} \\
\mkwd{pUpdate} \defeq \lambda (\mkwd{Click}, \mkwd{Pinging}(c)) . \\
    \quad \calcwd{let} \: {\var{cmd}} = \\
    \qquad {\cmdspawn{(
        {\begin{array}[t]{l}
         \letintwo{c}{\gvsend{\mkwd{Ping}}{c}} \\
         \letintwo{(\var{pong}, c)}{\gvrecv{c}} \\
         {\mkwd{Ponged}(c)}) \: \calcwd{in}
    \end{array}}}}  \\
    \quad
    \transition{\mkwd{Waiting}}{\mkwd{wView}}{\mkwd{wUpdate}}{\mkwd{wExtract}}{\var{cmd}}
    \vspace{0.4em} \\
    \mkwd{pExtract} : \mkwd{PModel} \to (\mkwd{PModel} \times \mkwd{PUModel}) \\
    \mkwd{pExtract} \defeq \lambda x . (x, ())
  \el
  \]
\end{minipage}
}
\hfill
    \framebox{
\begin{minipage}[t]{0.4\textwidth}
  \begin{center}
    \mkwd{Waiting} \text{state}
  \end{center}
  \vspace{-1em}
  \[
  \bl
  \mkwd{WModel} \defeq \mkwd{Waiting} \\
  \mkwd{WUModel} \defeq \one \\
  \mkwd{WMessage} \defeq \mkwd{Ponged}(c) \\ \\
  \mkwd{wView} : \mkwd{WUModel} \to \htmlty{\mkwd{WMessage}} \\
  \mkwd{wView} \defeq \lambda \var{()} . \: \calcwd{html} \\
  \quad \tagzero{html} \\
  \qquad \tagzero{body} \\
  \quad \qquad \opentag{button} \: \attribute{disabled}=\textnode{``true''}
   \closetag \\
  \quad \qquad \quad \textnode{Send Ping!} \\
  \quad \qquad \tagzeroend{button} \\
  \quad \quad \tagzeroend{body} \\
  \quad \tagzeroend{html} \\ \\
  \mkwd{wUpdate} : (\mkwd{WMessage} \times \mkwd{WModel}) \to \\
  \qquad \transitionty{\mkwd{WModel}}{\mkwd{WMessage}} \\
  \mkwd{wUpdate} \defeq \lambda(\mkwd{Ponged}(c), \mkwd{Waiting}) . \\
  \quad
  \calcwd{transition} \: {\mkwd{Pinging}(c)} \: {\mkwd{pView}} \\
  \qquad {\mkwd{pUpdate}} \: {\mkwd{pExtract}} \: {\cmdempty} \\ \\
    \mkwd{wExtract} : \mkwd{WModel} \to \\ \qquad (\mkwd{WModel} \times \mkwd{WUModel}) \\
  \mkwd{wExtract} \defeq \lambda x . (x, ())
  \vspace{0.01em}
  \el
  \]
\end{minipage}
}
}
\caption{\mkwd{PingPong} application using model transitions}
\label{fig:extensions:pingpong-transitions}
\end{figure}
\floatstyle{boxed}
\restylefloat{figure}
 \subsection{Model transitions}\label{sec:extensions:transitions}
Our proposal is still not quite satisfactory: as we saw with the \mkwd{PingPong}
example, we need to include cases in the \mkwd{update} function which
cannot arise. We highlight these in red. This is even more pronounced when
dealing with linear resources, such as needing to handle a \mkwd{Ponged}
message when waiting to send a \mkwd{Ping}.

The problem is that we are encoding the \mkwd{Model} type using
a sum type, whereas in fact we require \emph{multiple} model types, and a way to
\emph{transition} between them.

\subparagraph{Example.}
Figure~\ref{fig:extensions:pingpong-transitions} shows how we can modify
\mkwd{PingPong} to use multiple model types. The left-hand side of the
figure shows the \mkwd{Pinging} state: the model type consists of the singleton
variant tag $\mkwd{Pinging}(\mkwd{PingPong})$ containing an endpoint of type
$\mkwd{PingPong}$, the unrestricted model is the unit type, and the only message
that the \mkwd{Pinging} state can receive is $\mkwd{Click}$.
The \mkwd{pView} function is similar to before, and the \mkwd{pExtract} function
returns a pair of the current state and the unit value.
The \mkwd{pUpdate} function is more interesting. Given the current state and a
$\mkwd{Click}$ message, the function constructs a command which will send the
\mkwd{Ping} session message, receive the $\mkwd{Pong}$ session message, and then
generate a $\mkwd{Ponged}(c)$ UI message containing the session channel.
The function \emph{transitions} into the
\mkwd{Waiting} state using the $\calcwd{transition}$ primitive, which allows the
developer to specify new model, view, update, extract functions, and a command
to evaluate. The functions for the \mkwd{Waiting} state follow a similar
pattern.
Session types rule out the communication errors besetting the example
in Figure~\ref{fig:extensions:pingpong-simple}, and model transitions eliminate
the redundant code paths arising due to illegal states.

\subsection{\mvu with Commands, Linearity, and Transitions}
\label{sec:extensions:combined}
Commands, linearity, and transitions are the three key ingredients needed to
extend MVU to support models which include session-typed channels. In this
section, we introduce a calculus which combines all three extensions, and prove
that the extended calculus is sound.

\subsubsection{Syntax and Typing}
Figure~\ref{fig:extensions:combined:syntax} shows the syntax of \mvu extended
with commands, linearity, and transitions.

\subparagraph{Types and kinds.}

To support linearity, types are assigned \emph{kinds}, ranged over by $\kappa$.
Types can either be \emph{linear} ($\lin$) or \emph{unrestricted} ($\unr$).
A value of linear type must be used precisely once, whereas a value of
unrestricted type can be used many times.

We modify function types to include a
kind annotation: linear functions may close over linear variables and so must be
used once.
To support commands, we introduce type $\cmdty{A}$ which is the type of a
command which produces messages of type $A$. To support transitions, we
introduce type $\transitionty{A}{B}$ which is parameterised by the
\emph{current} model type $A$ and message type $B$. Finally, we extend types to
include session types $S$ as described in~\secref{sec:extensions:linearity}.

\begin{figure}[t]
\begin{syntax}
\text{Kinds} & \kind & ::= & \lin \midspace \unr \\
\text{Types} & A, B, C & ::= & \one \midspace A \kto{\kind} B \midspace A \times B
  \midspace A + B \midspace \stringty \midspace \intty \midspace S \\
               &      & \midspace & \htmlty{A} \midspace \attrty{A} \midspace
               \cmdty{A} \midspace \transitionty{A}{B} \vspace{\syntaxskip} \\
\text{Session types} & S & ::= & \gvout{A}{S} \midspace \gvin{A}{S} \midspace
\recty{t}{S} \midspace \rtv{t} \midspace \gvdual{\rtv{t}} \midspace \gvend \\ \\
\text{Terms} & L, M, N & ::= &
  x \midspace \lambda x . M \midspace M \app N \midspace K \app M \midspace () \midspace s \midspace n \\
  & & \midspace & (M, N) \midspace \letin{(x, y)}{M}{N} \\
  & & \midspace & \inl{x} \midspace \inr{x} \midspace \caseof{L}{\inl{x} \mapsto M; \inr{y} \mapsto N} \\
  & & \midspace & \coretag{t}{M}{N} \midspace \htmltext{M} \\
  & & \midspace & \attr{\ak}{M} \midspace \attrempty \\
  & & \midspace & \cmdspawn{M} \midspace \cmdempty \midspace \append{M}{N} \\
  & & \midspace & \transition{M_m}{M_v}{M_u}{M_e}{M_c} \midspace
  \notransition{M_m}{M_c} \\
  & & \midspace & \raiseexn \midspace \tryasinotherwise{L}{x}{M}{N} \vspace{\syntaxskip} \\
  \text{Constants} & K & ::= & \calcwd{send} \midspace \calcwd{receive}
  \midspace \calcwd{new} \midspace \calcwd{cancel} \midspace \calcwd{close}
\end{syntax}
\caption{Syntax of extended calculus}
\label{fig:extensions:combined:syntax}
\end{figure}

 \subparagraph{Terms.}
Term $\cmdspawn{M}$ is a command which can spawn term $M$ as a thread, and is
monoidally composable using $\star$ and $\cmdempty$.

There are two terms for transitions: the $\notransition{M_m}{M_c}$ term denotes
that no transition is to occur, and that the model should be updated to $M_m$
and command $M_c$ should be evaluated; and
$\transition{M_m}{M_v}{M_u}{M_e}{M_c}$ denotes that a transition should occur,
with new model $M_m$, view function $M_v$, update function $M_u$, extraction
function $M_e$, and command $M_c$ to be run once the transition has taken place.

To support session typing, we introduce session typing constants, ranged over by
$K$, as described in~\secref{sec:extensions:linearity}. We also introduce
an application form for constants, $K \app M$.

Finally, as discussed in~\secref{sec:extensions:linearity}, it is useful to
be able to explicitly discard (or \emph{cancel}) a session channel. In
particular, cancellation is crucial in order to handle the interplay between
linearity and transitions, as all unprocessed messages (which may contain linear
resources) must be safely discarded when a transition occurs.

Following~\citet{MostrousV18:affine} and Exceptional GV (EGV)
by~\citet{FowlerLMD19:stwt}, if a thread tries to receive from an endpoint whose
peer has been cancelled, an exception is raised ($\raiseexn$). Exceptions can be
handled using the $\tryasinotherwise{L}{x}{M}{N}$ construct, which tries to
evaluate term $L$, and binds the result to $x$ in $M$ if the term evaluates to a
value, and evaluates $N$ if the term raises an exception.

\begin{figure}[t]
  {\footnotesize
~\textbf{Context splitting} \hfill \framebox{$\Gamma = \Gamma_1 + \Gamma_2$}
\begin{mathpar}
    \inferrule
    { }
    { \cdot = \cdot + \cdot}

    \inferrule
    { A :: \unr }
    { \Gamma, x : A = \\\\
      (\Gamma_1, x : A) + (\Gamma_2, x : A)}

    \inferrule
    { }
    { \Gamma_1 + \Gamma_2, x : A = \\\\ (\Gamma_1, x : A) + \Gamma_2}

    \inferrule
    { }
    { \Gamma_1 + \Gamma_2, x : A = \\\\ \Gamma_1 + (\Gamma_2, x : A)}
\end{mathpar}

~\textbf{Modified typing rules for terms} \hfill \framebox{$\Gamma \vdash M : A$}
\begin{mathpar}
    \inferrule
    [T-Var]
    { \Gamma :: \unr }
    { \Gamma, x \oftype A \vdash x \oftype A }

    \inferrule
    [T-Abs]
    { \Gamma, x \oftype A \vdash M \oftype B \\ \Gamma :: \kind }
    { \Gamma \vdash \lambda x . M \oftype A \kto{\kind} B }

    \inferrule
    [T-AppK]
    { \Sigma(K) = A \kto{\unr} B \\ \Gamma \vdash M \oftype A }
    { \Gamma \vdash K \app M \oftype B }

    \inferrule
      [T-Cmd]
      { \Gamma \vdash M {:} A }
      { \Gamma \vdash \cmdspawn{M} {:} \cmdty{A} }

    \inferrule
      [T-CmdEmpty]
      { \Gamma :: \unr }
      { \Gamma \vdash \cmdempty {:} \cmdty{A} }

    \inferrule
      [T-CmdAppend]
      { \Gamma_1 \vdash M {:} \cmdty{A} \\
        \Gamma_2 \vdash N {:} \cmdty{A} }
      { \Gamma_1 + \Gamma_2 \vdash \append{M}{N} {:} \cmdty{A} }

     \inferrule
      [T-Transition]
      { \Gamma_1 \vdash M_m : A \\
        \Gamma_2 \vdash M_v : A \uto \htmlty{B} \\
        \Gamma_3 \vdash M_u : (B \times A) \uto \transitionty{A}{B} \\\\
        \Gamma_4 \vdash M_e : A \uto (A \times C) \\
        \Gamma_5 \vdash M_c : \cmdty{A} \\
        C :: \unr
      }
      { \Gamma_1 + \ldots + \Gamma_5 \vdash
      \transition{M_m}{M_v}{M_u}{M_e}{M_c} : \transitionty{A'}{B'} }

      \inferrule
      [T-EvtAttr]
      { \Gamma \vdash M : \evtty{h} \uto A }
      { \Gamma \vdash \attr{h}{M} : \attrty{A} }

      \inferrule
      [T-NoTransition]
      { \Gamma_1 \vdash M {:} A \\ \Gamma_2 \vdash N {:} \cmdty{B} }
      { \Gamma_1 + \Gamma_2 \vdash \notransition{M}{N} {:} \transitionty{A}{B}  }

      \inferrule
      [T-Try]
      { \Gamma_1 \vdash L {:} A \\\\ \Gamma_2, x {:} A \vdash M {:} B \\ \Gamma_2 \vdash
      N {:} B }
      { \Gamma_1 + \Gamma_2 \vdash \tryasinotherwise{L}{x}{M}{N} {:} B }

      \inferrule
      [T-Raise]
      { \Gamma :: \unr }
      { \Gamma \vdash \raiseexn {:} A }

    (\text{other rules modified to split contexts})
\end{mathpar}

\begin{minipage}[t]{0.55\textwidth}
~\textbf{Typing of constants} \hfill \framebox{$\Sigma(c)$}
\[
  \begin{array}{rcl}
    \constty{\calcwd{send}}  & = &  (A \times \gvout{A}{S}) \uto S \\
    \constty{\calcwd{receive}} & =  & \gvin{A}{S} \uto (A \times S) \\
    \constty{\calcwd{new}}  & =  & \one \uto (S \times \gvdual{S})
  \end{array}
  \]
  \vspace{-0.75em}
  \begin{mathpar}
    \constty{\calcwd{cancel}}   =  S \uto \one

    \constty{\calcwd{close}}  =  \gvend \uto \one
  \end{mathpar}
\end{minipage}
\hfill
\begin{minipage}[t]{0.4\textwidth}
~\textbf{Duality} \hfill \framebox{$\gvdual{S}$}
\begin{mathpar}
  \gvdual{\gvout{A}{S}} = \gvin{A}{\gvdual{S}}

  \gvdual{\gvin{A}{S}} = \gvout{A}{\gvdual{S}}

  \gvdual{\recty{t}{S}} = \mu \rtv{t} . \gvdual{S \{ \gvdual{t} / t \}}

  \gvdual{\gvdual{\rtv{t}}} = \rtv{t}

  \gvdual{\gvend} = \gvend
\end{mathpar}
\end{minipage}
}
\caption{Term typing for extended calculus.}
\label{fig:extensions:combined:typing}
\end{figure}
 \subparagraph{Kinding and subkinding.}
The kinding relation $A :: \kind$ assigns kind $\kappa$ to type $A$; our
formulation is inspired by that of~\citet{Padovani17:context-free}. Base types
and HTML and attribute types are unrestricted. The kind of a function type is
determined by its kind annotation.  Session types are linear. The kinds of
product, sum, command and transition types are determined by the kinds of their
type parameters.  The reflexive \emph{subkinding} rule $\unr \leq \lin$ combined
with the kinding subsumption rule states that if a value can be used many times,
then it can also be treated as only being used once.
We write $\Gamma :: \kappa$ if $A :: \kappa$ for each $x : A \in \Gamma$.

\begin{definition}[Kinding and subkinding]
  We define the \emph{subkinding} relation as the reflexive relation
  on kinds $\leq$ such that $\unr \leq \lin$.
  We define the \emph{kinding} relation $A :: \kind$ as the largest relation
  between types and kinds such that:
  \begin{itemize}
    \item $A :: \kind'$ if $A :: \kind$ and $\kind \leq \kind'$
    \item $S :: \lin$
    \item $A :: \unr$ if $A \in \{ \one, \stringty, \intty, \htmlty{B}, \attrty{B} \}$
    \item $A \kto{\kind} B :: \kind$
    \item $\cmdty{A} :: \kind$ if $A :: \kind$
    \item $C :: \kappa$ if $C \in \{ A \times B, A + B, \transitionty{A}{B} \}$ and
      both $A :: \kappa$ and $B :: \kappa$
  \end{itemize}
\end{definition}
\subparagraph{Term typing.}
Figure~\ref{fig:extensions:combined:typing} shows the typing rules for the
extended calculus.
The splitting relation $\Gamma = \Gamma_1 + \Gamma_2$~\cite{CervesatoP96:llf}
splits a typing context $\Gamma$ into two subcontexts which may share only
unrestricted variables. We support linearity by changing \textsc{T-Var} to only
type a variable in an unrestricted context, and by using the context splitting
judgement when typing subterms. The adaptation of the remaining rules to use
context splitting is standard, so we omit them.

The constant application rule \textsc{T-AppK} types term $K \app M$ and makes
use of the type schema function $\Sigma(K)$ to ensure that the argument $M$ is
of the correct type. Rule \textsc{T-CmdSpawn} assigns term $\cmdspawn{M}$ type
$\cmdty{A}$ if term $M$ has type $A$, and rules \textsc{T-CmdEmpty} and
\textsc{T-CmdAppend} allow commands to be composed monoidally.

Rule \textsc{T-Transition} types a $\calcwd{transition}$ term. The typing rule
ensures that the types of the new model, and view, update and extract functions
are compatible. Note that the type parameters of the $\transitionty{A'}{B'}$
need not match the types of the new model and functions.
Rule \textsc{T-NoTransition} assigns term $\notransition{M}{N}$ type
$\transitionty{A}{B}$ if new model $M$ has type $A$, and $N$ is a command of
type $\cmdty{B}$. Note that in this way, the $\notransition{M}{N}$ term replaces
the standard result of the \lstinline+update+ function.

Rule \textsc{T-Try} types an exception handler: the continuations share
a typing environment, but the success continuation is augmented with the a
variable of the type of the possibly-failing continuation. Finally,
$\raiseexn$ can have any type as is it does not return (\textsc{T-Raise}).

The type and kinding system ensures that the kind of type $A$ determines the kind
of the typing environment needed to type a term of type $A$.

\begin{lemma}\label{lem:combined:env-kinding}
  If $\Gamma \vdash M : A$ and $A :: \kappa$, then $\Gamma :: \kappa$.
\end{lemma}

\subparagraph{Duality.}
The duality relation for session types is standard: output types are dual to
input types; we use a self-dual $\gvend$ type; and we use the formulation of the
duality of recursive session types advocated by~\citet{LindleyM16:bananas}.

\begin{figure}[t]
  {\footnotesize
  ~\textbf{Runtime syntax}
  \begin{syntax}
    \text{Runtime names} & c, d \vspace{\syntaxskip} \\
    \text{Values} & U, V, W & ::= & \cdots \midspace c \midspace \cmdspawn{M}
    \midspace \notransition{V}{W} \vspace{\syntaxskip} \\
                  & & \midspace &  \transition{V_m}{V_v}{V_u}{V_e}{V_c} \vspace{\syntaxskip} \\
    \text{Active thread} & T & ::= & \idle{V_m} \midspace \updating{M} \midspace
    \extracting{V_c}{M} \\
                         & & \midspace & \extractingt{F}{V_c}{M} \midspace
                         \rendering{V_m, V_c}{M}  \\
                         & & \midspace & \transitioningext{V_m}{F}{V_c}{\mh} \vspace{\syntaxskip}\\
    \text{Versions} & \vers \vspace{\syntaxskip} \\
    \text{Processes} & P, Q & ::= & \run{M} \midspace \handlerproctrans{T}{F}{\vers}
                     \midspace \threadtrans{M}{\vers} \midspace P \parallel Q \\
                     & & \midspace & (\nu c d) P \midspace \serverthread{M} \midspace \zap{c} \midspace \halt \vspace{\syntaxskip} \\
    \text{Function state} & F & ::= & \statecomb{V_v}{V_u}{V_e} \vspace{\syntaxskip}\\
    \text{Configurations}  & \config{C} & ::= & \sys{P}{\vh} \vspace{\syntaxskip} \\
    \\
    \text{Process contexts} & \procctx & ::= & [~] \midspace \procctx \parallel P \midspace (\nu c d) \procctx \vspace{\syntaxskip} \\
    \text{Evaluation contexts} & E & ::= & \cdots \midspace K \app E
    \midspace \notransition{E}{M} \midspace \notransition{V}{E} \\
                               & & \midspace &
                               \transition{E}{M_v}{M_u}{M_e}{M_c} \midspace \cdots \midspace \transition{V_m}{V_v}{V_u}{V_e}{E} \\
                               & & \midspace & \tryasinotherwise{E}{x}{M}{N} \vspace{\syntaxskip} \\
    \text{Pure contexts} & E_{\textsf{P}} & ::= & \text{(as $E$, but without exception handling frames)} \vspace{\syntaxskip} \\
    \text{Active thread contexts} & \ta & ::= & \updating{E} \midspace
    \rendering{V_m, V_c}{E} \midspace
                           \extracting{V_c}{E} \\
                                       & & \midspace &
                                       \extractingt{V_c}{F'}{E} \midspace
                                       \transitioningext{V_m}{F'}{V_c}{E} \vspace{\syntaxskip} \\

    \text{Pure active thread contexts} \!\!\!\!\!\!\!\!\!\!\!\!\!\!\!\! & \tp & ::= &
    \text{(as $\ta$, but for pure contexts)} \vspace{\syntaxskip} \\
    \text{Thread contexts} & \config{T} & ::= &
      \run{E} \midspace
      \handlerproctrans{\ta}{F}{\vers} \midspace
      \threadtrans{E}{\vers} \midspace \serverthread{E}
  \end{syntax}
    {\footnotesize
  ~\textbf{Active thread state machine}
      \begin{center}
        \scalebox{0.75}{
\begin{tikzpicture}[->,>=stealth',shorten >=1pt,left,node distance=3cm,
                    semithick]
  \tikzstyle{every state}=[fill=white,draw=black,text=black, ellipse, minimum width=2.5cm]

  \node[state]         (A)                    {$\calcwd{idle}$};
  \node[state]         (B) [right of=A]       {$\calcwd{updating}$};
  \node[state]         (C) [above right of=B, node distance=2cm] {$\calcwd{extracting}$};
  \node[state]         (D) [right of=C]       {$\calcwd{rendering}$};
  \node[state]         (E) [below right of=B, node distance=2cm] {$\calcwd{extractingT}$};
  \node[state]         (F) [right of=E]       {$\calcwd{transitioning}$};

  \path (A) edge              node {} (B)
    (B) edge              node {(No model transition)\phantom{AD}} (C)
        (C) edge              node {} (D)
        (B) edge                node {(Model transition)\phantom{AD}} (E)
        (E) edge              node {} (F);
\end{tikzpicture}
         }
      \end{center}
    }
}

\caption{Runtime syntax for extended calculus}
\label{fig:extensions:combined:rt-syntax}
\end{figure}

\subsubsection{Operational Semantics}
\subparagraph{Runtime syntax.}
Figure~\ref{fig:extensions:combined:rt-syntax} shows the runtime syntax for the combined
calculus. We introduce \emph{runtime names} $c, d$ which identify session
channel endpoints.

The biggest departure is that we require a richer structure on active threads,
which form a state machine based on whether a model transition occurs. The
$\calcwd{idle}$ state is as
before, and the $\calcwd{updating}$ state evaluates the update function.
If there is no model transition, then the thread moves to the
$\calcwd{extracting}$ state to extract the unrestricted model, and the
$\calcwd{rendering}$ state to render the new HTML. If there is a model
transition, then the thread moves to the $\calcwd{extractingT}$ state followed
by the $\calcwd{transitioning}$ state to calculate the new HTML to be displayed
after the transition. Each state records values which are required in later
states: for example, the $\rendering{V_m, V_c}{M}$ state
records the new model $V_m$ and the command to be executed upon updating the
page $V_c$.

We introduce four new types of process. To model client-server communication, we
introduce server processes $\serverthread{M}$ to model a process $M$ running on
the server; the thread to spawn is given as an argument to $\calcwd{run}$.
As an example, we could write a Ponger server process for the $\mkwd{PingPong}$
example, which immediately responds with a $\mkwd{Pong}$ message:

{\small
\begin{minipage}{0.45\textwidth}
\[
  \bl
  \letintwo{(c, s)}{\gvnew{()}} \\
  (\mkwd{Pinging}(c), \mkwd{pView}, \mkwd{pUpdate}, \mkwd{pExtract}, \\
  \quad \cmdempty, \mkwd{ponger}(s))
  \el
\]
\end{minipage}
\hfill
  \begin{minipage}{0.45\textwidth}
\[
  \bl
  \mkwd{ponger}(s) \defeq \lambda (). \\
  \quad (\rectwo{f}{s} \\
  \qquad \letintwo{(\mkwd{Ping}, s)}{\gvrecv{s}} \\
  \qquad \letintwo{s}{\gvsend{\mkwd{Pong}}{s}} \app f \app s) \app s
  \el
\]
\end{minipage}
}

A name restriction $(\nu c d) P$ binds
runtime names $c$ and $d$ in process $P$, following the double-binder
formulation due to~\citet{Vasconcelos12:fundamentals}.
A \emph{zapper thread} $\zap{c}$ denotes
an endpoint $c$ that has been cancelled and cannot be used in future
communications; we write $\zap{V}$ to mean $\zap{c_1} \parallel \cdots
\parallel \zap{c_n}$ for $c_i \in \fn{V}$, where $\fn{V}$ enumerates the free
runtime names in a value $V$, and extend this sugar to evaluation contexts. The
$\halt$ process denotes that the event loop process has terminated due to an
unhandled exception.

We extend evaluation contexts in the standard
way, and introduce a class of \emph{pure contexts} $\ep$, which are
evaluation contexts which do not contain any exception handling frames.

\subparagraph{Versions.}
\emph{Versions} $\vers$ ensure that threads spawned in a previous state do not
deliver incompatible messages.  We annotate event loop processes and event
handler threads with versions: given an event loop
$\handlerproctrans{T}{F}{\vers}$, a thread $\threadtrans{M}{\vers'}$ where
$\vers \ne \vers'$
can be of
arbitrary type as it will be discarded.
We write $\version{P} = \vers$ if $P$ contains a subprocess
$\handlerproctrans{T}{F}{\vers}$.

\begin{figure}[t]
  {\footnotesize
    \begin{minipage}[t]{0.5\textwidth}
      \footnotesize
      ~\textbf{Additional term reduction rule} \hfill \framebox{$M \teval N$}
\vspace{-0.5em}
  \[
    \begin{array}{lcl}
        \textsc{E-Try}   \\
        \quad \tryasinotherwise{V}{x}{M}{N} & \teval & M \{ V / x \} \\
    \end{array}
  \]%
    \end{minipage}
  \hfill
  \begin{minipage}[t]{0.475\textwidth}
    ~\textbf{Additional meta-level definitions}
\vspace{-0.5em}
  \[
    \begin{array}{rcl}
      \procs{\cmdempty} & = & \epsilon \\
      \procs{\cmdspawn{M}} & = & M \\
      \procs{\append{V}{W}} & = & \procs{V} \cdot \procs{W} \\
    \end{array}
  \]
\end{minipage}

~\textbf{Equivalence of processes} \hfill \framebox{$P \equiv P'$}
\begin{mathpar}
  (\nu c d)(\nu c' d') P \equiv (\nu c' d')(\nu c d) P

  P \parallel ((\nu c d) Q)\equiv (\nu c d) (P \parallel Q) \quad \text{if }
  c, d \not\in \fn{P}

  (\nu c d) P \equiv (\nu d c) P

  P_1 \parallel P_2 \equiv P_2 \parallel P_1

  P_1 \parallel (P_2 \parallel P_3) \equiv (P_1 \parallel P_2) \parallel P_3

  (\nu c d)(\zap{c} \parallel \zap{d}) \parallel P \equiv P

  \serverthread{()} \parallel P \equiv P
\end{mathpar}

\vspace{-0.5em}
~\textbf{Reduction of processes} \hfill \framebox{$P \ceval P'$}

\vspace{-1em}
\begin{center}
~MVU reduction rules
\end{center}
\vspace{-0.5em}
\[
  \begin{array}{@{}l@{\:\:}r@{\:\:\:}l@{}}
    \textsc{E-Discard} \hfill
    \handlerproctrans{T}{F}{\vers} \parallel \threadtrans{V}{\vers'} & \ceval &
    \handlerproctrans{T}{F}{\vers} \parallel \zap{V} \quad \text{ if } \vers \ne \vers' \\
  \textsc{E-DiscardHalt} \hfill
    \halt \parallel \threadtrans{V}{\vers} & \ceval &
    \halt \parallel \zap{V} \\
    \textsc{E-Handle} \hfill
    \hfill \handlerprocexptc{\idle{V_m}}{V_v}{V_u}{V_e}{\vers} \parallel \threadtrans{V}{\vers} & \ceval &
    \handlerprocexptc{\updating{V_u \app (V, V_m)}}{V_v}{V_u}{V_e}{\vers} \\
    \textsc{E-Extract}\quad
    \hfill \handlerproctrans{\updating{(\notransition{V_m}{V_c})}}{F}{\vers} & \ceval &
    \handlerproctrans{\extracting{V_c}{(V_e \app V_m)}}{F}{\vers} \\
                                                                                             & &
                                                                                             \text{where
                                                                                             }
                                                                                             F= \statecomb{V_v}{V_u}{V_e} \\
       \textsc{E-ExtractT} \\
       \hfill
       \handlerproctrans{\updating{(\transition{V_m}{V_v}{V_u}{V_e}{V_c})}}{F}{\vers} & \ceval &
       \handlerproctrans{\extractingtexp{V_v}{V_u}{V_e}{V_c}{(V_e \app V_m)}}{F}{\vers} \\
       \textsc{E-Render}
       \hfill
       \handlerproctrans{\extracting{V_c}{(V_m, V_{um})}}{F}{\vers} & \ceval &
       \handlerproctrans{\renderingext{V_m}{V_c}{(V_v \app V_{um})}}{F}{\vers} \\
                                                                                             & &
                                                                                             \text{where
                                                                                             }
                                                                                             F= \statecomb{V_v}{V_u}{V_e} \\

       \textsc{E-RenderT}
       \hfill \handlerproctrans{\extractingt{F'}{V_c}{(V_m, V_{um})}}{F}{\vers} & \ceval &
       \handlerproctrans{\transitioningext{V_m}{F'}{V_c}{(V_v \app V_{um})}}{F}{\vers} \\
                                                                                                & &
                                                                                                \text{where
                                                                                                }
                                                                                                F' =
\statecomb{V_v}{V_u}{V_e}
     \end{array}
     \]%

\vspace{-1.25em}
\begin{center}
  Session reduction rules
\end{center}
\vspace{-0.5em}
\[
  \begin{array}{@{}l@{\:\:}r@{\:\:\:}c@{\:\:\:}l@{}}
       \textsc{E-New} & \config{T}[\calcwd{new}()] & \ceval & (\nu c d)(\config{T}[(c, d)])  \quad \text{where } c, d  \text{ fresh } \\
       \textsc{E-Comm} & (\nu c d)(\config{T}[\gvsend{V}{c}] \parallel \config{T}'[\gvrecv{d}])
                      & \ceval &
                      (\nu c d)(\config{T}[c] \parallel \config{T}'[(V, d)]) \\
       \textsc{E-Close} & (\nu c d) (\config{T}[\gvclose{c}] \parallel \config{T}'[\gvclose{d}]) & \ceval &
       \config{T}[()] \parallel \config{T}'[()] \\
       \textsc{E-Cancel} & \config{T}[\gvcancel{c}] & \ceval & \config{T}[()] \parallel \zap{c} \\

       \textsc{E-SendZap} & (\nu c d)(\config{T}[\gvsend{V}{c}] \parallel \zap{d})
                      & \ceval &
                      (\nu c d)(\config{T}[\raiseexn] \parallel \zap{c} \parallel \zap{V} \parallel \zap{d}) \\
       \textsc{E-RecvZap} & (\nu c d)(\config{T}[\gvrecv{c}] \parallel \zap{d})
                      & \ceval &
                      (\nu c d)(\config{T}[\raiseexn] \parallel \zap{c} \parallel \zap{d}) \\

       \textsc{E-CloseZap} & (\nu c d)(\config{T}[\gvclose{c}] \parallel \zap{d})
                      & \ceval &
                      (\nu c d)(\config{T}[\raiseexn] \parallel \zap{c} \parallel \zap{d})
     \end{array}
   \]

\vspace{-1em}
   \begin{center}
     Exception reduction rules
   \end{center}
\vspace{-0.5em}
   \[
  \begin{array}{@{}l@{\:\:}r@{\:\:\:}c@{\:\:\:}l@{}}
       \textsc{E-RaiseH} & \config{T}[\tryasinotherwise{\ep[\raiseexn]}{x}{M}{N}] & \ceval & \config{T}[N] \parallel \zap{\ep} \\
      \textsc{E-RaiseURun} &
      \run{(\ep[\raiseexn])} & \ceval &
      \halt \parallel \zap{\ep} \\
      \textsc{E-RaiseUMain} &
      \handlerproctrans{\tp[\raiseexn]}{F}{\vers} & \ceval &
      \halt \parallel \zap{\tp} \\
      \textsc{E-RaiseUThread} & \threadtrans{\ep[\raiseexn]}{\vers} & \ceval & \zap{\ep} \\
      \textsc{E-RaiseUServer} & \serverthread{\ep[\raiseexn]} & \ceval & \zap{\ep}
  \end{array}
\]

\vspace{-1em}
\begin{center}
~Administrative reduction rules
\end{center}
\vspace{-0.5em}
\[
  \begin{array}{lrcll}
      \textsc{E-LiftT} &
      \config{T}[M] & \ceval & \config{T}[N] & \text{if } M \teval N \\
      \textsc{E-Nu} &
      (\nu a b) P & \ceval & (\nu a b) P' & \text{if } P \ceval P' \\
      \textsc{E-Par} &
      P_1 \parallel P_2& \ceval & P'_1 \parallel P_2 & \text{if } P_1 \ceval P'_1
\end{array}
\]%
}
\vspace{-0.5em}
\caption{Reduction rules for extended calculus (1)}
\label{fig:combined:reduction-1}
\end{figure}
\begin{figure}[t]
  {
  \footnotesize
~\textbf{Reduction of configurations} \hfill \framebox{$\config{C} \ceval \config{C}'$}
\[
  \begin{array}{@{}l@{\:\:}c@{\:\:\:}l@{}}
    \textsc{E-Run} \\
    \hfill
    \sys{\procctx[\run{(V_m, V_v, V_u, V_e, V_c, \lambda () . M)}]}{D}  & \ceval &
      \sys{\handlerprocexptc{\extracting{V_c}{(V_e \app
      V_m)}}{V_v}{V_u}{V_e}{0} \parallel \serverthread{M}}{D}
      \vspace{\redrowskip} \\
\textsc{E-Update} \\
\hfill
\sys{\procctx[\handlerproctrans{\renderingext{V_m'}{V_c}{U}}{F}{\vers}]}{\vh}
                 & \ceval &
                 \sys{\procctx[\handlerproctrans{\idle{V'_m}}{F}{\vers} \parallel \threadtrans{M_1}{\vers} \parallel \cdots \parallel \threadtrans{M_n}{\vers}]}{\vh'}
                 \\
                 & & \quad \text{where } \diff{U}{\vh} = \vh' \text{ and } \procs{V_c} = \seq{M} \\
  \textsc{E-Transition} \\
  \hfill
    \sys{\procctx[\handlerproctrans{\transitioningext{V_m}{F'}{V_c}{U}}{F}{\vers}]}{\vh} & \ceval &
    \sys{\procctx[\handlerproctrans{\idle{V_m}}{F'}{\vers'} \parallel \threadtrans{M_1}{\vers'} \parallel
    \cdots \parallel \threadtrans{M_n}{\vers'}]}{\vh'} \\
    & & \quad \text{where }
    {\begin{array}[t]{l}
      \vers' = \vers + 1,
      \diff{U}{D} = D' \\
      \text{ and } \procs{V_c} = \seq{M}
    \end{array}} \\
      \textsc{E-Evt} \\
      \hfill
    \sys{P}{\config{D}[\pgtag{t}{U}{D}{\evtpayload{\evt{ev}}{W} \cdot \seq{e}}]}
                     & \ceval &
                     \sys{P \parallel \threadtrans{V_1 \app W}{\vers} \parallel \cdots \parallel \threadtrans{V_n \app W}{\vers}}{\config{D}[\pgtag{t}{U}{D}{\seq{e}}]} \\
                     & &  \quad \text{where } \handlers{\evt{ev}}{U} = \seq{V} \text{ and } \version{P} = \vers
  \end{array}
  \]
  \begin{center}
  (\textsc{E-Interact}, \textsc{E-Struct}, \textsc{E-LiftP} unchanged)
\end{center}

~\textbf{Cancellation of pure active thread contexts} \hfill
\framebox{$\zap{\tp}$}
\begin{mathpar}
  \zap{\updating{\ep}} = \zap{\ep}

  \zap{\rendering{V_m, V_c}{\ep}} = \zap{V_m} \parallel \zap{V_c} \parallel \zap{\ep}

  \zap{\extracting{V_c}{\ep}} = \zap{V_c} \parallel \zap{\ep}

  \zap{\extractingt{V_c}{F}{\ep}} = \zap{V_c} \parallel \zap{\ep}

  \zap{\transitioning{V_m}{F}{V_c}{\ep}} = \zap{V_m} \parallel \zap{V_c}
  \parallel \zap{\ep}
\end{mathpar}
}
\vspace{-1.5em}
\caption{Reduction rules for extended calculus (2)}
\label{fig:combined:reduction-2}
\end{figure}

\subparagraph{Reduction.}
Figures~\ref{fig:combined:reduction-1} and~\ref{fig:combined:reduction-2} show
the extended process equivalence and reduction rules. Rule \textsc{E-Try}
handles evaluation of the success continuation of an exception handler, and the
$\mkwd{procs}$ meta-definition returns a sequence of processes to be spawned by
a command.
Process equivalence is extended to allow commutativity of name
restrictions, reordering of names in a binder, and scope extrusion. The final
``garbage collection'' equivalences $(\nu c d)(\zap{c} \parallel \zap{d})
\parallel P \equiv P$ and $\serverthread{()} \parallel P \equiv P$ allow us to discard a
channel where both endpoints have been cancelled, and a completed server thread,
respectively.

Figure~\ref{fig:combined:reduction-1} details the extended MVU process reduction rules.

\subparagraph{MVU reduction.}
MVU reduction rules are specific to MVU.
Central to safely integrating linearity and transitions are rules
\textsc{E-Discard}, \textsc{E-DiscardHalt}, and \textsc{E-Handle}.
Rule \textsc{E-Handle} is modified so
that the event loop process only handles a message if the message has the same
version. If the versions do not match, then \textsc{E-Discard} safely discards
any channel endpoints in the discarded message by generating zapper threads.
Rules \textsc{E-Extract}, \textsc{E-ExtractT}, \textsc{E-Render}, and
\textsc{E-RenderT} handle the state machine transitions described in
Figure~\ref{fig:extensions:combined:rt-syntax} and are used to calculate the new
model and HTML.

\subparagraph{Session reduction.}
Session reduction rules encode session-typed communication and are mostly
standard: \textsc{E-New} generates a name restriction and returns two fresh
endpoints; \textsc{E-Comm} handles synchronous communication; and
\textsc{E-Close} discards the endpoints of a completed session.
The remaining session communication rules handle session cancellation, and are
a synchronous variant of Exceptional GV described by~\citet{FowlerLMD19:stwt}.
Rule \textsc{E-Cancel} discards an endpoint. Rules \textsc{E-SendZap},
\textsc{E-RecvZap}, and \textsc{E-CloseZap} raise an exception if a thread tries
to communicate along an endpoint whose peer is cancelled, ensuring resources are
discarded safely.

\subparagraph{Exception reduction.}
Rule \textsc{E-RaiseH} describes exception handling: as $\raiseexn$ occurs in a
pure context, the exception is handled by the innermost handler; the rule
evaluates the failure continuation and discards all linear resources in the
aborted context. Rules \textsc{E-RaiseURun} and \textsc{E-RaiseUMain} apply to
unhandled exceptions in a main thread, generating the $\halt$ configuration and
cancelling any linear resources in the aborted context.  Rules
\textsc{E-RaiseUThread} and \textsc{E-RaiseUServer} apply to unhandled
exceptions in event loop thread and server threads respectively, by cancelling
any channels in the aborted continuation.

\subparagraph{Configuration reduction.}
Figure~\ref{fig:combined:reduction-2} shows the modified configuration reduction
rules. We modify \textsc{E-Run} to take into account the new arguments, and
spawn the given server thread.
We modify \textsc{E-Update} to spawn threads described by the returned command;
\textsc{E-Transition} is similar but changes the function state and increments
the version. We modify \textsc{E-Evt} to tag each
spawned event handler thread with the version of the event handler process.

\begin{figure}[t]
{\footnotesize
\textbf{Typing rules for names, events, and function state} \hfill
\framebox{$\vphantom{\statetytrans{A}{B}{C}} \Gamma \vdash M : A$}
\framebox{$\vphantom{\statetytrans{A}{B}{C}} \vdash e$}
\framebox{$\Psi \vdash F : \statetytrans{A}{B}{C}$}
\vspace{-0.5em}
  \begin{mathpar}
    \inferrule
    [T-Name]
    { \Gamma :: \unr }
    { \Gamma, c : S \vdash c : S }

    \inferrule
    [TE-Evt]
    { \vdash V : \evtty{\evt{ev}} \\\\ \evtty{\evt{ev}} :: \unr }
    { \vdash \evtpayload{\evt{ev}}{V} }

  \inferrule
  [TF-State]
  {
    \Psi_1 \vdash V_v : A \uto \htmlty{B} \\
    \Psi_2 \vdash V_u : (B \times A) \uto \transitionty{A}{B} \\\\
    \Psi_3 \vdash V_e : A \uto (A \times C) \\
    C :: \unr
  }
  { \Psi_1, \Psi_2, \Psi_3 \vdash (V_m, V_v, V_u) : \statetytrans{A}{B}{C}}
\end{mathpar}

~\textbf{Typing rules for active threads} \hfill \framebox{$\Psi \vdashs T : \evtlooptytrans{A}{B}{C}$}
\vspace{-0.5em}
 \begin{mathpar}
  \inferrule
  [TT-Idle]
  { \Psi \vdash V_m {:} A  }
  { \Psi \vdashs \idle{V_m} {:} \evtlooptytrans{A}{B}{C} }

  \inferrule
  [TT-Updating]
  { \Psi \vdash M {:} \transitionty{A}{B} }
  { \Psi \vdashs \updating{M} {:} \evtlooptytrans{A}{B}{C} }

  \inferrule
  [TT-Rendering]
  { \Psi_1 \vdash V_m {:} A \\ \Psi_2 \vdash V_c {:} \cmdty{B} \\ \Psi_3 \vdash M {:} \htmlty{B} }
  { \Psi_1, \Psi_2, \Psi_3 \vdashs \rendering{V_m, V_c}{M} {:} \evtlooptytrans{A}{B}{C} }

  \inferrule
  [TT-Extracting]
  { \Psi_1 \vdash V_c {:} \cmdty{B} \\ \Psi_2 \vdash M {:} (A \times C) }
  { \Psi_1, \Psi_2 \vdashs \extracting{V_c}{M} {:} \evtlooptytrans{A}{B}{C} }

  \inferrule
  [TT-ExtractingT]
  { \Psi_1 \vdash F {:} \statetytrans{A}{B}{C} \\
    \Psi_2 \vdash V_c {:} \cmdty{B} \\
    \Psi_3 \vdash M {:} (A \times C) }
  { \Psi_1, \Psi_2, \Psi_3 \vdashs \extractingt{F}{V_c}{M} {:} \evtlooptytrans{A'}{B'}{C'} }

  \inferrule
  [TT-Transitioning]
  { \Psi_1 \vdash V_m {:} A \\
    \Psi_2 \vdash F {:} \statetytrans{A}{B}{C} \\
    \Psi_3 \vdash V_c {:} \cmdty{B} \\
    \Psi_4 \vdash \mh {:} \htmlty{B}
  }
  { \Psi_1, \ldots, \Psi_4 \vdashs \transitioningext{V_m}{F}{V_c}{\mh} {:} \evtlooptytrans{A'}{B'}{C'} }
\end{mathpar}

~\textbf{Typing rules for processes} \hfill \framebox{$\Psi \vdashtrans{\phi}{\vers} P : A$}
\vspace{-0.5em}
    \begin{mathpar}
    \inferrule
    [TP-Run]
    { \Psi \vdash M :\!\!\!
      {\begin{array}[t]{l}
          (A \times (A \uto \htmlty{B}) \times ((B \times A) \uto
          \transitionty{A}{B}) \: \times \\
        \quad (A \uto (A \times C)) \times \cmdty{B} \times (\one \lto \one))
       \el
      }
      \\
      C :: \unr
    }
    { \Psi \vdashtrans{\bcirc}{\vers} \run{M} : B }

    \inferrule
    [TP-EventLoop]
    { \Psi_1 \vdashs T : \evtlooptytrans{A}{B}{C} \\\\
      \Psi_2 \vdash F : \statetytrans{A}{B}{C}
    }
    { \Psi_1, \Psi_2 \vdashtrans{\bcirc}{\vers} \handlerproctrans{T}{F}{\vers} : B }

    \inferrule
    [TP-Thread]
    { \Psi \vdash M : A }
    { \Psi \vdashtrans{\wcirc}{\vers} \threadtrans{M}{\vers} : A }

    \inferrule
    [TP-OldThread]
    { \Psi \vdash M : B \\ \vers \ne \vers' }
    { \Psi \vdashtrans{\wcirc}{\vers} \threadtrans{M}{\vers'} : A }

    \inferrule
    [TP-ServerThread]
    { \Psi \vdash M : \one }
    { \Psi \vdashtrans{\wcirc}{\vers} \serverthread{M} : A }

    \inferrule
    [TP-Par]
    { \Psi_1 \vdashtrans{\phi_1}{\vers} P_1 : A \\ \Psi_2 \vdashtrans{\phi_2}{\vers} P_2 : A }
    { \Psi_1, \Psi_2 \vdashtrans{\phi_1 + \phi_2}{\vers} P_1 \parallel P_2 : A }

    \inferrule
    [TP-Zap]
    {  }
    { c : S \vdashtrans{\wcirc}{\vers} \zap{c} : A }

    \inferrule
    [TP-Halt]
    {  }
    { \cdot \vdashtrans{\bcirc}{\vers} \halt : A }

    \inferrule
    [TP-Nu]
    { \Psi, c : S, d : \gvdual{S} \vdashtrans{\phi}{\vers} P : A }
    { \Psi \vdashtrans{\phi}{\vers} (\nu c d) P : A }
  \end{mathpar}
}
\caption{Runtime typing for extended calculus}
\label{fig:extensions:combined:runtime-typing}
\end{figure}
 \subsubsection{Metatheory}
\subparagraph{Runtime typing.}
Figure~\ref{fig:extensions:combined:runtime-typing} shows the runtime typing
rules for the extended calculus. Rule \textsc{T-Name} types channel endpoints,
and rule \textsc{TE-Evt} mandates that event payload types are unrestricted. The
rules for active threads ensure that the types of the terms being evaluated
correspond to the state in the state machine (for example, that the
$\calcwd{updating}$ state returns a term of type $\transitionty{A}{B}$), and
that any recorded values have the correct types.

Let $\Psi$ range over environments containing only runtime names: $\Psi ::=
\cdot \midspace \Psi, c : S$.
We write $\Psi_1, \Psi_2$ for the disjoint union of environments $\Psi_1$ and
$\Psi_2$.

We modify the shape of the process typing judgement to $\Psi
\vdashtrans{\phi}{\vers} P : A$, which can be read ``under typing environment
$\Psi$ and thread flag $\phi$, process $P$ has type $A$ and version $\vers$''.
We modify rule \textsc{TP-EventLoop} to include the extraction function, and
mandate that the unrestricted model type $C$ has kind $\unr$. We modify rule
\textsc{T-Thread} to state that type of an event handler thread
$\threadtrans{M}{\vers}$  has type $A$ if term $M$ has type $A$ and the version
matches that of the event handler process. Rule
\textsc{TP-OldThread} allows a thread to have a mismatching type to the
event handler process if the versions are incompatible. Finally, \textsc{TP-Zap} and
\textsc{TP-Halt} type zapper threads and the $\halt$ thread, and \textsc{TP-Nu}
types a name restriction $(\nu c d) P$ by adding $c$ and $d$ with dual session
types into the typing environment.

\subparagraph{Properties.}
The extended calculus satisfies preservation.

\begin{theorem}[Preservation]\label{thm:combined:preservation}
  If $\vdash \config{C}$ and $\config{C} \ceval \config{C}'$, then $\vdash \config{C}'$.
\end{theorem}

Although session types rule out deadlock within a single session, without imposing
a tree-like structure on processes~\cite{Wadler14:propositions-sessions, LindleyM15:semantics}
(which is too inflexible for our purposes) or using techniques such as channel
priorities~\cite{Padovani14:df-pi, PadovaniN15:df-lambda, Kobayashi06:df}, it is
not possible to rule out deadlocks when considering multiple sessions.  Since
communication over multiple sessions can introduce deadlocks, we begin by
proving an error-freedom property, similar to that of~\citet{GayV10:last}. An
\emph{error process} involves a communication mismatch.

\begin{definition}[Error process]
  A process $P$ is an \emph{error process} it contains one of the following
  processes as a subprocess:
  \begin{enumerate}
    \item $(\nu c d)(\config{T}[\gvsend{V}{c}] \parallel \config{T}'[\gvsend{W}{d}])$
    \item $(\nu c d)(\config{T}[\gvsend{V}{c}] \parallel \config{T}'[\gvclose{d}])$
    \item $(\nu c d)(\config{T}[\gvrecv{c}] \parallel \config{T}'[\gvrecv{d}])$
    \item $(\nu c d)(\config{T}[\gvrecv{c}] \parallel \config{T}'[\gvclose{d}])$
  \end{enumerate}
\end{definition}

Configuration typing ensures error-freedom.

\begin{theorem}[Error-freedom]\label{thm:combined:rt-errors}
  If $\Psi \vdashtrans{\phi}{\vers} P : A$, then $P$ is not an error process.
\end{theorem}

Error-freedom shows that session typing ensures the absence of communication
mismatches. What remains is to show that, apart from the possibility of
deadlock, the additional features do not interfere with the progress property
enjoyed by \mvu. We begin by classifying the notion of a \emph{blocked} thread,
which is a thread blocked on a communication action.

\begin{definition}[Blocked thread]
  We say a thread $\config{T}[M]$ is \emph{blocked} if either $M = \gvsend{V}{W}$, $M = \gvrecv{V}$, or $M = \gvclose{V}$.
\end{definition}

Let us refer to $\halt$, $\handlerproctrans{T}{F}{\vers}$, and $\run{M}$ as
\emph{main threads}, and $\threadtrans{M}{\vers}$, $\serverthread{M}$, and
$\zap{c}$ as \emph{auxiliary threads}. Each well-typed configuration has
precisely one main thread.

We can now classify the notion of progress enjoyed by the extended
calculus.
Either the configuration can reduce; is waiting for an
event; has halted due to an unhandled exception; or is deadlocked.
Again, let $\cevalminus$ be the $\ceval$ relation without \textsc{E-Interact}.

\begin{theorem}[Weak Event Progress]\label{thm:combined:weak-progress}
  Suppose $\vdash \config{C}$. Either there exists some $\config{C'}$ such
  that $\config{C} \ceval \config{C}'$, or there exists some $\config{C'}$ such
  that $\config{C} \equiv \config{C}'$ and:

  \begin{enumerate}
    \item $D$ cannot be written $\config{D}[\pgtag{\tagname{t}}{V}{D}{\seq{e}}]$
      for a non-empty $\seq{e}$.
    \item If the main thread of $\config{C}'$ is $\halt$, then all auxiliary
      threads are blocked or zapper threads.
    \item If the main thread of $\config{C}'$ is $\run{M}$, then $M$ is blocked,
      and all auxiliary threads are either blocked, values, or zapper threads.
    \item If the main thread of $\config{C}'$ is
      $\handlerproctrans{T}{F}{\vers}$, then:
      \begin{enumerate}
        \item if $T = \idle{V_m}$, then each auxiliary thread is either blocked or a
          zapper thread; or
        \item if $T = \ta[L]$ then $L$ is blocked, and each auxiliary thread is
          either blocked, a value, or a zapper thread.
      \end{enumerate}
  \end{enumerate}
\end{theorem}
\section{Implementation and Example Application}\label{sec:example}
We have implemented an MVU library for the Links tierless web programming
language, which includes all extensions in the paper;
Links already has a linear type system and distributed session types, so is an
ideal fit.
\begin{figure}[t]
\centering{
  \begin{subfigure}{0.45\textwidth}
    \centering
    \includegraphics[scale=0.25]{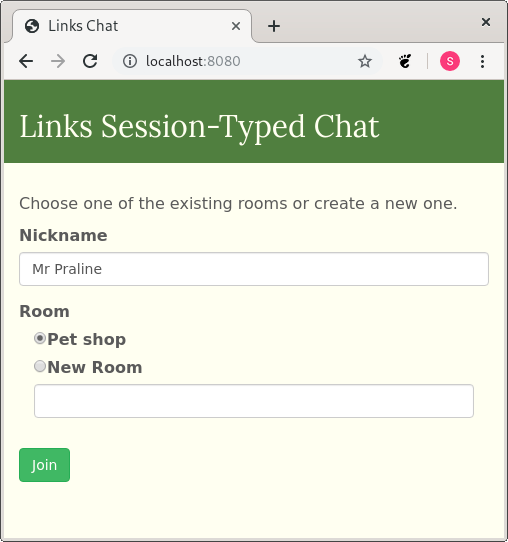}
    \caption{Login}
  \end{subfigure}
  \qquad
  \begin{subfigure}{0.45\textwidth}
    \centering
    \includegraphics[scale=0.25]{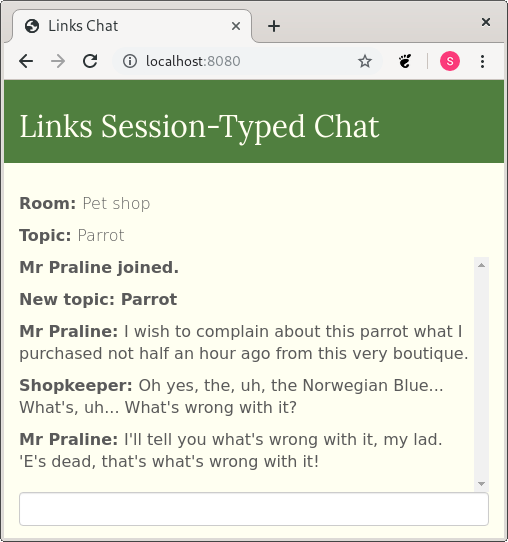}
    \caption{Chat}
\end{subfigure}
}

\totheleft{\footnotesize \textbf{Client session types}}
\vspace{-1em}
\begin{lstlisting}[language=links]
typename ClientConnect = ?([RoomName]).ClientSelect;
typename ClientSelect = !(RoomName, Nickname).
  [&| JoinedOK: ?(Topic, [Nickname], ClientReceive) . ClientSend,
      JoinedOKAsModerator: ?(Topic, [Nickname], ClientReceive, ModeratorSend). ClientSend,
      Nope: ?ConnectError.End |&];
\end{lstlisting}
\vspace{-1em}
\begin{minipage}[t]{0.5\textwidth}
\begin{lstlisting}[language=links]
typename ClientReceive = [&|
  IncomingChatMessage:
    ?(Nickname, Message). ClientReceive,
  NewUser: ?(Nickname). ClientReceive,
  NewTopic: ?(Topic). ClientReceive,
  UserLeft: ?(Nickname). ClientReceive,
  UserMuted: ?(Nickname). ClientReceive,
  UserUnmuted: ?(Nickname). ClientReceive,
  BecomeModerator: ?ModeratorSend. ClientReceive,
  Kick: End |&];
\end{lstlisting}
\end{minipage}
\hfill
\begin{minipage}[t]{0.45\textwidth}
\begin{lstlisting}[language=links]
typename ClientSend = [+|
  ChatMessage: !(Message).ClientSend,
  ChangeTopic: !(Topic).ClientSend,
  Leaving: End |+];

typename ModeratorSend = [+|
  KickUser: !(Nickname).ModeratorSend,
  MuteUser: !(Nickname).ModeratorSend,
  MakeModerator: !(Nickname).ModeratorSend
|+];
\end{lstlisting}
\end{minipage}
\caption{Chat server application}
\label{fig:example:screenshot}
\end{figure}

 We now describe a chat application, extending the application presented
by~\citet{FowlerLMD19:stwt}.
The application (Figure~\ref{fig:example:screenshot}) has two main
stages shown to the user: on the first, the user is presented with a list of
rooms, and enters a username and selects a room.
If a user with the given nickname is not already in the selected room,
then the user joins the room, receiving the current topic, a list of other
nicknames, and a channel used to receive messages from the server. The user can
then send chat messages, change the topic, and leave the room. If the user is
the first user in the room, then they join as a moderator
and receive an additional channel which can be used to kick, mute, or promote
other users to moderators.
Users can receive incoming chat messages, and system messages detailing changes
such as a new topic or a user joining the room.

We can encode these interaction patterns using session types.
Links session type notation for offering a choice is \lstinline+[&|...|&]+, and
making a choice is \lstinline-[+|...|+]-.
Type \lstinline+ClientConnect+ describes the client receiving the room list.
Type \lstinline+ClientSelect+ describes the client sending the room name and
nickname, and receiving the response from the server: either joining as a
regular user (\lstinline+JoinedOK+); joining as a moderator
(\lstinline+JoinedOKAsModerator+); or an error (\lstinline+Nope+).  Types
\lstinline+ClientSend+ and \lstinline+ClientReceive+ detail the messages that
the client can send to, and receive from the server, respectively.  Type
\lstinline+ModeratorSend+ details privileged moderator actions.

Although the original version of Links~\cite{CooperLWY06:links} ran as a CGI
script, modern Links applications run as a persistent webserver.
Upon execution, the chat application creates an \emph{access point} for
sessions of type \lstinline+ClientConnect+, which supports session
establishment, and spawns an acceptor thread to accept incoming requests on
the access point. Each chat room is represented as a process on the server.
When an HTTP request is made, the response contains the MVU application and
the access point ID which can be used to establish a session of type
\lstinline+ClientConnect+. After the initial HTTP response, further
communication between the client and server happens over a
WebSocket~\cite{FetteM11:websockets}.

The application has three states: connection, chatting, and a ``waiting'' state
shown while waiting for a response. For the purposes of the paper,
we consider the connection state.

\begin{minipage}{0.45\textwidth}
\begin{lstlisting}
typename SelectedRoom =
  [| NewRoom | SelectedRoom: String |];
typename NotConnectedModel =
  (nickname: String, rooms: [RoomName],
   selectedRoom: SelectedRoom,
   newRoomText: RoomName, error: Maybe(Error));
 \end{lstlisting}
\end{minipage}
\hfill
\begin{minipage}{0.45\textwidth}
 \begin{lstlisting}
typename NCModel =
  (ClientSelect, NotConnectedModel);
typename NCMessage = [|
 [| UpdateNickname: Nickname
  | UpdateSelectedRoom: SelectedRoom
  | UpdateNewRoom: RoomName | SubmitJoinRoom |];
  \end{lstlisting}
\end{minipage}

The \lstinline+NotConnectedModel+ is the unrestricted part of the model, and
contains the current nickname (\lstinline+nickname+), list of rooms
(\lstinline+rooms+), selected room (\lstinline+selectedRoom+), value
of the ``new room'' text box (\lstinline+newRoomText+), and an optional error
message to display (\lstinline+error+). The model, \lstinline+NCModel+, is a
pair of a session endpoint of type \lstinline+ClientSelect+ and a
\lstinline+NotConnectedModel+.
The UI messages are described by the \lstinline+NCMessage+ type: for example,
the \lstinline+UpdateNickname+ message is generated by the \lstinline+onInput+
event of the nickname input box.

Upon receiving the \lstinline+SubmitJoinRoom+ UI message when the form is
submitted, the application can send the nickname and selected room along the
\lstinline+ClientSelect+ channel, all of which are contained in the model,
without requiring ad-hoc messaging or imperative updates.
\section{Related work}

Flapjax~\cite{MeyerovichGBCGBK09:flapjax} was the first web programming language to use
\emph{functional reactive programming} (FRP)~\cite{ElliottH97:frp} in the setting of web
applications. Flapjax provides \emph{behaviours}, which are variables whose
contents change over time, and \emph{event streams}, which are an infinite
stream of discrete events which change a behaviour.
ScalaLoci~\cite{WeisenburgerKS18:scalaloci} is a multi-tier reactive programming
framework written in Scala, where changes in reactive \emph{signals} are propagated
across tiers, rather than using explicit message passing.
Ur/Web~\cite{chlipala15:urweb} and WebSharper UI~\cite{FowlerDG15:uinext}
store data in mutable variables, and allow
views of the data to be combined using monadic combinators.

\citet{FelleisenFFK09:fffk} describe an earlier approach similar to MVU written
in the DrScheme~\cite{FindlerCFFKSF02:drscheme} system.  Similar to the MVU
update function, events such as key presses and mouse movements are handled
using functions of type $(\mkwd{Model} \times \mkwd{Event}) \to \mkwd{Model}$.
The approach handles ``environment'' events rather than events dispatched by
individual elements, and the approach is not formalised. Environment events can
be handled using \emph{subscriptions} in Elm, which can be added to $\mvu$ (see
the extended version of the paper~\cite{Fowler20:mvuc-extended}).

React~\cite{react} is a popular JavaScript UI framework.
In React, a user defines data models and rendering functions, and similar to
Elm, updates are propagated to the DOM by diffing.
Differently to MVU, there is no notion of a message, and a page consists
of multiple components rather than being derived from a single model. We expect
some technical machinery from \mvu (e.g., event queues, DOM
contexts, and diffing) could be reused when formalising React. Redux~\cite{redux}
is a state container for JavaScript applications: to modify the state, one
dispatches an \emph{action}, and a function takes the previous state and
an action and produces a new state. In combination with React, the
approach strongly resembles MVU.

Hop.js~\cite{SerranoP16:hopjs} is a multi-tier web framework written in
JavaScript. Hop.js \emph{services} allow remote function invocation, and
the framework supports client-side message-passing concurrency using Web
Workers~\cite{Green12:web-workers}, but there is no cross-tier message-passing
concurrency.

Session types were introduced by~\citet{Honda93:dyadic}
and were first considered in a linear functional
language by~\citet{GayV10:last}; \citet{Wadler14:propositions-sessions} later introduced
a session-typed functional language GV and a logically-grounded session-typed
calculus CP (following~\citet{CairesP10:propositions}),
and translated GV into CP. \citet{LindleyM15:semantics} introduced an
operational semantics for GV, and showed
type- and semantics-preserving
translations between GV and CP. GV inspires FST~\cite{LindleyM17:fst}, which
is the core calculus for Links' treatment of session typing.

\citet{FowlerLMD19:stwt} extend GV with failure handling, and extend Links with
cross-tier session-typed communication. They do not formally consider GUI
development, and their approach to frontend web programming using session types
(described in Section~\ref{sec:introduction}) leads to a disconnect between the
state of the page and the application logic.  We build upon their approach to
session-typed web programming, while also allowing idiomatic GUI development.

\citet{KingNY19:mpst-web} present a toolchain for writing web applications which
respect multiparty session types~\cite{HondaYC16:multiparty}.  Protocols are
compiled to PureScript~\cite{purescript} using a parameterised
monad~\cite{Atkey09:parameterised} to guarantee linearity, and the authors
integrate their encoding of session types with the Concur UI
framework~\cite{concur}. Each application may only have a single session
connecting the client and server, whereas in our system there may be multiple;
our approach supports first-class linearity and cross-tier typechecking;
our approach is formalised; and our approach supports failure handling.  Links
does not yet support multiparty session types.
\section{Conclusion}
Session types allow conformance to protocols to be checked statically.
The last few years have seen a flurry of
activity in implementing session types in a multitude of programming languages,
but linearity---a vital prerequisite for implementing session types
safely---is difficult to reconcile with the asynchronous nature of graphical
user interfaces. Consequently, the vast majority of implementations using
session types are command line applications, and the few implementations which
do integrate session types and GUIs do so in an ad-hoc manner.

In this paper, we have addressed this problem by extending the Model-View-Update
architecture, pioneered by the Elm programming language.
We have presented the first formal study of MVU by introducing a core calculus,
\mvu.
Leveraging our formal characterisation of MVU, we have introduced three
extensions: commands, linearity, and model transitions, enabling us to
present the first formal integration of session-typed communication
with a GUI framework. Informed by our formal model, we have implemented our
approach in Links. %
As future work, we will investigate how to encode allowed transitions as a
behavioural type.

\bibliographystyle{plainnat}
\bibliography{bibliography}

\begin{thebibliography}{47}
\providecommand{\natexlab}[1]{#1}
\providecommand{\url}[1]{\texttt{#1}}
\expandafter\ifx\csname urlstyle\endcsname\relax
  \providecommand{\doi}[1]{doi: #1}\else
  \providecommand{\doi}{doi: \begingroup \urlstyle{rm}\Url}\fi

\bibitem[elm(2019)]{elm-lang}
Elm: {A} delightful language for reliable webapps, 2019.
\newblock URL \url{http://www.elm-lang.org}.
\newblock Accessed on 2019-07-04.

\bibitem[rea(2019)]{react}
React -- a {J}ava{S}cript library for building user interfaces, 2019.
\newblock URL \url{http://www.reactjs.org}.
\newblock Accessed on 2019-09-02.

\bibitem[web(2019)]{websharper-mvu}
{WebSharper.Mvu}, 2019.
\newblock URL \url{https://github.com/dotnet-websharper/mvu}.
\newblock Accessed on 2019-07-04.

\bibitem[flu(2020)]{flux}
Flux, 2020.
\newblock URL \url{https://facebook.github.io/flux/}.
\newblock Accessed on 2020-01-08.

\bibitem[red(2020)]{redux}
{R}edux - a predictable state container for {J}ava{S}cript apps, 2020.
\newblock URL \url{https://redux.js.org/}.
\newblock Accessed on 2020-01-08.

\bibitem[Atkey(2009)]{Atkey09:parameterised}
Robert Atkey.
\newblock Parameterised notions of computation.
\newblock \emph{J. Funct. Program.}, 19\penalty0 (3-4):\penalty0 335--376,
  2009.

\bibitem[Caires and Pfenning(2010)]{CairesP10:propositions}
Lu{\'{\i}}s Caires and Frank Pfenning.
\newblock Session types as intuitionistic linear propositions.
\newblock In \emph{{CONCUR}}, volume 6269 of \emph{Lecture Notes in Computer
  Science}, pages 222--236. Springer, 2010.

\bibitem[Cervesato and Pfenning(1996)]{CervesatoP96:llf}
Iliano Cervesato and Frank Pfenning.
\newblock A linear logical framework.
\newblock In \emph{{LICS}}, pages 264--275. {IEEE} Computer Society, 1996.

\bibitem[Chlipala(2015)]{chlipala15:urweb}
Adam Chlipala.
\newblock {U}r/{W}eb: {A} simple model for programming the web.
\newblock In \emph{{POPL}}, pages 153--165. {ACM}, 2015.

\bibitem[Cooper et~al.(2006)Cooper, Lindley, Wadler, and
  Yallop]{CooperLWY06:links}
Ezra Cooper, Sam Lindley, Philip Wadler, and Jeremy Yallop.
\newblock Links: Web programming without tiers.
\newblock In \emph{{FMCO}}, volume 4709 of \emph{Lecture Notes in Computer
  Science}, pages 266--296. Springer, 2006.

\bibitem[Czaplicki(2016)]{farewell-frp}
Evan Czaplicki.
\newblock Farewell to {FRP}, 2016.
\newblock URL \url{https://elm-lang.org/news/farewell-to-frp}.
\newblock Accessed on 2019-09-02.

\bibitem[Czaplicki and Chong(2013)]{CzaplickiC13:elm}
Evan Czaplicki and Stephen Chong.
\newblock Asynchronous functional reactive programming for {GUI}s.
\newblock In Hans{-}Juergen Boehm and Cormac Flanagan, editors, \emph{{ACM}
  {SIGPLAN} Conference on Programming Language Design and Implementation,
  {PLDI} '13, Seattle, WA, USA, June 16-19, 2013}, pages 411--422. {ACM}, 2013.
\newblock \doi{10.1145/2491956.2462161}.
\newblock URL \url{https://doi.org/10.1145/2491956.2462161}.

\bibitem[Dardha et~al.(2017)Dardha, Giachino, and
  Sangiorgi]{DardhaGS17:revisited}
Ornela Dardha, Elena Giachino, and Davide Sangiorgi.
\newblock Session types revisited.
\newblock \emph{Inf. Comput.}, 256:\penalty0 253--286, 2017.

\bibitem[Elliott and Hudak(1997)]{ElliottH97:frp}
Conal Elliott and Paul Hudak.
\newblock Functional reactive animation.
\newblock In \emph{{ICFP}}, pages 263--273. {ACM}, 1997.

\bibitem[Esch(2016)]{virtual-dom}
Matt Esch.
\newblock \verb+virtual-dom+, 2016.
\newblock URL \url{https://github.com/Matt-Esch/virtual-dom}.
\newblock Accessed on 2019-09-11.

\bibitem[Felleisen et~al.(2009)Felleisen, Findler, Flatt, and
  Krishnamurthi]{FelleisenFFK09:fffk}
Matthias Felleisen, Robert~Bruce Findler, Matthew Flatt, and Shriram
  Krishnamurthi.
\newblock A functional {I/O} system or, fun for freshman kids.
\newblock In \emph{{ICFP}}, pages 47--58. {ACM}, 2009.

\bibitem[Fette and Melnikov(2011)]{FetteM11:websockets}
Ian Fette and Alexey Melnikov.
\newblock The {W}eb{S}ocket protocol, 2011.

\bibitem[Findler et~al.(2002)Findler, Clements, Flanagan, Flatt, Krishnamurthi,
  Steckler, and Felleisen]{FindlerCFFKSF02:drscheme}
Robert~Bruce Findler, John Clements, Cormac Flanagan, Matthew Flatt, Shriram
  Krishnamurthi, Paul Steckler, and Matthias Felleisen.
\newblock {D}r{S}cheme: a programming environment for {S}cheme.
\newblock \emph{J. Funct. Program.}, 12\penalty0 (2):\penalty0 159--182, 2002.

\bibitem[Fowler(2019)]{Fowler20:mvuc-extended}
Simon Fowler.
\newblock {M}odel-{V}iew-{U}pdate-{C}ommunicate: Session types meet the {E}lm
  architecture ({E}xtended version), 2019.
\newblock URL \url{https://arxiv.org/abs/1910.11108}.

\bibitem[Fowler et~al.(2015)Fowler, Denuzi{\`{e}}re, and
  Granicz]{FowlerDG15:uinext}
Simon Fowler, Lo{\"{\i}}c Denuzi{\`{e}}re, and Adam Granicz.
\newblock Reactive single-page applications with dynamic dataflow.
\newblock In \emph{{PADL}}, volume 9131 of \emph{Lecture Notes in Computer
  Science}, pages 58--73. Springer, 2015.

\bibitem[Fowler et~al.(2019)Fowler, Lindley, Morris, and
  Decova]{FowlerLMD19:stwt}
Simon Fowler, Sam Lindley, J.~Garrett Morris, and S{\'{a}}ra Decova.
\newblock Exceptional asynchronous session types: session types without tiers.
\newblock \emph{{PACMPL}}, 3\penalty0 ({POPL}):\penalty0 28:1--28:29, 2019.
\newblock \doi{10.1145/3290341}.
\newblock URL \url{https://doi.org/10.1145/3290341}.

\bibitem[Gay and Vasconcelos(2010)]{GayV10:last}
Simon~J. Gay and Vasco~Thudichum Vasconcelos.
\newblock Linear type theory for asynchronous session types.
\newblock \emph{J. Funct. Program.}, 20\penalty0 (1):\penalty0 19--50, 2010.

\bibitem[Green(2012)]{Green12:web-workers}
Ido Green.
\newblock \emph{{W}eb {W}orkers - Multithreaded Programs in {J}ava{S}cript}.
\newblock O'Reilly, 2012.

\bibitem[Honda(1993)]{Honda93:dyadic}
Kohei Honda.
\newblock Types for dyadic interaction.
\newblock In \emph{{CONCUR}}, volume 715 of \emph{Lecture Notes in Computer
  Science}, pages 509--523. Springer, 1993.

\bibitem[Honda et~al.(1998)Honda, Vasconcelos, and Kubo]{HondaVK98:primitives}
Kohei Honda, Vasco~Thudichum Vasconcelos, and Makoto Kubo.
\newblock Language primitives and type discipline for structured
  communication-based programming.
\newblock In \emph{{ESOP}}, volume 1381 of \emph{Lecture Notes in Computer
  Science}, pages 122--138. Springer, 1998.

\bibitem[Honda et~al.(2016)Honda, Yoshida, and Carbone]{HondaYC16:multiparty}
Kohei Honda, Nobuko Yoshida, and Marco Carbone.
\newblock Multiparty asynchronous session types.
\newblock \emph{J. {ACM}}, 63\penalty0 (1):\penalty0 9:1--9:67, 2016.

\bibitem[Jain(2019)]{concur}
Anupam Jain.
\newblock Concur, 2019.
\newblock URL \url{https://ajnsit.github.io/concur}.
\newblock Accessed on 2019-09-02.

\bibitem[King et~al.(2019)King, Ng, and Yoshida]{KingNY19:mpst-web}
Jonathan King, Nicholas Ng, and Nobuko Yoshida.
\newblock Multiparty session type-safe web development with static linearity.
\newblock In \emph{PLACES@ETAPS}, volume 291 of \emph{{EPTCS}}, pages 35--46,
  2019.

\bibitem[Kobayashi(2002)]{Kobayashi02:type-systems}
Naoki Kobayashi.
\newblock Type systems for concurrent programs.
\newblock In \emph{10th Anniversary Colloquium of {UNU/IIST}}, volume 2757 of
  \emph{Lecture Notes in Computer Science}, pages 439--453. Springer, 2002.

\bibitem[Kobayashi(2006)]{Kobayashi06:df}
Naoki Kobayashi.
\newblock A new type system for deadlock-free processes.
\newblock In \emph{{CONCUR}}, volume 4137 of \emph{Lecture Notes in Computer
  Science}, pages 233--247. Springer, 2006.

\bibitem[Krasner and Pope(1988)]{KrasnerP88:mvc}
Glenn~E. Krasner and Stephen~T. Pope.
\newblock A {C}ookbook for using the {M}odel-view {C}ontroller user interface
  paradigm in {S}malltalk-80.
\newblock \emph{J. Object Oriented Program.}, 1\penalty0 (3):\penalty0 26--49,
  August 1988.
\newblock ISSN 0896-8438.
\newblock URL \url{http://dl.acm.org/citation.cfm?id=50757.50759}.

\bibitem[Lindley and Morris(2015)]{LindleyM15:semantics}
Sam Lindley and J.~Garrett Morris.
\newblock A semantics for propositions as sessions.
\newblock In \emph{{ESOP}}, volume 9032 of \emph{Lecture Notes in Computer
  Science}, pages 560--584. Springer, 2015.

\bibitem[Lindley and Morris(2016)]{LindleyM16:bananas}
Sam Lindley and J.~Garrett Morris.
\newblock Talking bananas: structural recursion for session types.
\newblock In \emph{{ICFP}}, pages 434--447. {ACM}, 2016.

\bibitem[Lindley and Morris(2017)]{LindleyM17:fst}
Sam Lindley and J~Garrett Morris.
\newblock Lightweight functional session types.
\newblock \emph{Behavioural Types: from Theory to Tools. River Publishers},
  pages 265--286, 2017.

\bibitem[Meyerovich et~al.(2009)Meyerovich, Guha, Baskin, Cooper, Greenberg,
  Bromfield, and Krishnamurthi]{MeyerovichGBCGBK09:flapjax}
Leo~A. Meyerovich, Arjun Guha, Jacob~P. Baskin, Gregory~H. Cooper, Michael
  Greenberg, Aleks Bromfield, and Shriram Krishnamurthi.
\newblock Flapjax: a programming language for {A}jax applications.
\newblock In \emph{{OOPSLA}}, pages 1--20. {ACM}, 2009.

\bibitem[Mostrous and Vasconcelos(2018)]{MostrousV18:affine}
Dimitris Mostrous and Vasco~T. Vasconcelos.
\newblock Affine sessions.
\newblock \emph{Logical Methods in Computer Science}, 14\penalty0 (4), 2018.
\newblock \doi{10.23638/LMCS-14(4:14)2018}.
\newblock URL \url{https://doi.org/10.23638/LMCS-14(4:14)2018}.

\bibitem[Niehren et~al.(2006)Niehren, Schwinghammer, and
  Smolka]{NiehrenSS06:concurrent}
Joachim Niehren, Jan Schwinghammer, and Gert Smolka.
\newblock A concurrent lambda calculus with futures.
\newblock \emph{Theor. Comput. Sci.}, 364\penalty0 (3):\penalty0 338--356,
  2006.

\bibitem[Padovani(2014)]{Padovani14:df-pi}
Luca Padovani.
\newblock Deadlock and lock freedom in the linear {$\pi$}-calculus.
\newblock In \emph{{CSL-LICS}}, pages 72:1--72:10. {ACM}, 2014.

\bibitem[Padovani(2017)]{Padovani17:context-free}
Luca Padovani.
\newblock Context-free session type inference.
\newblock In \emph{{ESOP}}, volume 10201 of \emph{Lecture Notes in Computer
  Science}, pages 804--830. Springer, 2017.

\bibitem[Padovani and Novara(2015)]{PadovaniN15:df-lambda}
Luca Padovani and Luca Novara.
\newblock Types for deadlock-free higher-order programs.
\newblock In \emph{{FORTE}}, volume 9039 of \emph{Lecture Notes in Computer
  Science}, pages 3--18. Springer, 2015.

\bibitem[Pedley(2019)]{mvu-flutter}
Adam Pedley.
\newblock Functional {M}odel-{V}iew-{U}pdate {A}rchitecture for {F}lutter,
  2019.
\newblock URL
  \url{https://buildflutter.com/functional-model-view-update-architecture-for-flutter/}.
\newblock Accessed on 2019-09-24.

\bibitem[Serrano and Prunet(2016)]{SerranoP16:hopjs}
Manuel Serrano and Vincent Prunet.
\newblock A glimpse of {H}opjs.
\newblock In \emph{{ICFP}}, pages 180--192. {ACM}, 2016.

\bibitem[{The PureScript Contributors}(2019)]{purescript}
{The PureScript Contributors}.
\newblock {P}ure{S}cript, 2019.
\newblock URL \url{http://www.purescript.org/}.
\newblock Accessed on 2019-09-02.

\bibitem[Vasconcelos(2012)]{Vasconcelos12:fundamentals}
Vasco~T. Vasconcelos.
\newblock Fundamentals of session types.
\newblock \emph{Inf. Comput.}, 217:\penalty0 52--70, 2012.

\bibitem[Wadler(1990)]{Wadler90:linear-types}
Philip Wadler.
\newblock Linear types can change the world!
\newblock In \emph{Programming Concepts and Methods}, page 561. North-Holland,
  1990.

\bibitem[Wadler(2014)]{Wadler14:propositions-sessions}
Philip Wadler.
\newblock Propositions as sessions.
\newblock \emph{J. Funct. Program.}, 24\penalty0 (2-3):\penalty0 384--418,
  2014.
\newblock \doi{10.1017/S095679681400001X}.
\newblock URL \url{https://doi.org/10.1017/S095679681400001X}.

\bibitem[Weisenburger et~al.(2018)Weisenburger, K{\"{o}}hler, and
  Salvaneschi]{WeisenburgerKS18:scalaloci}
Pascal Weisenburger, Mirko K{\"{o}}hler, and Guido Salvaneschi.
\newblock Distributed system development with {S}cala{L}oci.
\newblock \emph{{PACMPL}}, 2\penalty0 ({OOPSLA}):\penalty0 129:1--129:30, 2018.

\end{thebibliography}

\clearpage
\begin{adjustwidth}{-2cm}{-2cm}

\setlength\linenumbersep{2.75cm}

\appendix
{
\topskip0pt
\vspace*{\fill}
\begin{center}
\Huge Appendices
\end{center}
\vspace*{\fill}
}
\clearpage

\floatstyle{plain}
\restylefloat{figure}
\section{Syntactic Sugar}\label{appendix:syntactic-sugar}

\begin{figure}[t]
  {\footnotesize
  ~Syntactic Sugar
  \begin{syntax}
    \text{Terms} & L, M, N & ::= & \cdots \midspace \htmlterm{\seq{H}} \midspace
    \attrterm{\seq{a}} \\
    \text{HTML} & H & ::= & \htmltag{t}{\seq{a}}{\seq{H}}
    \midspace s \midspace \antiquote{M} \\
  \end{syntax}
  \[
  \begin{array}{lrclp{5em}lrcl}
    \text{Attributes} & a & ::= & ak = b \midspace \antiquote{M} & &
    \text{Attribute Bodies} & b & ::= & s \midspace \antiquote{M}
  \end{array}
  \]

  ~Desugaring of terms \hfill \framebox{$\desugarterm{M}$}
  \begin{mathpar}
    \desugarterm{\htmlterm{\seq{H}}} = \desugarhtml{\seq{H}}

    \desugarterm{\attrterm{\seq{a}}} = \desugarattr{\seq{a}}
  \end{mathpar}

  \begin{minipage}[t]{0.45\textwidth}
    ~Desugaring of HTML \hfill \framebox{$\vphantom{\seq{H}}\desugarhtml{H}$} \framebox{$\desugarhtml{\seq{H}}$}
  \[
    \begin{array}{rcl}
      \desugarhtml{\antiquote{M}} & = & \desugarterm{M} \\
      \desugarhtml{\htmltag{t}{\seq{a}}{\seq{H}}} & = &
      \coretag{t}{\desugarattr{\seq{\vphantom{H}a}}}{\desugarhtml{\seq{H}}} \\
      \desugarhtml{s} & = & \htmltext{s} \vspace{0.5em} \\
      \desugarhtml{\epsilon} & = & \htmlempty \\
      \desugarhtml{H_1 \cdot \ldots \cdot H_n} & = & \desugarhtml{H_1} \star \ldots \star \desugarhtml{H_n} \\
    \end{array}
  \]
\end{minipage}
\hfill
\begin{minipage}[t]{0.45\textwidth}
  ~Desugaring of attributes \hfill \framebox{$\desugarattr{a}$}
  \framebox{$\desugarattr{\seq{a}}$}
  \[
    \begin{array}{rcl}
      \desugarattr{\var{at} = s} & = & \attr{\var{at}}{s} \\
      \desugarattr{h = \antiquote{M}} & = & \attr{h}{\desugarterm{M}} \\ \vspace{0.5em} \\
      \desugarattr{\epsilon} & = & \attrempty \\
      \desugarattr{a_1 \cdot \ldots \cdot a_n} & = & \desugarhtml{a_1} \star \ldots \star \desugarhtml{a_n} \\
    \end{array}
  \]%
\end{minipage}

  \caption{Syntactic sugar and desugaring}
  \label{fig:syntactic-sugar}
}
\end{figure}

\begin{figure}
~Additional term typing rules \hfill \framebox{$\Gamma \vdash M : A$}
\begin{mathpar}
    \inferrule
    [T-Html]
    { (\Gamma \vdash H_i \produces{A})_i }
    { \Gamma \vdash \htmlterm{\seq{H}} \oftype \htmlty{A} }

    \inferrule
    [T-Attr]
    { (\Gamma \vdash a_i \produces{A})_i }
    { \Gamma \vdash \attrterm{\seq{a}} \oftype \attrty{A} }
\end{mathpar}

~Typing rules for attributes \hfill \framebox{$\Gamma \vdash a \produces{A}$}
\begin{mathpar}
  \inferrule
  [TA-Evt]
  { \Gamma \vdash b \oftype \evtty{\eh} \to A }
  { \Gamma \vdash \eh = b \produces{A} }

  \inferrule
  [TA-Attr]
  { \Gamma \vdash b : \stringty }
  { \Gamma \vdash \at = b \produces{A} }
\end{mathpar}

~Typing rules for HTML \hfill \framebox{$\Gamma \vdash H \produces{A}$}

\small
\begin{mathpar}
  \inferrule
  [TH-Tag]
  { (\Gamma \vdash a_i \produces{A} )_i \\\\
    (\Gamma \vdash H_i \produces{A})_i  }
    { \Gamma \vdash \htmltag{\tagname{t}}{\seq{a}}{\seq{H}} \produces{A} }

  \inferrule
  [TH-Text]
  { }
  { \Gamma \vdash s \produces{A} }

  \inferrule
  [TH-Antiquote]
  { \Gamma \vdash M \oftype \htmlty{A} }
  { \Gamma \vdash \antiquote{M} \produces{A} }
\end{mathpar}

\caption{Typing rules for sugared terms, HTML, and attributes}
\label{fig:sugar-typing}
\end{figure}

We introduce two new terms: $\htmlterm{\seq{H}}$, representing
sugared HTML elements $\seq{H}$, and $\attrterm{\seq{a}}$, representing sugared HTML
attributes $\seq{a}$.
An \emph{antiquoted expression} $\antiquote{M}$ allows a term $M$ to
be embedded within an HTML tree.
Sugared HTML $H$ consists of HTML tags and
antiquoted expressions. HTML tags $\htmltag{t}{\seq{a}}{\seq{H}}$ describe an HTML tag
with name \tagname{t}, sugared attributes $\seq{a}$, and child elements
$\seq{H}$.
Sugared attributes $a$ are either a pair $\ak = b$ of an attribute key $\ak$ and
an attribute body $b$ or an antiquoted expression.
An attribute body $b$ is either a string or an
antiquoted expression.

\subparagraph{Desugaring.}
Desugaring translates HTML tags and attributes to the relevant underlying core
construct: for instance, HTML tags $\htmltag{t}{\seq{a}}{\seq{H}}$ are
desugared into $\coretag{t}{\desugarattr{\seq{a}}}{\desugarhtml{\seq{H}}}$,
where $\desugarattr{\seq{a}}$ and $\desugarhtml{\seq{H}}$ are the monoidal
compositions of the translations of attributes and HTML elements respectively.

Figure~\ref{fig:syntactic-sugar} shows the desugaring transformation from
sugared HTML to the core \mvu monoidal constructs. Antiquoted expressions
$\antiquote{M}$ are
translated to the translation of $M$. HTML tags are translated to the
$\calcwd{htmlTag}$ construct, and attributes $ak = b$ are translated to the
$\calcwd{attr}$ construct. String literals are translated to themselves. Empty
sequences of HTML elements and attributes are translated to $\htmlempty$ and
$\attrempty$ respectively, and sequences are translated to use the monoidal
concatenation operator $\star$.

Figure~\ref{fig:sugar-typing} shows the typing rules for the syntactic sugar.
We extend the standard term typing judgement $\Gamma \vdash M : A$, and
introduce judgements $\Gamma \vdash H
\produces{A}$ and $\Gamma \vdash a \produces{A}$ which state that HTML elements $H$
(resp.\ attributes $a$) \emph{produce messages of type} $A$.

Rule \textsc{T-Html} assigns term $\htmlterm{\seq{H}}$ the type $\htmlty{A}$ if HTML
elements $\seq{H}$ all produce messages of type $A$. Similarly, \textsc{T-Attr} assigns
term $\attrterm{\seq{a}}$ the type $\attrty{A}$ if the attributes $\seq{a}$
produce messages of type $A$.

Event handler attributes map an event handler name $h$ to a function of
type $\evtty{h} \to A$, recalling that $\evtty{-}$ is a meta-level
function mapping event handler names to payload types.
Key-value attributes do not produce messages, so can be typed as
producing any arbitrary type of message $A$.  An event handler attribute $a$
produces messages of type $A$ if the result type of its event handling function is
$A$.
An HTML element $H$ produces messages of type $A$ if all attributes in the tree
produce messages of type $A$.

It is straightforward to see that the desugaring translation preserves typing.

\begin{lemma}[Correctness of desugaring]
\begin{enumerate}
\item If $\Gamma \vdash M : A$, then $\Gamma \vdash \desugarterm{M} : A$
\item If $\Gamma \vdash H \produces A$, then $\Gamma \vdash \desugarhtml{H} : \htmlty{A}$
\item If $\Gamma \vdash a \produces A$, then $\Gamma \vdash \desugarattr{a} : \attrty{A}$
\end{enumerate}
\end{lemma}

\begin{proof}
  By mutual induction on the derivations of $\Gamma \vdash M : A$ and $\Gamma
  \vdash H \produces A$ and $\Gamma \vdash a \produces A$.

  For $\Gamma \vdash M : A$, the only interesting cases are \textsc{T-Html} and
  \textsc{T-Attr}: these follow by (2) and (3) respectively, followed by
  translation on sequences; the remaining cases are either true by definition or
  by appeal to the induction hypothesis on subterms.

  For $\Gamma \vdash H \produces{A}$:
  \begin{proofcase}{TH-Tag}
    \begin{mathpar}
    \inferrule*
    { (\Gamma \vdash a_i \produces{A} )_i \\
      (\Gamma \vdash H_i \produces{A})_i  }
    { \Gamma \vdash \htmltag{\tagname{t}}{\seq{a}}{\seq{H}} \produces{A} }
  \end{mathpar}

  By induction hypothesis (3), we have that $(\Gamma \vdash \translate{a_i} : \attrty{A})_i$.
  By the translation on attribute sequences:
  \begin{itemize}
    \item if $\seq{a} = \cdot$ then $\translate{\cdot} = \attrempty$ and $\Gamma \vdash \attrempty : \attrty{A}$
    \item if $\seq{a} = a_1 \star \ldots \star a_n$, then $\translate{\seq{a}} =
      \translate{a_1} \star \ldots \star \translate{a_n}$ and $\Gamma \vdash \translate{\seq{a}} : \attrty{A}$.
  \end{itemize}

  By exactly the same reasoning using induction hypothesis (2), we have that $\Gamma \vdash \translate{\seq{H}} : \htmlty{A}$.

  Thus, we can show:
  \begin{mathpar}
    \inferrule*
    { \Gamma \vdash \translate{\seq{a}} : \attrty{A} \\
      \Gamma \vdash \translate{\seq{H}} : \htmlty{A}
    }
    { \Gamma \vdash \coretag{t}{\translate{\seq{a}}{\translate{\seq{H}}} : \htmlty{A} } }
  \end{mathpar}
  as required.
  \end{proofcase}

  \begin{proofcase}{TH-Text}
    By \textsc{T-String}, string literals $s$ have type $\stringty$. By \textsc{T-HtmlText},
    we can show $\Gamma \vdash \htmltext{s} : \htmlty{A}$ for any $A$ as required.
  \end{proofcase}

  \begin{proofcase}{TH-Antiquote}
    Assumption:
    \begin{mathpar}
      \inferrule*
      { \Gamma \vdash M : \htmlty{A} }
      { \Gamma \vdash \antiquote{M} \produces A }
    \end{mathpar}

    By induction hypothesis (1), we have that $\Gamma \vdash \translate{M} : \htmlty{A}$,
    as required.
  \end{proofcase}

  For $\Gamma \vdash a \produces{A}$:
  \begin{proofcase}{TA-Evt}
    Assumption:
    \begin{mathpar}
      \inferrule*
      { \Gamma \vdash b : \evtty{h} \to A }
      { \Gamma \vdash h = b \produces{A} }
    \end{mathpar}

    By inversion on the typing relation, it must be the case that $b =
    \antiquote{M}$ for some term $M$ where $M : \evtty{h} \to A$.
    By induction hypothesis (1), we have that $\Gamma \vdash M : \evtty{h} \to
    A$.

    Thus, we can show:

    \begin{mathpar}
      \inferrule
      { \Gamma \vdash \translate{M} : \evtty{h} \to A }
      { \Gamma \vdash \attr{h}{\translate{M}} : \attrty{A} }
    \end{mathpar}

    as required.
  \end{proofcase}
  \begin{proofcase}{TA-Attr}
    Assumption:
    \begin{mathpar}
      \inferrule
      { \Gamma \vdash b : \stringty }
      { \Gamma \vdash \at = b \produces{A}  }
    \end{mathpar}

    By case analysis on $b$ and inversion on the typing relation, either $b = s$
    for some string literal $s$, or $b = \antiquote{M}$ for some term $M$ such
    that $\Gamma \vdash M : \stringty$.

    If $b = s$ for some string literal $s$, then we can show:
    \begin{mathpar}
      \inferrule*
      {
        \inferrule*
        { }
        { \Gamma \vdash s : \stringty }
      }
      { \Gamma \vdash \attr{\at}{s} : \attrty{A}}
    \end{mathpar}
  \end{proofcase}

  If $b = \antiquote{M}$ for some term $M$ such that $\Gamma \vdash M :
  \stringty$, then by induction hypothesis (1) $\Gamma \vdash \translate{M} :
  \stringty$, and thus we can show:
  \begin{mathpar}
    \inferrule*
    { \Gamma \vdash \translate{M} : A }
    { \Gamma \vdash \attr{\at}{\translate{M}} : \attrty{A} }
  \end{mathpar}
  as required.

\end{proof}
 \clearpage
\section{Omitted details}\label{sec:appendix:omitted}

\subsection{Simply-typed $\lambda$-calculus}
\begin{figure}
~\textbf{Typing rules} \hfill \framebox{$\Gamma \vdash M : A$}
\begin{mathpar}
    \inferrule
    [T-Var]
    { x \oftype A \in \Gamma }
    { \Gamma \vdash x \oftype A }

    \inferrule
    [T-Abs]
    { \Gamma, x \oftype A \vdash M \oftype B }
    { \Gamma \vdash \lambda x . M \oftype A \to B }

    \inferrule
    [T-Rec]
    { \Gamma, f : A \to B, x : A \vdash M : B }
    { \Gamma \vdash \rec{f}{x}{M} : A \to B }

    \inferrule
    [T-App]
    { \Gamma \vdash M \oftype A \to B \\ \Gamma \vdash N \oftype A }
    { \Gamma \vdash M \app N \oftype B }
    \\
    \inferrule
    [T-Unit]
    { }
    { \Gamma \vdash () \oftype \one}

    \inferrule
    [T-Pair]
    { \Gamma \vdash M \oftype A \\ \Gamma \vdash N \oftype B }
    { \Gamma \vdash (M, N) \oftype A \times B }

    \inferrule
    [T-LetPair]
    { \Gamma \vdash M \oftype A \times B \\ \Gamma, x \oftype A, y \oftype
    B \vdash N \oftype C }
    { \Gamma \vdash \letin{(x, y)}{M}{N} \oftype C }
    \\
    \inferrule
    [T-Inl]
    { \Gamma \vdash M \oftype A }
    { \Gamma \vdash \inl{M} \oftype A + B }

    \inferrule
    [T-Inr]
    { \Gamma \vdash M \oftype B }
    { \Gamma \vdash \inr{M} \oftype A + B }

    \inferrule
    [T-Case]
    { \Gamma \vdash L \oftype A + B \\\\
        \Gamma, x \oftype A \vdash M \oftype C \\
        \Gamma, y \oftype B \vdash N \oftype C }
    { \Gamma \vdash
        \caseof{L}{ \inl{x} \mapsto M; \inr{y} \mapsto N} \oftype C } \\

    \inferrule
    [T-Int]
    { n \in \mathbb{N} }
    { \Gamma \vdash n : \intty }

    \inferrule
    [T-String]
    { }
    { \Gamma \vdash s : \stringty }
\end{mathpar}
~\textbf{Runtime syntax}
\[
  \begin{array}{lccl}
\text{Evaluation contexts} & E & ::= & [~] \midspace E \app M \midspace V \app E \\
                               &   & \midspace & (E, M) \midspace (V, E) \midspace \letin{(x, y)}{E}{M} \\
                               &   & \midspace & \inl{E} \midspace \inr{E} \midspace \caseof{E}{\inl{x} \mapsto M; \inr{y} \mapsto N} \\
                               &   & \midspace & \coretag{t}{E}{M} \midspace \coretag{t}{V}{E} \midspace
                               \htmltext{E} \\
                               &   & \midspace & \htmltext{E} \midspace
                               \attr{\ak}{E} \midspace \append{E}{M} \midspace
                               \append{V}{E}
\end{array}
\]

~\textbf{Reduction on terms} \hfill \framebox{$M \teval N$}
  \[
    \begin{array}{lrcl}
        \textsc{E-Lam}   & (\lambda x . M) \app V & \:\: \teval \:\: & M \{ V / x \} \\
        \textsc{E-Rec}   & (\rec{f}{x}{M}) \app V & \:\: \teval \:\: & M \{ \rec{f}{x}{M} / f, V / x \} \\
        \textsc{E-Pair}  &  \letin{(x, y)}{(V, W)}{M} & \teval & M \{ V / x, W / y \} \\
        \textsc{E-Inl}   & \quad \caseof{\inl{V}}{\inl{x} \mapsto M; \inr{y} \mapsto N} & \teval & M \{ V / x \} \\
        \textsc{E-Inr}   & \quad \caseof{\inr{V}}{\inl{x} \mapsto M; \inr{y} \mapsto N} & \teval & N \{ V / y \} \\
        \textsc{E-LiftM}  & E[M] & \teval & E[N], \quad \text{if } M \teval N
    \end{array}
  \]%
  \caption{Typing and $\beta$-reduction rules for standard $\lambda$-calculus
  terms}
  \label{fig:appendix:stlc}
\end{figure}
Figure~\ref{fig:appendix:stlc} shows the typing and $\beta$-reduction rules for
the simply-typed $\lambda$-calculus.
Evaluation contexts $E$ are set up for standard call-by-value, left-to-right
evaluation.

\subsection{Equivalence}
\begin{mathpar}
P_1 \parallel P_2  \equiv P_2 \parallel P_1

P_1 \parallel (P_2 \parallel P_3)  \equiv (P_1 \parallel P_2) \parallel P_3

\sys{P}{D} \equiv \sys{P'}{D} \quad \text{if } P \equiv P'
\end{mathpar}

\subsection{Typing rules for pages}

\begin{figure}
~\textbf{Typing rules for pages} \hfill \framebox{$\Gamma \vdash D : \pagety{A}$}
\begin{mathpar}
  \inferrule
  [TD-Tag]
  { (\vdash e_i)_i \\ \Gamma \vdash V : \attrty{A} \\ \Gamma \vdash D : \pagety{A} }
  { \Gamma \vdash \pgtag{t}{V}{D}{\seq{e}} : \pagety{A} }

  \inferrule
  [TD-Empty]
  { }
  { \Gamma \vdash \htmlempty : \pagety{A} }

  \inferrule
  [TD-Text]
  { \Gamma \vdash V : \stringty }
  { \Gamma \vdash \htmltext{V} : \pagety{A} }

  \inferrule
  [TD-Append]
  { \Gamma \vdash D_1 : \pagety{A} \\ \Gamma \vdash D_2 : \pagety{A}    }
  { \Gamma \vdash \append{D_1}{D_2} : \pagety{A} }
\end{mathpar}
\caption{Typing rules for pages}
\label{fig:omitted:page-typing}
\end{figure}

Figure~\ref{fig:omitted:page-typing} shows the typing rules for pages, which
mostly follow the typing rules for HTML.
 \clearpage
\section{Proofs for Section~\ref{sec:formalism}}\label{appendix:proofs}

\subsection{Preservation}

\begin{lemma}[Erasure]\label{lem:erasure}
  If $\vdash D : \pagety{A}$, then $\cdot \vdash \erase{D}: \htmlty{A}$.
\end{lemma}
\begin{proof}
  By induction on the derivation of $\vdash D : \pagety{A}$.
\end{proof}

\begin{lemma}[Diffing]\label{lem:diffing}
  If:
  \begin{itemize}
    \item $\cdot \vdash U : \htmlty{A}$
    \item $\vdash D : \pagety{A}$
    \item $\diff{U}{D} = D'$
  \end{itemize}

  then $\vdash D' : \pagety{A}$
\end{lemma}
\begin{proof}
  By the definition of $\diff{U}{D} = D'$, $D'$ is derived from $D$ by adding or removing a
  node, or modifying a node's attributes, and $\erase{D'} = U$.
  By Lemma~\ref{lem:erasure}, since $\vdash D' : \pagety{A}$ and $\erase{D'} = U$, it follows that $\vdash \erase{D'} : \htmlty{A}$.

  Since $\vdash \erase{D'} : \htmlty{A}$, any nodes which are added must
  have attributes of type $\attrty{A}$, and since event queues have no bearing
  on type parameters, it follows that $\vdash D' : \pagety{A}$.
\end{proof}

\begin{lemma}[Subterm typeability (Page contexts)]
  \label{lem:dom-ctx-typeability}
  If $\deriv{D}$ is a derivation of $\vdash \pgctx{D} : \pagety{A}$,
  there exists some subderivation $\deriv{D}'$ concluding $\vdash D :
  \pagety{A}$ and where the position of $\deriv{D}'$ in $\deriv{D}$ corresponds to
  that of the hole in $\config{D}$.
\end{lemma}
\begin{proof}
  By induction on the structure of $\pgctx{D}$.
\end{proof}

\begin{lemma}[Subterm replacement (Page contexts)]
  \label{lem:dom-ctx-replacement}
  If:
  \begin{itemize}
    \item $\deriv{D}$ is a derivation of $\vdash \pgctx{D} : \pagety{A}$,
    \item there exists some subderivation $\deriv{D}'$ concluding $\vdash D :
  \pagety{A}$
    \item the position of $\deriv{D}'$ in $\deriv{D}$ corresponds to that of the
      hole in $\config{D}$
    \item $\vdash D' : \pagety{A}$
  \end{itemize}
  then $\vdash \pgctx{D'} : \pagety{A}$.
\end{lemma}
\begin{proof}
  By induction on the structure of $\pgctx{D}$.
\end{proof}

\begin{lemma}[Preservation (Process equivalence)] \label{lem:equiv-pres}
  If $\vdash^\phi P : A$ and $P \equiv P'$, then $\vdash^\phi P' : A$.
\end{lemma}
\begin{proof}
  By induction on the derivation of $P \equiv P'$.
\end{proof}

\begin{lemma}[Preservation (Configuration equivalence)] \label{lem:config-equiv-pres}
  If $\vdash \config{C}$ and $\config{C} \equiv \config{C}'$, then $\vdash \config{C}'$.
\end{lemma}
\begin{proof}
  Immediate by Lemma~\ref{lem:equiv-pres}.
\end{proof}

\begin{lemma}[Preservation (Terms)]\label{lem:term-pres}
  If $\Gamma \vdash M : A$ and $M \teval N$, then $\Gamma \vdash N : A$.
\end{lemma}
\begin{proof}
Standard; by induction on the derivation of $M \teval N$.
\end{proof}

\begin{lemma}[Preservation (Processes)]\label{lem:process-pres}
  If $\vdash^\phi P : A$ and $P \ceval P'$, then $\vdash^\phi P' : A$.
\end{lemma}
\begin{proof}
By induction on the derivation of $P \ceval P'$.
\begin{proofcase}{E-Handle}

  \[
    \handlerproc{\idle{V_m}}{F} \parallel \thread{V} \ceval
    \handlerproc{\handle{V_m, F, V}}{F}
  \]

  Assumption:
  \begin{mathpar}
    \inferrule*
    {
      \inferrule*
      {
        \inferrule*
        { \cdot \vdash V_m : A }
        { \vdash \idle{V_m} : \evtloopty{A}{B} } \\
        {\bl
        \cdot \vdash V_v : A \to \htmlty{B} \\
        \cdot \vdash V_u : (B \times A) \to A \\
        \el}
      }
      { \vdash^\bcirc \handlerprocexp{\idle{V_m}}{V_v}{V_u} : B }
      \\
      \inferrule*
      { \cdot \vdash V : B }
      { \vdash^\wcirc \thread{V} : B }
    }
    {
      \vdash^\bcirc \handlerprocexp{\idle{V_m}}{V_v}{V_u} \parallel \thread{V} : B
    }
  \end{mathpar}

  Let $\deriv{D}'$ be the derivation:
  \begin{mathpar}
    \inferrule*
    {
      \cdot \vdash V_u : (B \times A) \to A \\
      \inferrule*
      { \cdot \vdash V : B \\ \cdot \vdash V_m : A }
      { \cdot \vdash (V, V_m) : (B \times A) }
    }
    { \cdot \vdash V_u \app (V, V_m) : A }
  \end{mathpar}

  Let $\deriv{D}$ be the derivation:
  \begin{mathpar}
    \inferrule*
    {
      \deriv{D}' \\
      \inferrule*
      {
        \inferrule*
        { }
        { m' : A \vdash m' : A} \\
        \inferrule*
        { m' : A \vdash V_v : A \to \htmlty{B} \\
          m' : A \vdash m' : A
        }
        { m' : A \vdash V_v \app m' : \htmlty{B}}
      }
      { m' : A \vdash (m', V_v \app m') : (A \times \htmlty{B}) }
    }
    { \cdot \vdash \letin{m'}{V_u \app (V, V_m)}{(m', V_v \app m')} : (A \times \htmlty{B})}
  \end{mathpar}

  Then we can show:
  {\small
  \begin{mathpar}
    \inferrule*
    {
      \inferrule*
      { \deriv{D} }
      { \vdash
          {\begin{aligned}[t]
              & \letintwo{m'}{V_u \app (V, V_m)} \\
              & (m', V_v \app m')
          \end{aligned}}
        : \evtloopty{A}{B} } \\
      \cdot \vdash V_v : (B \times A) \to A \\
      \cdot \vdash V_u : A \to \htmlty{B}
    }
    { \vdash^\bcirc \handlerprocexp{
      \begin{aligned}[t]
        & \letintwo{m'}{V_u \app (V, V_m)} \\
        & (m', V_v \app m')
      \end{aligned}}{V_v}{V_u} : B }
  \end{mathpar}
  }
  as required.

\end{proofcase}

\begin{proofcase}{E-LiftT}

  Immediate by Lemma~\ref{lem:term-pres}.
\end{proofcase}

\begin{proofcase}{E-Par}
Immediate by the IH.
\end{proofcase}

\begin{proofcase}{E-Struct}
  Immediate by Lemma~\ref{lem:equiv-pres} and the induction hypothesis.
\end{proofcase}

\end{proof}

\begin{fake}{Theorem~\ref{thm:config-pres} (Preservation (Configurations))}
  If $\vdash \config{C}$ and $\config{C} \ceval \config{C}'$, then
  $\vdash \config{C}'$.
\end{fake}
\begin{proof}
By induction on the derivation of $\config{C} \ceval \config{C}'$.

\begin{proofcase}{E-Interact}
  \[
\sys{P}{\config{D}[\pgtag{t}{U}{D}{\seq{e}}]} \ceval
    \sys{P}{\config{D}[\pgtag{t}{U}{D}{\seq{e} \cdot \evtpayload{\evt{ev}}{V}}]}
  \]
  for some $\evt{ev}$, $V$ such that $\vdash \evtpayload{\evt{ev}}{V}$.

  Assumption:
  \begin{mathpar}
    \inferrule*
    { \vdash^\bcirc P : A \\ \vdash
      \config{D}[\pgtag{t}{U}{D}{\seq{e}}] : \pagety{A} }
    { \vdash \sys{P}{\config{D}[\pgtag{t}{U}{D}{\seq{e}}]} }
  \end{mathpar}

  By Lemma~\ref{lem:dom-ctx-typeability}:
  \begin{mathpar}
    \inferrule*
    { \cdot \vdash U : \attrty{A} \\ \vdash D : \pagety{A} \\ (\vdash
    e_i)_i}
    { \vdash \pgtag{t}{U}{D}{\seq{e}} : \pagety{A} }
  \end{mathpar}

  As $\vdash \evtpayload{\evt{ev}}{V}$:
  \begin{mathpar}
    \inferrule*
    { \cdot \vdash U : \attrty{A} \\ \vdash D : \pagety{A} \\ (\vdash
    e_i)_i \\ \vdash \evtpayload{\evt{ev}}{V} }
    {  \vdash \pgtag{t}{U}{D}{\seq{e} \cdot \evtpayload{\evt{ev}}{V}} :
    \pagety{A} }
  \end{mathpar}

  By Lemma~\ref{lem:dom-ctx-replacement}, $\vdash
  \sys{P}{\config{D}[\pgtag{t}{U}{D}{\seq{e}\cdot \evtpayload{\evt{ev}}{V}}]}$,
  thus we have that:

  \begin{mathpar}
    \inferrule*
    { \vdash^\bcirc P : A \\ \vdash
      \config{D}[\pgtag{t}{U}{D}{\seq{e}\cdot \evtpayload{\evt{ev}}{V}}] : \pagety{A} }
    { \vdash \sys{P}{\config{D}[\pgtag{t}{U}{D}{\seq{e}\cdot \evtpayload{\evt{ev}}{V}}]} }
  \end{mathpar}

 as required.
\end{proofcase}

\begin{proofcase}{E-Evt}

  \[
\sys{P}{\config{D}[\pgtag{t}{U}{D}{\evtpayload{\evt{ev}}{W} \cdot \seq{e}}]}
 \ceval
 \sys{P \parallel \thread{V_1 \app W} \parallel \cdots \parallel \thread{V_n \app W}}{\config{D}[\pgtag{t}{U}{D}{\seq{e}}]}
  \]%
  where $\handlers{\ev}{U} = \seq{V}$.

  Assumption:
  \begin{mathpar}
    \inferrule
    {  \vdash^\bcirc P : A \\
       \vdash \config{D}[\pgtag{t}{U}{D}{\evtpayload{\evt{ev}}{W} \cdot
      \seq{e}}] : \pagety{A} \\
    }
    { \vdash \sys{P}{\config{D}[\pgtag{t}{U}{D}{\evtpayload{\evt{ev}}{W} \cdot \seq{e}}]} }
  \end{mathpar}

  By Lemma~\ref{lem:dom-ctx-typeability}:
  \begin{mathpar}
    \inferrule*
    { \cdot \vdash U : \attrty{A} \\ \vdash D : \pagety{A}  \\
      \inferrule*
      { \vdash W : \evtty{ev} }
      { \vdash \evtpayload{ev}{W} } \\
      (\vdash e_i)_i
    }
    { \vdash \pgtag{t}{U}{\vh}{\evtpayload{\evt{ev}}{W} \cdot \seq{e}}  : \pagety{A} }
  \end{mathpar}

  We can straightforwardly show:
  \begin{mathpar}
    \inferrule*
    { \cdot \vdash U : \attrty{A} \\ \vdash D : \pagety{A}  \\
      (\vdash e_i)_i
    }
    { \vdash \pgtag{t}{U}{\vh}{\seq{e}}  : \pagety{A} }
  \end{mathpar}

  By Lemma~\ref{lem:dom-ctx-replacement}, we therefore have that $\vdash
  \pgctx{\pgtag{t}{U}{\vh}{\seq{e}}} : \pagety{A} $

  Knowing that $U : \attrty{A}$, by \textsc{T-EvtAttr} we have that each
  $V_i$ in $\handlers{\ev}{U}$ has type $\evtty{h} \to A$. Since by the definition
  of $\handlers{\ev}{U}$ we have that $h = \handler{\ev}$ and that there
  is a bijective mapping between event names and handlers, we have that $\evtty{h} =
  \evtty{\ev}$. Thus, we have that for each $V_i \in \seq{V}$, $\cdot \vdash V_i : \evtty{\ev} \to A$.

  For each $V_i$, we can therefore show:
  \begin{mathpar}
    \inferrule*
    {
      \inferrule*
      { \cdot \vdash V_i : \evtty{ev} \to A \\ \cdot \vdash W : \evtty{ev} }
      { \cdot \vdash V_i \app W : A }
    }
    { \vdash^\wcirc \thread{V_i \app W} : A }
  \end{mathpar}

  By \textsc{TP-Par}, we have that $\vdash^\wcirc \thread{V_1 \app W}
  \parallel \cdots \parallel \thread{V_n \app W} : A$.

  Recomposing:
  \begin{mathpar}
    \inferrule
    {
      \inferrule*
      {
        \vdash^\bcirc P : A \\
        \vdash^\wcirc \thread{V_1 \app W} \parallel \cdots \parallel
        \thread{V_n \app W} : A
      }
      { \vdash^\bcirc P \parallel \thread{V_1 \app W} \parallel \cdots
      \parallel \thread{V_n \app W} : A} \\
      \vdash \pgctx{\pgtag{t}{U}{\vh}{\seq{e}}} : \pagety{A}
    }
    { \vdash \sys{P \parallel \thread{V_1 \app W} \parallel \cdots
      \parallel \thread{V_n \app W}}{\pgctx{\pgtag{t}{U}{\vh}{\seq{e}}}} }
  \end{mathpar}
  as required.
\end{proofcase}

\begin{proofcase}{E-Update}

  \[
    \sys{\procctx[\handlerprocexp{(V'_m, U)}{V_v}{V_u}]}{\vh} \ceval
    \sys{\procctx[\handlerprocexp{\idle{V'_m}}{V_v}{V_u}]}{\vh'}
  \]
  where $\diff{U}{\vh} = \vh'$

  We show the illustrative case where $\procctx = [~] \parallel P$; the case for $\procctx = [~]$ is similar.

  \[
  \sys{\handlerprocexp{(V'_m, U)}{V_v}{V_u} \parallel P}{\vh} \ceval
  \sys{\handlerprocexp{\idle{V'_m}}{V_v}{V_u} \parallel P}{\vh'}
  \]
  where $\diff{U}{\vh} = \vh'$

  Let $\deriv{D}$ be the derivation:
  \begin{mathpar}
    \inferrule*
    {
      \inferrule*
      {
        \inferrule*
        { \cdot \vdash V'_m : A \\ \cdot \vdash U : \htmlty{B} }
        { \cdot \vdash (V'_m, U) : (A \times \htmlty{B}) }
      }
      { \vdashs (V'_m, U) : \evtloopty{A}{B} } \\
      {\begin{aligned}[b]
        & \cdot \vdash V_v : A \to \htmlty{B} \\
        & \cdot \vdash V_u : (B \times A) \to A
      \end{aligned}}
    }
    { \vdash^\bcirc \handlerprocexp{(V'_m, U)}{V_v}{V_u} : B }
  \end{mathpar}

  Assumption:
  \begin{mathpar}
    \inferrule*
    {
      \inferrule*
      {
        \deriv{D} \\
        \inferrule*
        { }
        { \vdash^\wcirc P : B }
      }
      { \vdash^\bcirc \handlerprocexp{(V'_m, U)}{V_v}{V_u} \parallel P : B } \\
      \vdash \vh : \pagety{B}
    }
    { \vdash^\bcirc
        \sys{\handlerprocexp{(V'_m, U)}{V_v}{V_u} \parallel P}{\vh} }
  \end{mathpar}

  Let $\deriv{D}'$ be the derivation:
  \begin{mathpar}
    \inferrule*
    {
      \inferrule*
      { \cdot \vdash V'_m : A }
      { \vdash \idle{V'_m} : \evtloopty{A}{B} } \\
      {\begin{aligned}[b]
        & \cdot \vdash V_v : A \to \htmlty{B} \\
        & \cdot \vdash V_u : (B \times A) \to A \\
        & (\cdot \vdash W_i : B)_i
      \end{aligned}}
    }
    { \vdash^\bcirc \handlerprocexp{\idle{V'_m}}{V_v}{V_u} : B }
  \end{mathpar}

  By Lemma~\ref{lem:diffing}, we have that $\vdash D' : \pagety{B}$.

  Recomposing:
  \begin{mathpar}
    \inferrule*
    {
      \inferrule*
      {
        \deriv{D} \\
        \inferrule*
        { }
        { \vdash^\wcirc P : B }
      }
      { \vdash^\bcirc \handlerprocexp{\idle{V'_m}}{V_v}{V_u} \parallel P : B } \\
      \vdash \vh' : \pagety{B}
    }
    { \vdash
    \sys{\handlerprocexp{\idle{V'_m}}{V_v}{V_u} \parallel P}{\vh'} }
  \end{mathpar}

  as required.
\end{proofcase}

\begin{proofcase}{E-Run}

  \[
    \sys{\procctx[\run{(V_m, V_v, V_u)}]}{D} \ceval
    \sys{\procctx[\handlerprocexp{(V_m, V_v \app V_m)}{V_v}{V_u}]}{D}
  \]

  We consider the case where $\procctx = [~]$; the case where $\procctx = [~]
  \parallel P$ is similar.

  Assumption:

  \begin{mathpar}
    \inferrule*
    {
      \inferrule*
      {
        \inferrule*
        { \cdot \vdash V_m : A \\ \cdot \vdash V_v : A \to \htmlty{B} \\
          \cdot \vdash V_u : (B \times A) \to A
        }
        {
          \cdot \vdash (V_m, V_v, V_u) : (A \times A \to \htmlty{B} \times (B \times A) \to A)
        }
      }
      { \vdash^\bcirc \run{(V_m, V_v, V_u)} : B }
      \\
      \vdash D : \pagety{B}
    }
    {
      \vdash \sys{\run{(V_m, V_v, V_u)}}{D} : B
    }
  \end{mathpar}

  Let $\deriv{D}$ be the derivation:
  \begin{mathpar}
    \inferrule*
    {
      \inferrule*
      {
        \inferrule*
        { \cdot \vdash V_m : A \\
          \inferrule*
          { \cdot \vdash V_v : A \to \htmlty{B} \\ \cdot \vdash V_m : A }
          { \cdot \vdash V_v \app V_m : \htmlty{B} }
        }
        { \cdot \vdash (V_m, V_v \app V_m) : (A \times \htmlty{B}) }
      }
      { \vdashs (V_m, V_v \app V_m) : \evtloopty{A}{B} } \\
      {\begin{aligned}[b]
        & \cdot \vdash V_v : A \to \htmlty{B} \\
        & \cdot \vdash V_u : (B \times A) \to A \\
      \end{aligned}}
    }
    { \vdash^\bcirc \handlerprocexp{(V_m, V_v \app V_m)}{V_v}{V_u} : B }
  \end{mathpar}
  Thus we can show:
  \begin{mathpar}
    \inferrule*
    {
      \inferrule*
      { \deriv{D} }
      { \vdash^\bcirc \handlerprocexp{(V_m, V_v \app V_m)}{V_v}{V_u} : B }
      \\
      \vdash \pgempty : \pagety{B} \\
    }
    { \vdash
        \sys{\handlerprocexp{(V_m, V_v \app V_m)}{V_v}{V_u}}{\pgempty} : B
    }
  \end{mathpar}

  as required.
\end{proofcase}

\begin{proofcase}{E-Struct}
  Immediate by Lemma~\ref{lem:config-equiv-pres} and the induction hypothesis.
\end{proofcase}

\begin{proofcase}{E-LiftI}
  Immediate by Lemma~\ref{lem:term-pres}.
\end{proofcase}

\begin{proofcase}{E-LiftP}
  Immediate by Lemma~\ref{lem:process-pres}.
\end{proofcase}

\end{proof}

\subsection{Progress}

\begin{fake}{Lemma~\ref{lem:term-progress}}
If $\cdot \vdash M : A$, then either $M$ is a value, or there exists some $N$ such that $M \teval N$.
\end{fake}
\begin{proof}
Standard; by induction on the derivation of $\cdot \vdash M : A$.
\end{proof}

\begin{lemma}[Progress (Processes)]\label{lem:process-progress}
Suppose $\cdot \vdash^\phi P$. Either there exists some $P'$ such that $P \equiv
\ceval \equiv P'$; or
  \begin{enumerate}
    \item if $\phi = \bcirc$, then either
      $P = \run{V} \parallel \thread{V_1} \parallel \cdots \parallel \thread{V_n}$,
      $P = \handlerprocexp{\idle{V_m}}{V_v}{V_u}$, or
      $P = \handlerprocexp{(V'_m, U)}{V_v}{V_u}$
    \item if $\phi = \wcirc$, then $P \equiv \thread{V_1} \parallel \cdots
      \parallel \thread{V_n}$.
  \end{enumerate}
\end{lemma}
\begin{proof}
  If $\phi = \wcirc$, then by typing on processes, $P$ does not contain an event
  handler process. Thus, $P \equiv \thread{M_1} \parallel \cdots \parallel \thread{M_n}$.
  Since each $M_i$ could reduce via \textsc{E-LiftT}, it must be the case that
  each $M_i$ is a value $V_i$, satisfying (2).

  If $\phi = \bcirc$, then by our previous reasoning either $P$ can reduce or $P
  \equiv P' \parallel \thread{V_1}
  \parallel \cdots \parallel \thread{V_n}$,
  where either $P' = \run{M}$, or $P' = \handlerprocexp{T}{V_v}{V_u}$.
  In the case that $P' = \run{M}$, by Lemma~\ref{lem:term-progress}, either $M$
  can reduce by \textsc{E-LiftT}, or $P' = \run{V}$, satisfying the first clause
  of (1).

  Now consider the case where $P' = \handlerprocexp{T}{V_v}{V_u}$.
  If $n > 0$, then the $P$ can reduce by \textsc{E-Handle}. Therefore, we
  consider the case where $P = \handlerprocexp{T}{V_v}{V_u}$.

  We proceed by case analysis on the active thread. If $T = \idle{V_m}$, then we
  satisfy the second clause of (1).
  If $T = M$, by Lemma~\ref{lem:term-progress}, we have that either $M$ can reduce, or that
  $M = V$ for some value $V$. By \textsc{TP-Process} and \textsc{TS-Processing},
  we have that $\cdot \vdash V : (A \times \htmlty{B})$, and by inversion on the
  typing relation we have that $V = (V', U$), satisfying the third clause of (1).
\end{proof}

\begin{fake}{Theorem~\ref{thm:event-progress} (Event Progress)}
  If $\vdash \config{C}$, then either:
  \begin{enumerate}
    \item there exists some $\config{C}'$ such that $\config{C} \ceval \config{C'}$; or
    \item $\config{C} = \sys{\handlerprocexp{\idle{V_m}}{V_v}{V_u}}{\vh}$ where $\vh$
      cannot be written $\config{D}[\pgtag{t}{V}{W}{\seq{e}}]$ for some
      non-empty $\seq{e}$.
  \end{enumerate}
\end{fake}
\begin{proof}
  By analysis of the derivation of $\cdot \vdash \config{C}$.

   By Lemma~\ref{lem:process-progress}, either:
  \begin{enumerate}
    \item process $P$ can reduce
    \item $P = \run{V} \parallel \thread{V_1} \parallel \cdots \parallel \thread{V_n}$
    \item $P = \handlerprocexp{(V, U)}{V_v}{V_u}$, or
    \item $P = \handlerprocexp{\idle{V_m}}{V_v}{V_u}$
  \end{enumerate}
  In case (1), by \textsc{TP-Run}, $V$ is of the form $(V_m, V_v, V_u)$, and
  we can reduce by \textsc{E-Run}. In case (2), we can reduce by
  \textsc{E-LiftP}. In case (3), we can reduce by \textsc{E-Update}.

  In case (4), if $\vh$ could be written $\config{D}[\pgtag{t}{V}{W}{\seq{e}}]$ for a
  non-empty $\seq{e}$, then the configuration could reduce by \textsc{E-Evt}.
  Otherwise, we satisfy condition (2).
  \end{proof}

\clearpage
\section{Subscriptions}\label{appendix:subscriptions}

For the purposes of the paper, we have described messages which arise as a
result of DOM interaction: for example, a user pressing a button or entering a
message into a text field.

Although orthogonal to our aim of supporting GUI applications making use of
session types, it is useful to consider \emph{subscriptions}, an Elm abstraction
to allow us to handle \emph{environment} events such as timers or mouse
movements.  Let us introduce another event, $\evt{mouseMove}$ with payload type
$(\mkwd{Int} \times \mkwd{Int})$ denoting the co-ordinates of the mouse pointer.

\subparagraph{Example.}
Consider the following example, which shows the co-ordinates of the mouse as it
moves (we assume the existence of an \mkwd{intToString} function):
{\footnotesize
\[
\bl
\mkwd{Model} \defeq (\intty \times \intty) \qquad \mkwd{Message} \defeq
\mkwd{UpdateCoords}(\intty \times \intty) \vspace{0.4em} \\
\mkwd{view} : \mkwd{Model} \to \htmlty{\mkwd{Message}} \\
\mkwd{view} = \lambda (x, y) . \: \calcwd{html} \\
\quad \tagzero{html} \\
\qquad \tagzero{body} \\
\qquad \quad (\antiquote{\htmltext{\intstr{x}}}, \antiquote{\htmltext{\intstr{y}}}) \\
\qquad \tagzeroend{body} \\
\quad \tagzeroend{html} \vspace{0.4em} \\
\mkwd{update} : (\mkwd{Message} \times \mkwd{Model}) \to \mkwd{Model} \\
\mkwd{update} = \lambda (\mkwd{UpdateCoords}(x, y), (\textit{oldX}, \textit{oldY})) . (x, y)
\vspace{0.35em} \\
\mkwd{subscriptions} : \mkwd{Model} \to \subscriptionty{\mkwd{Message}} \\
\mkwd{subscriptions} = \lambda (mX, mY) .
\subscription{\mkwd{onMouseMove}}{(\lambda (x, y) . \mkwd{UpdateCoords}(x, y))}
\el
\]%
}%
The \mkwd{Model} and \mkwd{Message} types incorporate a pair containing
mouse co-ordinates. The \mkwd{view} function takes the current co-ordinates, and
splices them into an HTML document. The \mkwd{update} function takes the old and
new co-ordinates, and updates the model with the new co-ordinates.
The \mkwd{subscriptions} function, evaluated after \mkwd{update}, takes the current model and returns a
subscription
$\subscription{\mkwd{onMouseMove}}{(\lambda (x, y) .
\mkwd{UpdateCoords}(x, y))}
$
which generates an $\mkwd{UpdateCoords}$ message
whenever the mouse moves.

Formalising subscriptions involves adding the subscription terms and types,
extending event loop processes to record the $\mkwd{subscriptions}$ function,
and extending system configurations to record the current subscriptions and
environment events.

\begin{figure}[htpb]
\begin{adjustwidth}{-3cm}{-3cm}
{\footnotesize
~Modified syntax
\[
\begin{array}{lrcl}
\text{Types} & A, B, C & ::= & \cdots \midspace \subscriptionty{A} \\
\text{Terms} & L, M, N & ::= & \cdots \midspace \subscription{\eh}{M} \midspace \subempty \\
\end{array}
\]

~Additional term typing rules \hfill \framebox{$\Gamma \vdash M : A$}
\begin{mathpar}
  \inferrule
  [T-Sub]
  { \Gamma \vdash M : \evtty{\eh} \to A }
  { \Gamma \vdash \subscription{\eh}{M} : \subscriptionty{A} }

  \inferrule
  [T-SubEmpty]
  { }
  { \Gamma \vdash \subempty : \subscriptionty{A} }

  \inferrule
  [T-SubAppend]
  { \Gamma \vdash M : \subscriptionty{A} \\ \Gamma \vdash N : \subscriptionty{A} }
  { \Gamma \vdash \append{M}{N} : \subscriptionty{A} }
\end{mathpar}

~Modified runtime syntax
\[
\begin{array}{lrcl}
\text{Subscription values} & \vs & ::= & \subscription{\eh}{V} \midspace \subempty \\
\text{Values} & V & ::= &  \cdots \midspace \vs \\
\text{Evaluation contexts} & E & ::= & \cdots \midspace \subscription{\eh}{E} \\
\text{Processes} & P & ::= & \run{M} \midspace \handlerprocsub{T}{V_v}{V_u}{V_s}
  \midspace \thread{M} \midspace P_1 \parallel P_2 \\
\text{Configurations} & \config{C} & ::= &  \syssub{P}{\vh}{\vs}{\seq{e}} \\
\end{array}
\]%

~Modified runtime typing rules \hfill
\framebox{$\vphantom{\vdash^\phi} \vdashs T : \evtloopty{A}{B}$}
\framebox{$\vphantom{\vdash^\phi\evtloopty{A}{B}} \vdash^\phi P : A$}
\framebox{$\vphantom{\vdash^\phi\evtloopty{A}{B}} \vdash \config{C}$}

\begin{mathpar}
    \inferrule
    [TS-Processing]
    { \vdash M : (A \times \htmlty{B} \shade{\times \subscriptionty{B}}) }
    { \vdashs M : \evtloopty{A}{B} }

    \inferrule
    [TP-Handler]
    { \vdashs T : \evtloopty{A}{B} \\
      \cdot \vdash V_v : A \to \htmlty{B} \\\\
      \cdot \vdash V_u : (B \times A) \to A \\\\
      \shade{\cdot \vdash V_s : A \to \subscriptionty{B}}
    }
    {  \vdash^\bcirc \handlerprocsub{T}{V_v}{V_u}{\shade{V_s}} : B }

    \inferrule
    [T-Run]
    { \cdot \vdash^\bcirc M : (A \times (A \to \htmlty{B}) \times ((B \times A) \to A) \shade{\times (A \to \subscriptionty{B})})}
    { \vdash^\bcirc \run{M} }

    \inferrule
    [T-System]
    { \vdash P : A \\
      \vdash \vh : \pagety{A} \\
      \shade{\cdot \vdash \vs : \subscriptionty{A}} \\
      (\vdash e_i)_i
    }
    { \vdash \syssub{P}{\vh}{\shade{\vs}}{\seq{e}} }
\end{mathpar}

~Modified meta-level definitions

\begin{minipage}{0.3\textwidth}
\[
  \bl
  \handle{\textit{m}, (\textit{v}, \textit{u}, \shade{\textit{s}}), \textit{msg}}
  \defeq \\
  \quad \letintwo{m'}{u \app (\textit{msg}, m)} \\
  \quad {(m', v \app m', \shade{s \app m'})} \\
  \el
\]%
\end{minipage}
\hfill
\begin{minipage}{0.65\textwidth}
\[
\begin{array}{rcl}
  \handlers{\ev}{\subempty} & = & \epsilon \\
  \handlers{\ev}{\append{V}{W}} & = & \handlers{\ev}{V} \cdot \handlers{\ev}{W} \\
  \handlers{\ev}{\subscription{h}{V}} & = &
  \begin{cases}
    V & \text{if } \handler{\ev} = h \\
    \epsilon & \text{otherwise}
  \end{cases}
  \end{array}
\]
\end{minipage}

~Modified reduction rules on configurations \hfill \framebox{$\config{C} \ceval \config{C}'$}
\[
  \begin{array}{lrcl}
    \textsc{E-Run} & \quad
    \runsub{V_m}{V_v}{V_u}{\shade{V_s}} & \ceval &
      \syssub{\handlerprocsub{(V_m, V_v \app V_m, \shade{V_s \app
      V_m})}{V_v}{V_u}{\shade{V_s}}}{\htmlempty}{\subempty}{\epsilon}
      \vspace{\redrowskip} \\
    \textsc{E-Interact} &
    \syssub{P}{\config{D}[\pgtag{t}{U}{D}{\seq{e}}]}{\shade{\vs}}{\shade{\seq{e'}}} & \ceval &
    \syssub{P}{\config{D}[\pgtag{t}{U}{D}{\seq{e} \cdot
    \evtpayload{\evt{ev}}{V}}]}{\shade{\vs}}{\shade{\seq{e'}}}
    \\
                                                  & & &
                                                  \quad \text{for some }
                                                  \evt{ev}, V \text{ such that } \vdash \evtpayload{\evt{ev}}{V}
                                                  \vspace{\redrowskip}\\
    \textsc{E-InteractS} & \quad
    \syssub{P}{D}{\vs}{\seq{e}} & \ceval &
    \syssub{P}{D}{\vs}{\seq{e} \cdot \evtpayload{\evt{ev}}{V}} \\
                                & & & \quad \text{for some } \evt{ev}, V
                                \text{ such that } \vdash
                                \evtpayload{\evt{ev}}{V} \\
    \textsc{E-EvtS} & \quad
\syssub{P}{\vh}{\vs}{\evtpayload{\evt{ev}}{W} \cdot \seq{e}}
       & \ceval &
       \syssub{P \parallel \thread{V_1 \app W} \parallel \cdots \parallel
       \thread{V_n \app W}}{\vh}{\vs}{\seq{e}} \\
       & & & \quad \text{where } \handlers{\ev}{\vs} = \seq{V} \vspace{\redrowskip} \\
    \textsc{E-Update} & \quad
  \syssub{\procctx[\handlerproc{(V_m, U,
\shade{\vs'})}{F}]}{\vh}{\shade{\vs}}{\shade{\seq{e}}} & \ceval &
\syssub{\procctx[\handlerproc{\idle{V_m}}{F}]}{\vh'}{\shade{\vs'}}{\shade{\seq{e}}}
  \\
                                                        & & & \quad \text{where
                                                        } \diff{U}{D} = D'
  \end{array}
\]
}
\caption{Runtime syntax and semantics for subscriptions}
\label{fig:subs}
\end{adjustwidth}
\end{figure}

Figure~\ref{fig:subs} shows the changes required to support
subscriptions. Modifications to existing rules are shaded.
The $\subscription{h}{M}$ term defines a subscription, which has
type $\subscriptionty{A}$ if $M$ is an event handler function for handler $h$
which takes payload type $\evtty{h}$ and produces messages of type $A$. The
$\subempty$ term denotes an empty subscription, and subscriptions can be composed using
$\star$.
Subscription values are ranged over by $\vs$.
We extend the definition of event handler processes to
$\handlerprocsub{T}{V_v}{V_u}{V_s}$ to incorporate an additional
function $V_s$, which given a model of type $A$, produces a subscription of type
$\subscriptionty{B}$. We further modify configurations $\config{C}$, allowing
the $\calcwd{run}$ configuration to specify an initial subscription function,
and we extend the system configuration $\syssub{P}{\vh}{\vs}{\seq{e}}$ to record
the current subscriptions $\vs$ and environment events $\seq{e}$.

We extend the $\mkwd{handle}$ function to apply $V_s$
to the updated model and to return the updated model, updated HTML,
and updated subscriptions. Rule \textsc{E-InteractS} models environment events
being added to the webpage.
The \textsc{E-EvtS} rule handles environment events:
the $\handlers{h}{V}$ meta-level function returns the subscription event
handlers, which are evaluated in parallel.  We further modify \textsc{E-Update}
to update the current subscription at the end of message processing, and
\textsc{E-Run} to take the initial subscription function into account.

\clearpage
\subsection{Metatheory}
The metatheory of \mvu extends straightforwardly to the subscription extension.

\begin{theorem}[Preservation (Configurations with subscriptions)]
If $\vdash \config{C}$ and $\config{C} \ceval \config{C}'$, then $\vdash \config{C}'$.
\end{theorem}

\begin{proof}
By induction on the derivation of $\config{C} \ceval \config{C}'$.
We show the modified cases:

\begin{proofcase}{E-EvtS}
Assumption:

\begin{mathpar}
\inferrule*
{
  \vdash^\bcirc P : A \\
   \vdash D : \pagety{A} \\
   \vdash \vs : \subscriptionty{A} \\
  \inferrule*
  { \vdash W : \evtty{\evt{ev}} }
  { \vdash \evtpayload{\evt{ev}}{W} } \\
    \vdash \seq{e}
}
{ \vdash \syssub{P}{D}{\vs}{\evtpayload{\evt{ev}}{W} \cdot \seq{e}} }
\end{mathpar}

By the definition of $\handlers{\evt{ev}}{\vs} = \seq{V}$, each $V_i$ has type
$\evtty{h} \to A$, where $h = \handler{\evt{ev}}$. As there is a bijective mapping
between handler names $h$ and event names, we have that $\evtty{h} =
\evtty{\evt{ev}}$. Thus, for each $V_i$, we have that $\vdash^\wcirc \thread{V_i
\app W} : A$

By \textsc{TP-Par}, we have that $\vdash^\bcirc P \parallel\thread{V_1 \app W}
\parallel \cdots \parallel \thread{V_n \app W} : A$. Recomposing:

\begin{mathpar}
  \inferrule*
  {
    \vdash^\bcirc P \parallel \thread{V_1 \app W} \parallel \cdots \parallel \thread{V_n \app W} : A \\
    \vdash D : \pagety{A} \\
    \vdash \vs : \subscriptionty{A} \\
    \vdash \seq{e}
  }
  { \vdash
    \syssub{P \parallel\thread{V_1 \app W} \parallel \cdots \parallel \thread{V_n \app W}}
    {D}{\vs}{\seq{e}}
  }
\end{mathpar}
\end{proofcase}

\begin{proofcase}{E-InteractS}

\begin{mathpar}
\inferrule*
{
  \vdash^\bcirc P : A \\
  \vdash D : \pagety{A} \\
  \vdash \vs : \subscriptionty{A} \\
  \vdash \seq{e}
}
{ \vdash \syssub{P}{D}{\vs}{\seq{e}} }
\end{mathpar}
where $\vdash \evtpayload{\evt{ev}}{W}$.

Recomposing:

\begin{mathpar}
\inferrule*
{
  \vdash^\bcirc P : A \\
  \vdash D : \pagety{A} \\
  \vdash \vs : \subscriptionty{A} \\
    \inferrule*
    { \vdash W : \evtty{\evt{ev}} }
    { \vdash \evtpayload{\evt{ev}}{W} } \\
    \vdash \seq{e}
}
{ \vdash \syssub{P}{D}{\vs}{\evtpayload{\evt{ev}}{W} \cdot \seq{e}} }
\end{mathpar}
\end{proofcase}

\begin{proofcase}{E-Update}
We consider the case where $\procctx[\handlerprocsub{(V_m, U,
\vs)}{V_v}{V_u}{V_s}] =
\handlerprocsub{(V_m, U, \vs)}{V_v}{V_u}{V_s} \parallel P$.

Let $\deriv{D}$ be the derivation:
\begin{mathpar}
\inferrule*
{
  \inferrule*
  { \cdot \vdash V_m : A \\ \cdot \vdash U : \htmlty{B} \\ \cdot \vdash \vs' : \subscriptionty{B} }
  { \cdot \vdash (V_m, U, \vs') : (A \times \htmlty{B} \times \subscriptionty{B}) }
}
{ \vdash (V_m, U, \vs') : \evtloopty{A}{B} }
\end{mathpar}

\begin{mathpar}
\inferrule*
{
  \inferrule*
  {
    \inferrule*
    {
      \deriv{D} \\
      {\bl
        \cdot \vdash V_v : A \to \htmlty{B} \\
        \cdot \vdash V_u : (B \times A) \to A \\
        \cdot \vdash V_s : A \to \subscriptionty{B}
       \el} \\
    }
    { \vdash^\bcirc \handlerprocsub{(V_m, U, \vs')}{V_v}{V_u}{V_s} } \\
      \vdash^\wcirc P : A
  }
  { \vdash^\bcirc \handlerprocsub{(V_m, U, \vs')}{V_v}{V_u}{V_s} \parallel P} \\
  \vdash D : \pagety{A} \\
  \cdot \vdash \vs : \subscriptionty{A} \\
  \vdash \seq{e}
}
{ \vdash \syssub{\handlerprocsub{(V_m, U, \vs')}{V_v}{V_u}{V_s} \parallel P}{D}{\vs}{\seq{e}} }
\end{mathpar}

Let $\deriv{D'}$ be the derivation:
\begin{mathpar}
  \inferrule*
  { \cdot \vdash V_m : A }
  { \vdash \idle{V_m} : \evtloopty{A}{B} }
\end{mathpar}

By Lemma~\ref{lem:diffing}, $\vdash D' : \pagety{A}$.

Recomposing:

\begin{mathpar}
\inferrule*
{
  \inferrule*
  {
    \inferrule*
    {
      \deriv{D}' \\
      {\bl
        \cdot \vdash V_v : A \to \htmlty{B} \\
        \cdot \vdash V_u : (B \times A) \to A \\
        \cdot \vdash V_s : A \to \subscriptionty{B}
       \el}
    }
    { \vdash^\bcirc \handlerprocsub{\idle{V_m}}{V_v}{V_u}{V_s} } \\
    \vdash^\wcirc P : A
  }
  { \vdash^\bcirc \handlerprocsub{\idle{V_m}}{V_v}{V_u}{V_s} \parallel P} \\
  {\bl
  \vdash D' : \pagety{A} \\
  \cdot \vdash \vs' : \subscriptionty{A} \\
  \vdash \seq{e}\el}
}
{ \vdash \syssub{\handlerprocsub{\idle{V_m}}{V_v}{V_u}{V_s}
\parallel P}{D'}{\vs'}{\seq{e}} }
\end{mathpar}

\end{proofcase}

\end{proof}

For event progress, we define the relation $\cevalminus$ to be $\ceval$ without
rules \textsc{E-Interact} or \textsc{E-InteractS}.

\begin{theorem}[Event Progress]
  If $\cdot \vdash \config{C}$, either:
  \begin{enumerate}
    \item there exists some $\config{C}'$ such that $\config{C} \cevalminus
      \config{C'}$; or
    \item $\config{C} =
    \syssub{\handlerprocsub{\idle{V_m}}{V_v}{V_u}{V_s}}{\vh}{\vs}{\epsilon}$ where $\vh$
      cannot be written $\config{D}[\pgtag{t}{V}{W}{\seq{e}}]$ for some
      non-empty $\seq{e}$.
  \end{enumerate}
\end{theorem}
\begin{proof}
Similar to the proof of Theorem~\ref{thm:event-progress}. The key difference is
that in case (2), the environment event queue must be $\epsilon$ as otherwise the
configuration could reduce by \textsc{E-EvtS}.
\end{proof}

 \clearpage
\section{Proofs for Section~\ref{sec:extensions}}\label{appendix:combined}

\subsection{Preservation.}
We firstly need some auxiliary results, allowing us to manipulate the various
contexts:

\begin{lemma}[Subprocess typeability]\label{lem:combined:subprocess-typeability}
  Let $\deriv{D}$ be a derivation of $\Psi \vdashtrans{\phi}{\vers}
  \procctx[P] : A$.
  There exist $\Psi', \phi'$ and some subderivation $\deriv{D}'$ of $\deriv{D}$
  concluding $\Psi' \vdashtrans{\phi'}{\vers} P : A$, where the position of
  $\deriv{D}'$ in $\deriv{D}$ corresponds to that of the hole in $\procctx$.
\end{lemma}
\begin{proof}
  By induction on the structure of $\procctx$.
\end{proof}

\begin{lemma}[Subprocess replacement (main threads)]\label{lem:combined:subprocess-replacement}
  If:
  \begin{itemize}
    \item $\deriv{D}$ is a derivation of $\Psi \vdashtrans{\bcirc}{\vers} \procctx[P] : A$
    \item $\deriv{D}'$ is a subderivation of $\deriv{D}$ concluding $\Psi' \vdashtrans{\bcirc}{\vers} P : A$
    \item The position of $\deriv{D}'$ in $\deriv{D}$ corresponds to that of the
      hole in $\procctx$
    \item $\Psi' \vdashtrans{\bcirc}{\vers'} P' : B$
  \end{itemize}
  then $\Psi \vdashtrans{\bcirc}{\vers'} \procctx[P'] : B$.
\end{lemma}
\begin{proof}
  By induction on the structure of $\procctx$.
\end{proof}

\begin{lemma}[Subterm typeability]\label{lem:combined:subterm-typeability}
  Let $\deriv{D}$ be a derivation of $\Gamma_1 + \Gamma_2 \vdash E[M] : A$.
  There exists some subderivation $\deriv{D}'$ of $\deriv{D}$
  concluding $\Gamma_2 \vdash M : A$, where the position of
  $\deriv{D}'$ in $\deriv{D}$ corresponds to that of the hole in $E$.
\end{lemma}
\begin{proof}
  By induction on the structure of $\procctx$.
\end{proof}

\begin{lemma}[Subterm replacement]\label{lem:combined:subterm-replacement}
  If:
  \begin{itemize}
    \item $\deriv{D}$ is a derivation of $\Gamma_1 + \Gamma_2 \vdash E[M] : A$
    \item $\deriv{D}'$ is a subderivation of $\deriv{D}$ concluding $\Gamma_2
      \vdash M : B$
    \item The position of $\deriv{D}'$ in $\deriv{D}$ corresponds to that of the
      hole in $E$
    \item $\Gamma_3 \vdash N : B$
    \item $\Gamma_1 + \Gamma_3$ is defined
  \end{itemize}
  then $\Gamma_1 + \Gamma_3 \vdash E[N] : A$.
\end{lemma}
\begin{proof}
  By induction on the structure of $E$.
\end{proof}

\begin{lemma}[Thread typeability]\label{lem:combined:thread-typeability}
  Let $\deriv{D}$ be a derivation of $\Psi \vdash^\phi \config{T}[M] : A$.
  There exist some $\Psi_1, \Psi_2$ such that $\Psi = \Psi_1, \Psi_2$ and
  a subderivation $\deriv{D}'$ of $\deriv{D}$ concluding $\Psi_2 \vdash M : A$,
  where the position of $\deriv{D}'$ in $\deriv{D}$ corresponds to that of the
  hole in $\config{T}$.
\end{lemma}
\begin{proof}
  By case analysis on $\config{T}$ followed by
  Lemma~\ref{lem:combined:subterm-typeability}.
\end{proof}

\begin{lemma}[Thread replacement]\label{lem:combined:thread-replacement}
  If:
  \begin{itemize}
    \item $\deriv{D}$ is a derivation of $\Psi_1, \Psi_2 \vdash^\phi \config{T}[M] : A$
    \item $\deriv{D}'$ is a subderivation of $\deriv{D}$ concluding $\Psi_2 \vdash M : B$
    \item The position of $\deriv{D}'$ in $\deriv{D}$ corresponds to that of the
      hole in $\config{T}$
    \item $\Psi_3 \vdash N : B$
    \item $\Psi_1, \Psi_3$ is defined
  \end{itemize}
  then $\Psi_1, \Psi_3 \vdash^\phi \config{T}[N] : A$.
\end{lemma}
\begin{proof}
  By case analysis on $\config{T}$ followed by
  Lemma~\ref{lem:combined:subterm-replacement}.
\end{proof}

It is also useful to show that if a term can be given a type of unrestricted
kind, then this implies that the environment is unrestricted.

\begin{fake}{Lemma~\ref{lem:combined:env-kinding} (Environment kinding)}
  If $\Gamma \vdash M : A$ and $A :: \kappa$, then $\Gamma :: \kappa$.
\end{fake}
\begin{proof}
  By induction on the derivation of $\Gamma \vdash M : A$.
\end{proof}

We often make use of Lemma~\ref{lem:combined:env-kinding} implicitly, and allow
ourselves to implicitly duplicate unrestricted environments.

Unrestricted environments admit a standard weakening property, which we
also use implicitly:
\begin{lemma}[Weakening]\label{lem:combined:weakening}
  If:
  \begin{enumerate}
    \item $\Gamma_1 \vdash M : A$
    \item $\Gamma_2 :: \unr$
    \item $\Gamma_1 + \Gamma_2$ is defined
  \end{enumerate}
  then $\Gamma_1 + \Gamma_2 \vdash M : A$.
\end{lemma}
\begin{proof}
  By induction on the derivation of $\Gamma_1 \vdash M : A$.
\end{proof}

If an term is typeable under an environment only containing runtime names, and
the resulting type is of unrestricted kind, then the environment must be empty.

\begin{lemma}\label{lem:combined:unr-psi-empty}
  If $\Psi \vdash M : A$ and $A :: \unr$, then $\Psi = \cdot$.
\end{lemma}
\begin{proof}
  $\Psi$ is either empty or contains only runtime names with session
  types. Since session types have kind $\lin$, by Lemma~\ref{lem:combined:env-kinding}
  it must be the case that $\Psi$ is empty.
\end{proof}

We can extend Lemma~\ref{lem:combined:unr-psi-empty} to state typing:
\begin{corollary}\label{cor:combined:unr-psi-empty-st}
  If $\Psi \vdash F : \statetytrans{A}{B}{C}$, then $\Psi = \cdot$.
\end{corollary}

\begin{lemma}[Diffing]\label{lem:combined:diffing}
  If:
  \begin{itemize}
    \item $\cdot \vdash U : \htmlty{A}$
    \item $\vdash D : \pagety{B}$
    \item $\diff{U}{D} = D'$
  \end{itemize}

  then $\vdash D' : \pagety{A}$
\end{lemma}
\begin{proof}
  Similar to the proof of Lemma~\ref{lem:diffing}.
\end{proof}

We can now state the required preservation results:

\begin{lemma}[Preservation (Equivalence)]\label{lem:combined:equiv-pres}
  If $\Psi \vdashtrans{\wcirc}{\vers} P : A$ and $P \equiv P'$, then
  $\Psi \vdashtrans{\wcirc}{\vers} P' : A$.
\end{lemma}

\begin{lemma}[Preservation (Term reduction)]
  If $\Gamma \vdash M : A$ and $M \teval N$, then $\Gamma \vdash N : A$.
\end{lemma}
\begin{proof}
  An entirely standard induction on the derivation of $M \teval N$.
\end{proof}

\begin{lemma}[Preservation (Process reduction)]
  If $\Psi \vdashtrans{\phi}{\vers} P : A$ and $P \ceval P'$, then
  $\Psi \vdashtrans{\phi}{\vers} P' : A$.
\end{lemma}

\begin{proof}
  By induction on the derivation of $P \ceval P'$. We show some illustrative
  cases:

  \begin{proofcase}{EP-Discard}
    \[
      \handlerproctrans{T}{F}{\vers} \parallel \threadtrans{V}{\vers'} \ceval
      \handlerproctrans{T}{F}{\vers} \parallel \zap{V} \qquad \text{where
    } \vers \ne \vers'
 \]

    Assumption:
    \begin{mathpar}
      \inferrule*
      {
        \inferrule*
        {
          \Psi_1 \vdashs T : \evtlooptytrans{A}{B}{C} \\
          \Psi_2 \vdash F : \statetytrans{A}{B}{C}
        }
        { \Psi_1, \Psi_2 \vdashtrans{\bcirc}{\vers} \handlerproctrans{T}{F}{\vers} : B }
        \\
        \inferrule*
        { \Psi_3 \vdash V : B' }
        { \Psi_3 \vdashtrans{\wcirc}{\vers} \threadtrans{V}{\vers'} : B  }
      }
      { \Psi_1, \Psi_2, \Psi_3 \vdashtrans{\bcirc}{\vers}
        \handlerproctrans{T}{F}{\vers} \parallel \threadtrans{V}{\vers'} : B
      }
    \end{mathpar}

    By \textsc{T-Name}, $\fn{V} = c_1, \ldots, c_n$ and $\Psi_3 = c_1: S_1, \ldots, c_n : S_n$. By
    \textsc{TP-Zap} and \textsc{TP-Par}, we have that $\Psi_3
    \vdashtrans{\wcirc}{\vers} \zap{c_1}
    \parallel \cdots \parallel \zap{c_n} : B$, which we write as $\Psi_3 \vdash \zap{V} : B $.

    Thus, recomposing:
    \begin{mathpar}
      \inferrule*
      {
        \inferrule*
        {
          \Psi_1 \vdashs T : \evtlooptytrans{A}{B}{C} \\
          \Psi_2 \vdash F : \statetytrans{A}{B}{C}
        }
        { \Psi_1, \Psi_2 \vdashtrans{\bcirc}{\vers} \handlerproctrans{T}{F}{\vers} : B
        }
        \\
        \Psi_3 \vdashtrans{\wcirc}{\vers} \zap{V} : B
      }
      { \Psi_1, \Psi_2, \Psi_3 \vdashtrans{\bcirc}{\vers}
        \handlerproctrans{T}{F}{\vers} \parallel \zap{V} : B
      }
    \end{mathpar}
    as required.
  \end{proofcase}

  \begin{proofcase}{EP-Handle}
    \[
    \handlerprocexptc{\idle{V_m}}{V_v}{V_u}{V_e}{\vers} \parallel V \ceval
    \handlerprocexptc{\updating{V_u \app (V, V_m)}}{V_v}{V_u}{V_e}{\vers}
    \]

    Assumption:
    \begin{mathpar}
      \inferrule*
      {
        \inferrule*
          {
            \inferrule*
            { \Psi_1 \vdash V_m : A }
            { \Psi_1 \vdashs \idle{V_m} : \evtlooptytrans{A}{B}{C} } \\
            \inferrule*
              {
                {\bl
                  \Psi_2 \vdash V_v :  C \uto \htmlty{B} \\
                  \Psi_3 \vdash V_u : (B \times A) \uto \transitionty{A}{B} \\
                  \Psi_4 \vdash V_e : A \uto (A \times C)
                \el}
              }
              {
                \Psi_2, \Psi_3, \Psi_4 \vdash (V_v, V_u, V_e) :
                \statetytrans{A}{B}{C}
              }
          }
          { \Psi_1, \ldots, \Psi_4 \vdashtrans{\bcirc}{\vers}
          \handlerprocexptc{\idle{V_m}}{V_v}{V_u}{V_e}{\vers} : B}
          \\
          \inferrule*
          { \Psi_5 \vdash V : B }
          { \Psi_5 \vdashtrans{\wcirc}{\vers} \thread{V} : B }
      }
      {
        \Psi_1, \ldots, \Psi_5 \vdashtrans{\bcirc}{\vers}
        \handlerprocexptc{\idle{V_m}}{V_v}{V_u}{V_e}{\vers} \parallel \thread{V} : B
      }
    \end{mathpar}

    By Lemma~\ref{lem:combined:unr-psi-empty}, we have that $\Psi_2, \Psi_3,
    \Psi_4 = \cdot$. Simplifying:

    \begin{mathpar}
      \inferrule*
      {
        \inferrule*
          {
            \inferrule*
            { \Psi'_1 \vdash V_m : A }
            { \Psi'_1 \vdashs \idle{V_m} : \evtlooptytrans{A}{B}{C} } \\
            \inferrule*
              {
                {\bl
                  \cdot \vdash V_v :  C \uto \htmlty{B} \\
                  \cdot \vdash V_u : (B \times A) \uto \transitionty{A}{B} \\
                  \cdot \vdash V_e : A \uto (A \times C)
                \el}
              }
              {
                \cdot \vdash (V_v, V_u, V_e) :
                \statetytrans{A}{B}{C}
              }
          }
          { \Psi'_1 \vdashtrans{\bcirc}{\vers}
          \handlerprocexptc{\idle{V_m}}{V_v}{V_u}{V_e}{\vers} : B}
          \\
          \inferrule*
          { \Psi'_2 \vdash V : B }
          { \Psi'_2 \vdashtrans{\wcirc}{\vers} \thread{V} : B }
      }
      {
        \Psi'_1, \Psi'_2 \vdashtrans{\bcirc}{\vers}
        \handlerprocexptc{\idle{V_m}}{V_v}{V_u}{V_e}{\vers} \parallel \thread{V} : B
      }
    \end{mathpar}

    We can show:
    \begin{mathpar}
      \inferrule*
      {
        \inferrule*
        { \cdot \vdash V_u : (B \times A) \uto \transitionty{A}{B} \\
          \inferrule*
          { \Psi'_2 \vdash V : B \\ \Psi'_1 : V_m : A }
          { \Psi'_1, \Psi'_2 \vdash (V, V_m) : (B \times A) }
        }
        { \Psi'_1, \Psi'_2 \vdash V_u \app (V, V_m) : \transitionty{A}{B} } \\
      }
      { \Psi'_1, \Psi'_2 \vdashs \updating{V_u \app (V, V_m)} : \evtlooptytrans{A}{B}{C} }
    \end{mathpar}
    Let us call this derivation $\deriv{D}$.

    Recomposing:
    \begin{mathpar}
      \inferrule*
        {
          \deriv{D} \\
          \cdot \vdash (V_v, V_u, V_e) : \statetytrans{A}{B}{C}
        }
        { \Psi'_1, \Psi'_2 \vdashtrans{\bcirc}{\vers}
        \handlerprocexptc{\updating{V_u \app (V, V_m)}}{V_v}{V_u}{V_e}{\vers} : B}
    \end{mathpar}
      where $\Psi'_1, \Psi'_2 = \Psi$, as required.
  \end{proofcase}

  \begin{proofcase}{EP-Extract}
    \[
      \handlerprocexptc{\updating{(\notransition{V_m}{V_c})}}{V_v}{V_u}{V_e}{\vers} \ceval
      \handlerprocexptc{\extracting{V_c}{(V_e \app V_m)}}{V_v}{V_u}{V_e}{\vers}
    \]

    Assumption:
    \begin{mathpar}
      \inferrule*
      {
        \inferrule*
        { \inferrule*
          { \Psi_1 \vdash V_m : A \\ \Psi_2 \vdash V_c : \cmdty{B} }
          { \Psi_1, \Psi_2 \vdash \notransition{V_m}{V_c} : \transitionty{A}{B}}
        }
        { \Psi_1, \Psi_2 \vdash \updating{(\notransition{V_m}{V_c})} :
        \evtlooptytrans{A}{B}{C} } \\
        \inferrule*
        {
          {\bl
            \Psi_3 \vdash V_v :  C \uto \htmlty{B} \\
            \Psi_4 \vdash V_u : (B \times A) \uto \transitionty{A}{B} \\
            \Psi_5 \vdash V_e : A \uto (A \times C) \\
            C :: \unr
          \el}
        }
        {
          \Psi_3, \Psi_4, \Psi_5 \vdash \statetytrans{A}{B}{C}
        }
      }
      {
        \Psi_1, \ldots, \Psi_5 \vdashtrans{\bcirc}{\vers}
        \handlerprocexptc{\updating{(\notransition{V_m}{V_c})}}{V_v}{V_u}{V_e}{\vers} : B
      }
    \end{mathpar}

    By Lemma~\ref{lem:combined:unr-psi-empty}, we have that $\Psi_3, \Psi_4,
    \Psi_5 = \cdot$. Simplifying:

    \begin{mathpar}
      \inferrule*
      {
        \inferrule*
        { \inferrule*
          { \Psi'_1 \vdash V_m : A \\ \Psi'_2 \vdash V_c : \cmdty{B} }
          { \Psi'_1, \Psi'_2 \vdash \notransition{V_m}{V_c} : \transitionty{A}{B}}
        }
        { \Psi'_1, \Psi'_2 \vdash \updating{(\notransition{V_m}{V_c})} :
        \evtlooptytrans{A}{B}{C} } \\
        \inferrule*
        {
          {\bl
            \cdot \vdash V_v :  C \uto \htmlty{B} \\
            \cdot \vdash V_u : (B \times A) \uto \transitionty{A}{B} \\
            \cdot \vdash V_e : A \uto (A \times C)
          \el}
        }
        {
          \cdot \vdash (V_v, V_u, V_e) : \statetytrans{A}{B}{C}
        }
      }
      {
        \Psi'_1, \Psi'_2 \vdashtrans{\bcirc}{\vers}
        \handlerprocexptc{\updating{(\notransition{V_m}{V_c})}}{V_v}{V_u}{V_e}{\vers} : B
      }
    \end{mathpar}

   We can show:
   {\small
    \begin{mathpar}
      \inferrule*
      {
        \inferrule*
        {
          \Psi'_1 \vdash V_c : \cmdty{B} \\
          \inferrule*
          { \cdot \vdash V_e : A \uto (A \times C) \\  \Psi'_1 \vdash V_m : A }
          { \Psi'_1, \Psi'_2 \vdash V_e \app V_m : (A \times C) }
        }
        { \Psi'_1, \Psi'_2 \vdashs \extracting{V_c}{(V_e \app V_m)} : \evtlooptytrans{A}{B}{C} } \\
        \inferrule*
        {
          {\bl
            \cdot \vdash V_v :  C \uto \htmlty{B} \\
            \cdot \vdash V_u : (B \times A) \uto \transitionty{A}{B} \\
            \cdot \vdash V_e : A \uto (A \times C)
          \el}
        }
        {
          \cdot \vdash (V_v, V_u, V_e) : \statetytrans{A}{B}{C}
        }
      }
      {
        \Psi'_1, \Psi'_2 \vdashtrans{\bcirc}{\vers}
        \handlerprocexptc{\updating{(\notransition{V_m}{V_c})}}{V_v}{V_u}{V_e}{\vers} : B
      }
    \end{mathpar}
  }

    where $\Psi'_1, \Psi'_2 = \Psi$, as required.
  \end{proofcase}

  \begin{proofcase}{EP-RenderT}
    \[
      \handlerproctrans{\extractingt{F'}{V_c}{(V_m \app V_{um})}}{F}{\vers} \ceval
      \handlerproctrans{\transitioningext{V_m}{F'}{V_c}{(V_v \app V_{um})}}{F}{\vers}
    \]
    \[
    \text{where } F' =
\statecomb{V'_v}{V'_u}{V'_e}
    \]

    Let $\deriv{D}$ be the derivation:
    \begin{mathpar}
      \inferrule*
        {
          \inferrule*
          {
            {\bl
            \Psi_1 \vdash V'_v :  C' \uto \htmlty{B'} \\
            \Psi_2 \vdash V'_u : (B' \times A') \uto \transitionty{A'}{B'} \\
            \Psi_3 \vdash V'_e : A' \uto (A' \times C') \\
            C' :: \unr
            \el}
          }
          {
            (V'_v, V'_u, V'_e) : \statetytrans{A'}{B'}{C'}
          } \\
          \Psi_4 \vdash V_c : \cmdty{B'} \\
          \inferrule*
          { \Psi_5 \vdash V_m : A' \\ \Psi_6 \vdash V_{um} : C' }
          { \Psi_5, \Psi_6 \vdash (V_m, V_{um}) : (A' \times C') }
        }
        { \Psi_1, \ldots, \Psi_6 \vdashs \extractingtexp{V'_v}{V'_u}{V'_e}{V_c}{(V_m, V_{um})} : \evtlooptytrans{A}{B}{C} }
    \end{mathpar}

    Assumption:
    \begin{mathpar}
      \inferrule*
      {
        \deriv{D} \\
        \Psi_7 \vdash (V_v, V_u, V_e) : \statetytrans{A}{B}{C}
      }
      { \Psi_1, \ldots, \Psi_{7} \vdashtrans{\bcirc}{\vers}
      \handlerprocexptc{\extractingtexp{V'_v}{V'_u}{V'_e}{V_c}{(V_m \app V_{um})}}{V_v}{V_u}{V_e}{\vers} : B }
    \end{mathpar}

    By Lemma~\ref{lem:combined:unr-psi-empty}, we have that $\Psi_1, \Psi_2,
    \Psi_3, \Psi_6, \Psi_7 = \cdot$. We can therefore
    substantially simplify:

    Let $\deriv{D}'$ be the derivation:
    \begin{mathpar}
      \inferrule*
        {
          \inferrule*
          {
            {\bl
            \cdot \vdash V'_v :  C' \uto \htmlty{B'} \\
            \cdot \vdash V'_u : (B' \times A') \uto \transitionty{A'}{B'} \\
            \cdot \vdash V'_e : A' \uto (A' \times C') \\
            C' :: \unr
            \el}
          }
          {
            (V'_v, V'_u, V'_e) : \statetytrans{A'}{B'}{C'}
          } \\
          \Psi'_1 \vdash V_c : \cmdty{B'} \\
          \inferrule*
          { \Psi'_2 \vdash V_m : A' \\ \cdot \vdash V_{um} : C' }
          { \Psi'_2 \vdash (V_m, V_{um}) : (A' \times C') }
        }
        { \Psi'_1, \Psi'_2 \vdashs \extractingtexp{V'_v}{V'_u}{V'_e}{V_c}{(V_m, V_{um})} : \evtlooptytrans{A}{B}{C} }
    \end{mathpar}

    Assumption:
    \begin{mathpar}
      \inferrule*
      {
        \deriv{D}' \\
        \cdot \vdash (V_v, V_u, V_e) : \statetytrans{A}{B}{C}
      }
      { \Psi'_1, \Psi'_2 \vdashtrans{\bcirc}{\vers}
      \handlerprocexptc{\extractingtexp{V'_v}{V'_u}{V'_e}{V_c}{(V_m \app V_{um})}}{V_v}{V_u}{V_e}{\vers} : B }
    \end{mathpar}

    Let $\deriv{D}''$ be the derivation:
    \begin{mathpar}
      \inferrule*
        {
          \Psi'_2 \vdash V_m : A' \\
          \inferrule*
            {
              {\bl
              \cdot \vdash V'_v :  C' \uto \htmlty{B'} \\
              \cdot \vdash V'_u : (B' \times A') \uto \transitionty{A'}{B'} \\
              \cdot \vdash V'_e : A' \uto (A' \times C') \\
              C' :: \unr
              \el}
            }
            {
              \cdot \vdash (V'_v, V'_u, V'_e) : \statetytrans{A'}{B'}{C'}
            }\\
          \Psi'_1 \vdash V_c : \cmdty{B'} \\
          \inferrule*
          { \cdot \vdash V'_v : C' \uto \htmlty{B'} \\\\ \cdot \vdash V_{um} : C' }
          { \cdot \vdash V'_v \app V_{um} : \htmlty{B'} }
        }
        { \Psi'_1, \Psi'_2 \vdashs
          \transitioningextexp{V_m}{V'_v}{V'_u}{V'_e}{V_c}{V'_v \app V_{um}} : \evtlooptytrans{A}{B}{C} }
    \end{mathpar}

    Recomposing, we have:
    \begin{mathpar}
      \inferrule*
      {
        \deriv{D}'' \\
        {
          \inferrule*
          {
            {\bl
            \cdot \vdash V_v :  C \uto \htmlty{B} \\
            \cdot \vdash V_u : (B \times A) \uto \transitionty{A}{B} \\
            \cdot \vdash V_e : A \uto (A \times C) \\
            C :: \unr
            \el
            }
          }
          {
            \cdot \vdash (V_v, V_u, V_e) : \statetytrans{A}{B}{C}
          }
        } \\
      }
      { \Psi'_1, \Psi'_2 \vdashtrans{\bcirc}{\vers}
        \handlerprocexptc{\transitioningextexp{V_m}{V'_v}{V'_u}{V'_e}{V_c}{(V'_v \app V_{um})}}{V_v}{V_u}{V_e}{\vers} : B }
    \end{mathpar}
    as required.
  \end{proofcase}

  \begin{proofcase}{EP-Comm}
    \[
      (\nu c d)(\config{T}[\gvsend{V}{c}] \parallel \config{T}'[\gvrecv{d}])
        \ceval
      (\nu c d)(\config{T}[c] \parallel \config{T}'[(V, d)])
    \]

    Assumption:
    \begin{mathpar}
      \inferrule*
      {
        \inferrule*
        {
          \Psi_1, \Psi_2, c : S \vdashtrans{\phi_1}{\vers} \config{T}[\gvsend{V}{c}] : B \\
          \Psi_3, d : \gvdual{S} \vdashtrans{\phi_2}{\vers} \config{T}'[\gvrecv{d}] : B
        }
        { \Psi_1, \Psi_2, \Psi_3, c : S, d : \gvdual{S} \vdashtrans{\phi_1 + \phi_2}{\vers}
            \config{T}[\gvsend{V}{c}] \parallel \config{T}'[\gvrecv{d}] : B
        }
      }
      { \Psi_1, \Psi_2, \Psi_3 \vdashtrans{\phi_1 + \phi_2}{\vers} (\nu
      c d)(\config{T}[\gvsend{V}{c}] \parallel \config{T}'[\gvrecv{d}]) : B }
    \end{mathpar}

    We consider the case where $\vers$ is the same in both threads, but the case
    where $\vers$ differs is similar.

    By Lemma~\ref{lem:combined:subterm-typeability}, we have that $S =
    \gvout{A}{S'}$ and:
    \begin{mathpar}
      \inferrule*
      {
        \Sigma(\calcwd{send}) = (A \times \gvout{A}{S'}) \uto S' \\
        \inferrule*
        { \Psi_2 \vdash V : A \\ c : \gvout{A}{S} \vdash c : \gvout{A}{S'}}
        { \Psi_2, c : \gvout{A}{S} \vdash (V, c) : (A \times \gvout{A}{S'}) }
      }
      { \Psi_2, c : \gvout{A}{S} \vdash \gvsend{V}{c} : S' }
    \end{mathpar}
    and $d : \gvin{A}{\gvdual{S'}} \vdash \gvrecv{d} : (A \times S')$.

    By Lemma~\ref{lem:combined:subterm-replacement}, we have that $\Psi_1, c :
    S' \vdashtrans{\phi_1}{\vers} \config{T}[c] : B$ and $\Psi_2, \Psi_3, d : \gvdual{S'}
    \vdashtrans{\phi_2}{\vers} \config{T}'[(V, d)] : B$.

    Thus, recomposing:
  \begin{mathpar}
    \inferrule*
    {
      \inferrule*
      {
        \Psi_1, c : S' \vdashtrans{\phi_1}{\vers} \config{T}[c] : B \\
        \Psi_2, \Psi_3, d : \gvdual{S'} \vdashtrans{\phi_2}{\vers}
        \config{T}'[(V, d)] : B
      }
      { \Psi_1, \Psi_2, \Psi_3, c : S, d : \gvdual{S} \vdashtrans{\phi_1 + \phi_2}{\vers}
        \config{T}[c] \parallel \config{T}'[(V, d)] : B
      }
    }
    { \Psi_1, \Psi_2, \Psi_3 \vdashtrans{\phi_1 + \phi_2}{\vers} (\nu
    c d)(\config{T}[c] \parallel \config{T}'[(V, d)]) : B }
  \end{mathpar}
  \end{proofcase}

  \begin{proofcase}{EP-SendZap}
    \[
      (\nu c d)(\config{T}[\gvsend{V}{c}] \parallel \zap{d}) \ceval
      (\nu c d)(\config{T}[\raiseexn] \parallel \zap{c} \parallel \zap{V} \parallel \zap{d})
    \]

    \begin{mathpar}
      \inferrule*
      {
        \inferrule*
        {
          \Psi_1, \Psi_2, c : S \vdashtrans{\phi}{\vers} \config{T}[\gvsend{V}{c}] : B \\
          d : \gvdual{S} \vdashtrans{\wcirc}{\vers} \zap{d} : B
        }
        { \Psi_1, \Psi_2, c : S, d : \gvdual{S} \vdashtrans{\phi}{\vers}
          \config{T}[\gvsend{V}{c}] \parallel \zap{d} : B }
      }
      { \Psi_1, \Psi_2 \vdashtrans{\phi}{\vers} (\nu c d)(\config{T}[\gvsend{V}{c}] \parallel
      \zap{d} )}
    \end{mathpar}

    By Lemma~\ref{lem:combined:thread-typeability}, $S = \gvout{A}{S'}$, and:
    \begin{mathpar}
      \inferrule*
      {
        \Sigma(\calcwd{send}) = (A \times \gvout{A}{S'}) \uto S' \\
        \inferrule*
        { \Psi_2 \vdash V : A \\ c : \gvout{A}{S'} \vdash c : \gvout{A}{S'}}
        { \Psi_2, c : \gvout{A}{S'} \vdash (V, c) : (A \times \gvout{A}{S'})}
      }
      { \Psi_2, c : \gvout{A}{S'} \vdash \gvsend{V}{c} : S' }
    \end{mathpar}
    By Lemma~\ref{lem:combined:thread-replacement}, $\Psi_1
    \vdashtrans{\phi}{\vers} \config{T}[\raiseexn] : B$.

  By the definition of $\Psi$, $\fn{V} = c_1, \ldots, c_n$ and $\Psi_2 = c_1 :
  S_1, \ldots, c_n : S_n$. We can therefore show that $\Psi_2
  \vdashtrans{\wcirc}{\vers} \zap{c_1} \parallel \cdots \parallel \zap{c_n} :
  B$, which we write as $\Psi_2 \vdash \zap{V} : B$.

  Recomposing:
  \begin{mathpar}
    \inferrule*
    {
      \inferrule*
      {
        \Psi_1 \vdashtrans{\phi}{\vers} \config{T}[\raiseexn] : B \\
        \inferrule*
          {
            c : S' \vdashtrans{\wcirc}{\vers} \zap{c} : B \\
            \Psi_2 \vdashtrans{\wcirc}{\vers} \zap{V} : B
          }
          { \Psi_2, c : S', d : \gvdual{S'} \vdashtrans{\wcirc}{\vers} \zap{c}
          \parallel \zap{V} : B }
      }
      {
        \Psi_1, \Psi_2, c : S', d : \gvdual{S'} \vdashtrans{\phi}{\vers}
          \config{T}[\raiseexn] \parallel \zap{c} \parallel \zap{V} \parallel \zap{d} : B
      }
    }
    { \Psi_1, \Psi_2 \vdashtrans{\phi}{\vers} (\nu c d)(\config{T}[\raiseexn] \parallel \zap{c} \parallel \zap{V} \parallel
    \zap{d} : B }
  \end{mathpar}

  as required.
  \end{proofcase}

  \begin{proofcase}{EP-RaiseH}
    \[
      \config{T}[\tryasinotherwise{\ep[\raiseexn]}{x}{M}{N}] \ceval
      \config{T}[N] \parallel \zap{\ep} : B
    \]
    Assumption: $\Psi_1, \Psi_2, \Psi_3 \vdashtrans{\phi}{\vers}
    \config{T}[\tryasinotherwise{\ep[\raiseexn]}{x}{M}{N}] : B$.

    By Lemma~\ref{lem:combined:thread-typeability}:
    \begin{mathpar}
      \inferrule*
      { \Psi_2 \vdash \ep[\raiseexn] : A' \\
        \Psi_3, x : A' \vdash M : A \\
        \Psi_3 \vdash N : A
      }
      { \Psi_2, \Psi_3 \vdash
      \tryasinotherwise{\ep[\raiseexn]}{x}{M}{N} : A }
    \end{mathpar}

    By Lemma~\ref{lem:combined:subterm-typeability}, $\Psi_2 \vdash \ep[-] :
    A'$. As $\Psi$ contains only runtime names, we know
    that $\fn{\ep} = c_1, \ldots c_n$ and $\Psi_2 = c_1 : S_1, \ldots, c_n :
    S_n$. Thus by \textsc{TP-Zap} and \textsc{TP-Par}, $\Psi_2
    \vdashtrans{\wcirc}{\vers} \zap{c_1} \parallel \ldots \parallel \zap{c_n} :
    B$, which we write as $\Psi_4 \vdashtrans{\wcirc}{\vers} \zap{\ep}$.

    By Lemma~\ref{lem:combined:thread-replacement}, $\Psi_3 \vdash
    \config{T}[N] : B$.

    Recomposing:
    \begin{mathpar}
      \inferrule*
      { \Psi_1, \Psi_3 \vdashtrans{\phi}{\vers} \config{T}[N] : B
        \\ \Psi_2 \vdashtrans{\wcirc}{\vers} \zap{\ep} : B }
      { \Psi_1, \Psi_2, \Psi_3 \vdashtrans{\phi}{\vers}
      \config{T}[N] \parallel \zap{\ep} : B }
    \end{mathpar}

    as required.
  \end{proofcase}

\end{proof}

\begin{fake}{Theorem~\ref{thm:combined:preservation} (Preservation (Configuration reduction))}
  If $\vdash \config{C}$ and $\config{C} \ceval \config{C}'$, then $\vdash \config{C}'$.
\end{fake}

\begin{proof}
  By induction on the derivation of $\config{C} \ceval \config{C}'$. We show the
  cases for \textsc{E-Update} and \textsc{E-Transition}.

  \begin{proofcase}{E-Update}
    \[
      \sys{\procctx[\handlerproctrans{\renderingext{V_m'}{V_c}{U}}{F}{\vers}]}{\vh}
                 \ceval
                 \sys{\procctx[\handlerproctrans{\idle{V'_m}}{F}{\vers} \parallel \threadtrans{M_1}{\vers} \parallel \cdots \parallel \threadtrans{M_n}{\vers}]}{\vh'}
    \]

                 $\text{where } \diff{U}{\vh} = \vh' \text{ and } \procs{V_c} =
                 \seq{M}$

                 Assumption:
     \begin{mathpar}
       \inferrule*
       {
           \cdot \vdashtrans{\bcirc}{\vers}
           \procctx[\handlerproctrans{\renderingext{V_m'}{V_c}{U}}{F}{\vers}] : B
         \\
         \vdash \vh : \pagety{B}
       }
       { \vdash
         \sys{\procctx[\handlerproctrans{\renderingext{V_m'}{V_c}{U}}{F}{\vers}]}{\vh}
       }
     \end{mathpar}

     By
     Lemmas~\ref{lem:combined:subprocess-typeability},~\ref{lem:combined:unr-psi-empty},
     and~\ref{cor:combined:unr-psi-empty-st}, there exist $\Psi_1, \Psi_2$ such that:
     \begin{mathpar}
       \inferrule*
       {
         \inferrule*
         {
           \Psi_1 \vdash V_m' : A \\\\
           \Psi_2 \vdash V_c : \cmdty{B} \\
           \cdot \vdash U : \htmlty{B}
         }
         { \Psi_1, \Psi_2 \vdashs \renderingext{V_m'}{V_c}{U} : \evtlooptytrans{A}{B}{C} } \\
           \cdot \vdash (V_v, V_u, V_e) : \statetytrans{A}{B}{C}
       }
       { \Psi_1, \Psi_2 \vdashtrans{\bcirc}{\vers}
         \handlerprocexptc{\renderingext{V_m'}{V_c}{U}}{V_v}{V_u}{V_e}{\vers} : B
       }
     \end{mathpar}

     Let $\deriv{D}$ be the following derivation:
  \begin{mathpar}
     \inferrule*
     {
       \inferrule*
       { \Psi_1 \vdash V_m' : A }
       { \Psi_1 \vdash \idle{V_m'} : \evtlooptytrans{A}{B}{C} } \\
       \cdot \vdash (V_v, V_u, V_e) : \statetytrans{A}{B}{C}
     }
     { \Psi_1 \vdashtrans{\bcirc}{\vers}
       \handlerprocexptc{\idle{V_m'}}{V_v}{V_u}{V_e}{\vers} : B
     }
  \end{mathpar}

  Recall that $\cdot \vdash U : \htmlty{B}$ and $\vdash \vh: \pagety{B}$, and
  $\diff{U}{D} = D'$.

  By Lemma~\ref{lem:combined:diffing}, we have that $\vdash D' : \pagety{B}$.

  By the definition of $\procs{V_c} = \seq{M}$, \textsc{T-Thread}, and
  \textsc{T-Par}, we have that $\Psi_2 \vdashtrans{\wcirc}{\vers}
\threadtrans{M_1}{\vers} \parallel \cdots \parallel \threadtrans{M_n}{\vers}$.

  Let $\deriv{D}'$ be the derivation:
  \begin{mathpar}
    \inferrule*
    { \deriv{D} \\ \Psi_2 \vdashtrans{\wcirc}{\vers} \threadtrans{M_1}{\vers} \parallel \cdots \parallel \threadtrans{M_n}{\vers} }
    { \Psi_1, \Psi_2 \vdashtrans{\bcirc}{\vers}
      \handlerprocexptc{\idle{V_m'}}{V_v}{V_u}{V_e}{\vers} \parallel \threadtrans{M_1}{\vers} \parallel \cdots \parallel
  \threadtrans{M_n}{\vers} : B }
  \end{mathpar}

  By Lemma~\ref{lem:combined:subprocess-replacement}, we have that:
  \[
    \cdot \vdashtrans{\bcirc}{\vers}
    \procctx[\handlerprocexptc{\idle{V_m'}}{V_v}{V_u}{V_e}{\vers}
      \parallel \threadtrans{M_1}{\vers} \parallel \cdots \parallel
          \threadtrans{M_n}{\vers}
    ] : B
  \]%

  Finally, we can show
  \begin{mathpar}
    \inferrule*
    {
      \deriv{D}' \\
      \vdash D' : \pagety{B}
    }
    { \vdash
      \sys{
      \procctx[\handlerprocexptc{\idle{V_m'}}{V_v}{V_u}{V_e}{\vers}
        \parallel \threadtrans{M_1}{\vers} \parallel \cdots \parallel
        \threadtrans{M_n}{\vers}]}{D'}
    }
  \end{mathpar}
  as required.
  \end{proofcase}

 \begin{proofcase}{E-Transition}
  \[
    \sys{\procctx[\handlerproctrans{\transitioningextexp{V_m}{V_v}{V_u}{V_e}{V_c}{U}}{F}{\vers}]}{\vh} \ceval
  \sys{\procctx[\handlerprocexptc{\idle{V_m}}{V_v}{V_u}{V_e}{\vers'} \parallel \threadtrans{M_1}{\vers'} \parallel
  \cdots \parallel \threadtrans{M_n}{\vers'}]}{\vh'}
  \]
  where
  $\vers' = \vers + 1$, $\diff{U}{\htmlempty} = D'$, and $\procs{V_c} = \seq{M}$.

   Assumption:
  \begin{mathpar}
      \inferrule*
      {
        \cdot
          \vdashtrans{\bcirc}{\vers}
          \procctx[\handlerproctrans{\transitioningextexp{V_m'}{V_v'}{V_u'}{V_e'}{V_c}{U}}{F}{\vers}] : B
        \\
        \vdash \vh : \pagety{B}
      }
      { \vdash
        \sys{\procctx[\handlerproctrans{\renderingext{V_m'}{V_c}{U}}{F}{\vers}]}{\vh}
      }
  \end{mathpar}

  By
  Lemmas~\ref{lem:combined:subprocess-typeability},~\ref{lem:combined:unr-psi-empty},
  and~\ref{cor:combined:unr-psi-empty-st}, there exist $\Psi_1, \Psi_2$ such that:

  \begin{mathpar}
    \inferrule*
    {
      \inferrule*
      {
        \Psi_1 \vdash V'_m : A' \\
        \cdot \vdash V'_v : C' \uto \htmlty{B'} \\\\
        \cdot \vdash V'_u : (B' \times A') \uto \transitionty{A'}{B'} \\
        \cdot \vdash V'_e : A' \uto (A' \times C') \\\\
        \Psi_2 \vdash V_c : \cmdty{B'} \\
        \cdot \vdash U : \htmlty{B'}
      }
      { \Psi_1, \Psi_2 \vdashs
      \transitioningextexp{V'_m}{V'_v}{V'_u}{V'_e}{V'_c}{U} : \evtlooptytrans{A}{B}{C} }
      \\
      \cdot \vdash F : \statetytrans{A}{B}{C}
    }
    { \Psi_1, \Psi_2 \vdashtrans{\bcirc}{\vers}
      \handlerproctrans{\transitioningextexp{V'_m}{V'_v}{V'_u}{V'_e}{V'_c}{U}}{F}{\vers} : B
    }
  \end{mathpar}

  We can show:
  \begin{mathpar}
    \inferrule*
    {
      \inferrule*
      { \Psi_1 \vdash V'_m : A' }
      { \Psi_1 \vdashs \idle{V'_m} : \evtlooptytrans{A'}{B'}{C'} } \\
      \inferrule*
        {
          \cdot \vdash V'_v : C' \uto \htmlty{B'} \\\\
          \cdot \vdash V'_u : (B' \times A') \uto \transitionty{A'}{B'} \\\\
          \cdot \vdash V'_e : A' \uto (A' \times C')
        }
        {
          \cdot \vdash (V'_v, V'_u, V'_e) : \statetytrans{A'}{B'}{C'}
        }
    }
    { \Psi_1 \vdashtrans{\bcirc}{\vers'}
        \handlerprocexptc{\idle{V'_m}}{V'_v}{V'_u}{V'_e}{\vers'} : B' }
  \end{mathpar}

  Since $\procs{V_c} = \seq{M}$, by \textsc{TP-Thread}, and \textsc{TP-Par}, we
  can show that $\Psi_2 \vdashtrans{\wcirc}{\vers'} \threadtrans{M_1}{\vers'}
  \parallel \cdots \parallel \threadtrans{M_n}{\vers'} : B'$.

  Recall that $\cdot \vdash U : \htmlty{B'}$, $\vdash D : \pagety{B}$, and $\diff{U}{\htmlempty} = D'$.
  Thus by Lemma~\ref{lem:combined:diffing}, we have that $\vdash D' :
  \pagety{B'}$.

  By Lemma~\ref{lem:combined:subprocess-replacement}, $\vdashtrans{\bcirc}{\vers'}
  \procctx[\handlerprocexptc{\idle{V'_m}}{V'_v}{V'_u}{V'_e}{\vers'} \parallel \threadtrans{M_1}{\vers'} \parallel
  \cdots \parallel \threadtrans{M_n}{\vers'}] : B' $

  Thus, we can show:

  \begin{mathpar}
    \inferrule*
    { \Psi
      \vdashtrans{\bcirc}{\vers'}
      \procctx[\handlerprocexptc{\idle{V'_m}}{V'_v}{V'_u}{V'_e}{\vers'}
      \parallel \threadtrans{M_1}{\vers'} \parallel \cdots \parallel
      \threadtrans{M_n}{\vers'}] : B' \\
      \vdash D' : \pagety{B'} }
    { \vdash
    \sys{\procctx[\handlerprocexptc{\idle{V'_m}}{V'_v}{V'_u}{V'_e}{\vers'}
    \parallel \threadtrans{M_1}{\vers'} \parallel \cdots \parallel
\threadtrans{M_n}{\vers}]}{D'} }
  \end{mathpar}

\end{proofcase}
\end{proof}

\subsection{Error-freedom.}
\begin{fake}{Theorem~\ref{thm:combined:rt-errors}}
  If $\Psi \vdashtrans{\phi}{\vers} P$, then $P$ is not an error
  process.
\end{fake}
\begin{proof}[Proof Sketch.]
  By \textsc{T-Nu}, endpoints $c$ and $d$ must be dual, and duality rules out
  all error processes. As an example, we cannot construct a derivation for (1).
  Suppose that $\config{T}$ and $\config{T'}$ are $\threadtrans{-}{\vers}$. Then:

 \begin{mathpar}
   \inferrule*
   {
     \inferrule*
     {
       \inferrule*
       { \Gamma_1, c : S \vdash E[\gvsend{V}{c}] : A }
       { \Gamma_1, c : S \vdashtrans{\wcirc}{\vers} \threadtrans{E[\gvsend{V}{c}]}{\vers} : A }
       \\
       \inferrule*
       { \Gamma_2, d : \gvdual{S} \vdash E'[\gvsend{W}{d}] : A }
       { \Gamma_2, d : \gvdual{S} \vdashtrans{\wcirc}{\vers} \threadtrans{E'[\gvsend{W}{d}]}{\vers} : A }
     }
     { \Gamma_1 + \Gamma_2, c : S, d : \gvdual{S} \vdashtrans{\wcirc}{\vers}
     \threadtrans{E[\gvsend{V}{c}]}{\vers} \parallel \threadtrans{E'[\gvsend{W}{d}]}{\vers} : A  }
   }
   { \Gamma_1 + \Gamma_2 \vdashtrans{\wcirc}{\vers} (\nu c d) (\threadtrans{E[\gvsend{V}{W}]}{\vers} \parallel
   \thread{E'[\gvsend{W}{d}])}{\vers} : A }
 \end{mathpar}

 To type $\Gamma_1, c : S \vdash E[\gvsend{V}{c}]$ and $\Gamma_2, d : \gvdual{S}
 \vdash E'[\gvsend{W}{d}]$, we would need $S = \gvout{A}{S'}$ and $\gvdual{S} =
 \gvout{B}{\gvdual{S'}}$, but this is impossible by the definition of duality. The
 remaining cases follow by the same reasoning.
\end{proof}

\subsection{Progress}

To prove progress, it is useful to define the notion of a canonical form, which
provides a global view of the configuration:

\begin{definition}[Canonical form]
  A process $P$ is in \emph{canonical form} if it has the form:
  \[
    (\nu c_1 d_1) \cdots (\nu c_n d_n)( P \parallel \threadtrans{M_1}{\vers_1} \parallel
      \cdots \parallel \threadtrans{M_m}{\vers_m}
      \parallel \serverthread{N_1} \parallel \cdots \parallel
      \serverthread{N_k} \parallel \zap{c'}_1 \parallel \cdots \parallel
      \zap{c'}_{l})
  \]
  where either $P = \run{M}$; $P = \handlerproctrans{T}{F}{\vers}$; or $P = \halt$.
\end{definition}

Any well-typed, closed process with a main thread can be written in canonical form.

\begin{lemma}[Canonical forms]\label{lem:combined:canonical-forms}
  If $\vdashtrans{\bcirc}{\vers} P : A$, then there exists some $P'$ such that
  $\vdashtrans{\bcirc}{\vers} P' : A$ where $P'$ is in canonical form.
\end{lemma}
\begin{proof}
  By case analysis on the derivation of $\cdot \vdashtrans{\bcirc}{\vers} P$, it
  must be the case that $P$ includes some subprocess $P'$ which is a main thread
  (either $P' = \run{M}$; $P' = \handlerproctrans{T}{F}{\vers}$; or $P' =
  \halt$).

  By Lemma~\ref{lem:combined:equiv-pres}, equivalence preserves typing.
  By use of the scope extrusion equivalence, we can lift all name restrictions
  to the head of the configuration. We can rearrange the remainder of the
  configurations using the associativity and commutativity equivalences.
\end{proof}

Functional reduction satisfies progress.

\begin{lemma}[Progress (Terms)]\label{lem:combined:term-progress}
  If $\Psi \vdash M : A$, then either $M$ is a value; or there exists some $M'$
  such that $M \teval M'$; or $M$ can be written $E[N]$ where $N$ is one of the
  following terms:
  \begin{itemize}
    \item $\gvnew{()}$
    \item $\gvclose{V}$
    \item $\gvcancel{V}$
    \item $\gvsend{V}{W}$
    \item $\gvrecv{V}$
    \item $\raiseexn$
  \end{itemize}
\end{lemma}
\begin{proof}
  A standard induction on the derivation of $\Psi \vdash M : A$.
\end{proof}

\begin{fake}{Theorem~\ref{thm:combined:weak-progress} (Weak Event Progress)}
  Suppose $\cdot \vdash \config{C}$. Either there exists some $\config{C'}$ such
  that $\config{C} \ceval \config{C}'$, or there exists some $\config{C'}$ such
  that $\config{C} \equiv \config{C}'$ and:

  \begin{enumerate}
    \item $D$ cannot be written $\config{D}[\pgtag{\tagname{t}}{V}{D}{\seq{e}}]$
      for a non-empty $\seq{e}$.
    \item If the main thread of $\config{C}'$ is $\halt$, then all auxiliary
      threads are either blocked or zapper threads.
    \item If the main thread of $\config{C}'$ is $\run{M}$, then $M$ is blocked,
      and all auxiliary threads are either blocked, values, or zapper threads.
    \item If the main thread of $\config{C}'$ is
      $\handlerproctrans{T}{F}{\vers}$, then:
      \begin{enumerate}
        \item if $T = \idle{V_m}$, then each auxiliary thread is either blocked or a
          zapper thread; or
        \item if $T = \ta[L]$ then $L$ is blocked, and each $M_i$ and $N_i$ is
          either blocked, a value, or a zapper thread.
      \end{enumerate}
  \end{enumerate}

\end{fake}
\begin{proof}

  Let $\config{C} = \sys{P'}{D}$.
  By Lemma~\ref{lem:combined:canonical-forms}, there exists some $P$ such that
  $P' \equiv P$ and $P$ is in canonical form:
  \[
    P = (\nu c_1 d_1) \cdots (\nu c_n d_n)( P_{\text{main}} \parallel \threadtrans{M_1}{\vers_1} \parallel
        \cdots \parallel \threadtrans{M_n}{\vers_n}
        \parallel \serverthread{N_1} \parallel \cdots \parallel
        \serverthread{N_m} \parallel \zap{c'}_1 \parallel \cdots \parallel \zap{c'}_{l}
  \]

  By the garbage collection equivalence $\serverthread{()} \parallel P \equiv
  P$ and inspection of the typing rule for server threads, we need not consider
  server threads which have evaluated to values.

  We proceed by case analysis on the derivation of $\cdot \vdash \sys{P}{D}$.

  If $D$ could be written $\config{D}[\pgtag{\tagname{t}}{V}{D}{\seq{e}}]$ for a
  non-empty $\seq{e}$, then the configuration could reduce by \textsc{E-Evt}.

  By Lemma~\ref{lem:combined:term-progress}, each $M_i$ and $N_i$ must either be
  a value or of the form
  $E[L]$ such that \\
  $L \in \{\gvnew{()}, \gvcancel{V}, \gvsend{V}{W}, \gvrecv{V}, \raiseexn  \}$.
  Of these, $\gvnew{()}$ can reduce by \textsc{EP-New}, and $\raiseexn$ can
  reduce by either \textsc{E-RaiseH}, \textsc{E-RaiseUThread}, or
  \textsc{E-RaiseUServer}.
  If $L = \gvcancel{V}$, as we consider closed configurations, by inversion on
  the typing relation, $V$ must be some runtime name $c$, which can reduce by
  \textsc{EP-Cancel}.  The remaining cases constitute a thread being blocked.
  Therefore, each $M_i$ and $N_i$ must be either a value or blocked; meaning
  each auxiliary thread must either be a value, blocked, or a zapper thread.

  We proceed by case analysis on $P_\text{main}$. If $P_\text{main} = \halt$, then
  all values will reduce by \textsc{E-DiscardHalt} to zapper threads, leaving
  only blocked terms and zapper threads, satisfying  (2).

  If $P_{\text{main}} = \run{M}$, then by
  Lemma~\ref{lem:combined:term-progress}, \textsc{LiftP}, and our previous
  reasoning, $M$ must be either a value or blocked. If $M$ is a value, then by
  \textsc{TP-Run} it is a 6-tuple and can thus reduce by \textsc{E-Run}.
  If it is blocked, then we satisfy (3).

  If $P_{\text{main}} = \handlerproctrans{T}{F}{\vers}$, then we proceed by case
  analysis on $T$.

  If $T = \idle{V_m}$, then if a thread were to be a value, then the
  configuration could reduce by \textsc{E-Handle}, so each thread must either be
  blocked or a zapper thread, satisfying (4a).
  For the remaining types of active thread, by previous reasoning the term being
  evaluated by the thread must be either blocked (satisfying (4b)), or a value;
  in the case that each term is a value, there exists another rule to enact a state
  transition, so the configuration could reduce.
\end{proof}

 \clearpage
\section{Implementation}\label{appendix:implementation}
We have implemented an MVU library in the Links tierless web programming
language. In this section, we give an overview of the implementation.

\subsection{Design}
The Links implementation of MVU does not require any compiler modifications,
other than for the syntactic sugar. A Links library provides the same monoidal
core constructs as found in \mvu, and uses Links' native concurrency to
implement an event loop. A JavaScript foreign function interface (FFI) allows us
to interact with a JavaScript library which efficiently propagates updates to
the DOM.

\subsection{Links Library}

\subparagraph{HTML and Attributes.}
The core \mvu HTML constructs are implemented as a Links variant type as follows:

\begin{minipage}{0.45\textwidth}
\begin{lstlisting}
typename HTML(a) =
  [| HTMLEmpty
   | HTMLAppend: (HTML(a), HTML(a))
   | HTMLText: (String)
   | HTMLTag: (tagName: String, attrs: Attr(a),
       children: HTML(a))
   |];
\end{lstlisting}
\end{minipage}
\hfill
\begin{minipage}{0.4\textwidth}
\begin{lstlisting}
typename Attr(a) =
  [| AttrEmpty
   | AttrAppend: (Attr(a), Attr(a))
   | AttrAttribute: (AttrKey, AttrValue)
   | AttrEventHandler: EventHandler(a) |];
\end{lstlisting}
\end{minipage}

The constructs map directly onto the constructs defined in \mvu, and are mostly
self-explanatory. Of note is \lstinline+AttrEventHandler+ constructor, which
defines an attribute containing an event handler.

\subparagraph{Event handlers.}
An \lstinline+EventHandler(a)+ defines an event handler producing messages of
type \lstinline+a+, and is defined as follows:
\begin{lstlisting}
typename EventHandler(a) =
  [| PropertyHandler:
      (EventName, PropertyName, (PropertyValue) {}~> Maybe(a))
   | UnitHandler: (EventName, () {}~> a)
   | MouseEventHandler: (EventName, (MouseEvent) {}~> Maybe(a))
   | KeyboardEventHandler: (EventName, (KeyboardEvent) {}~> Maybe(a))
   |];
\end{lstlisting}
A \lstinline+PropertyHandler+ is an event handler which is attached to a DOM
element. When the event \lstinline+EventName+ is fired, then the handler
function receives the value \lstinline+PropertyValue+ of the property name
defined by \lstinline+PropertyName+. The handler function returns
\lstinline+Nothing+ if no message is to be produced, or \lstinline+Just(msg)+
to produce some message \lstinline+msg+ of type \lstinline+a+. Note that the
\lstinline+{}~>+ notation means that the function may not perform any effects.
As an example of a \lstinline+PropertyHandler+, we can define the
\lstinline+onInput+ handler as follows:
\begin{lstlisting}
fun onInput(f) {
  AttrEventHandler(PropertyHandler("input", "value", fun(val) { Just(f(val)) }))
}
\end{lstlisting}
The \lstinline+onInput+ handler is triggered on an \lstinline+input+ event,
and passes \lstinline+value+ property of its associated DOM element to the
handler function \lstinline+f+.

A \lstinline+UnitHandler+ handler is useful for side-effecting operations such
as button clicks, and does not pass any properties to the event handler function. The
\lstinline+MouseEventHandler+ and \lstinline+KeyboardEventHandler+ event
handlers handle mouse and keyboard events respectively, where
\lstinline+MouseEvent+ and \lstinline+KeyboardEvent+ are Links types which
encode the respective DOM properties.

\subparagraph{Desugaring.}
As well as creating tags and attributes by constructing values of
\lstinline+HTML(a)+ and \lstinline+Attr(a)+ directly, Links directly supports XML style
syntax. The XML syntax is desugared down into
the \lstinline+HTML(a)+ and \lstinline+Attr(a)+ types via a source-to-source
translation in a desugaring pass.

\subparagraph{Event loop.}

Links natively supports message-passing concurrency and a JavaScript foreign
function interface. This allows us to implement an MVU event loop:

\begin{minipage}{0.5\textwidth}
\begin{lstlisting}
module VDom {
  alien javascript "/lib/vdom.js" {
    runDom :
      (String, HTML(a), AP(?a.End)) ~> ();
    updateDom : (HTML(a)) ~> ();
  }
}
\end{lstlisting}
\end{minipage}
\hfill
\begin{minipage}{0.45\textwidth}
\begin{lstlisting}
sig evtLoop:
  (AP(?msg.End), model,
   (model) ~> HTML(msg),
   (msg, model) ~> model) ~> ()
fun evtLoop(ap, model, view, updt) {
  var (message, s) = receive(accept(ap));
  close(s);
  var model = updt(message, model);
  VDom.updateDom(view(model));
  evtLoop(ap, model, view, updt)
}
\end{lstlisting}
\end{minipage}

The \lstinline+VDom+ module contains an \lstinline+alien+ block, which contains
FFI bindings. The \lstinline+runDom+ function initialises the JavaScript
library, taking a string which is the ID of a placeholder element to replace
with the MVU document; the initial document \lstinline+HTML(a)+;
and an \emph{access point} \lstinline+AP(?a.End)+ which allows communication
with an event loop. Access points are a method of establishing a session channel
between two parties: communication with the event loop is done via establishment
of a one-shot channel.

The event loop \lstinline+evtLoop+ takes four parameters: the access point which
it uses to receive messages; the current \lstinline+model+; the \lstinline+view+ function; and the
update function \lstinline+updt+. The function begins by accepting a connection on the access
point, receiving a message, and then closing the connection. It then calls the
\lstinline+VDom.updateDom+ FFI function to update the web page, and loops with
the new model.

\subparagraph{JavaScript library.}
Given a Links \lstinline+HTML(a)+ value, the library efficiently propagates
changes to the DOM. The JavaScript library interprets the \lstinline+HTML(a)+ and
\lstinline+Attr(a)+ types, generating a representation which can then be used
with the open-source \lstinline+virtual-dom+ library~\cite{virtual-dom}.

The \lstinline+runDom+ function and \lstinline+updateDom+ FFI functions called
from the Links event loop are implemented as follows:

\begin{minipage}{0.45\textwidth}
\begin{lstlisting}[language=JavaScript]
function _runDom(id, doc, ap) {
  evtHandlerAP = ap;
  currentVDom = evalToplevelHTML(doc);
  rootNode = createElement(currentVDom);
  document.getElementById(id)
    .appendChild(rootNode);
}
\end{lstlisting}
\end{minipage}
\hfill
\begin{minipage}{0.45\textwidth}
\begin{lstlisting}[language=JavaScript]
function _updateDom(doc) {
  var newTree = evalToplevelHTML(doc);
  var patches = diff(currentVDom, newTree);
  currentVDom = newTree;
  rootNode = patch(rootNode, patches);
}
\end{lstlisting}
\end{minipage}

The \lstinline+evtHandlerAP+, \lstinline+currentVDom+, and \lstinline+rootNode+
global variables keep track of the event handler access point; the current
virtual representation of the DOM; and the root VDOM node respectively. The
\lstinline+runDom+ function sets the initial value of the global variables,
interprets the Links HTML representation (\lstinline+doc+), and adds the root
element to the page. The \lstinline+updateDom+ function interprets the HTML,
diffs it against \lstinline+currentVDom+, sets \lstinline+currentVDom+ to be
the newly-calculated HTML, and then propagates the updates to the physical DOM
using the \lstinline+patch+ function.

\end{adjustwidth}

\end{document}